\documentclass{article}

\usepackage[left=1.25in,top=1.25in,right=1.25in,bottom=1.25in,head=1.25in]{geometry}
\usepackage{amsfonts,amsmath,amssymb,amsthm}
\usepackage{mathtools}
\usepackage{verbatim,float,url,bbm}
\usepackage{psfrag,color}
\usepackage{natbib}
\usepackage{floatpag}
\usepackage{smileryan}
\usepackage{url}
\usepackage{algorithm}
\floatpagestyle{empty}

\theoremstyle{plain}
\newtheorem{theorem}{Theorem}
\newtheorem{lemma}[theorem]{Lemma}
\newtheorem{corollary}[theorem]{Corollary}
\theoremstyle{definition}
\newtheorem{remark}{Remark}
\theoremstyle{definition}

\usepackage[usenames,dvipsnames,svgnames,table]{xcolor}
\usepackage[colorlinks=true, linkcolor=blue,
urlcolor=blue,citecolor=blue]{hyperref}

\newcommand{\st}{\mathop{\mathrm{subject\,\,to}}}

\def\Proj{\mathbb{\Pi}}
\def\E{\mathbb{E}}
\def\P{\mathbb{P}}

\def\half{\frac{1}{2}}

\def\sign{\mathrm{sign}}

\def\spa{\mathrm{span}}

\def\hx{\hat{x}}

\def\hw{\hat{w}}
\def\hz{\hat{z}}
\def\hS{\hat{S}}
\def\htheta{\hat{\theta}}
\def\ttheta{\tilde{\theta}}

\def\htheta{\hat{\theta}}

\def\R{\mathbb{R}}
\def\cA{\mathcal{A}}

\def\cM{\mathcal{M}}

\def\cR{\mathcal{R}}

\def\cV{\mathcal{V}}

\def\one{\mathbbm{1}}
\def\TV{\mathrm{TV}}

\def\Proj{P}

\def\op{D}

\def\erf{\mathrm{erf}}

\definecolor{pr}{rgb}{0.4,0,0.65} 

\title{Approximate Recovery in Changepoint Problems, from $\ell_2$
  Estimation Error Rates}
\author{Kevin Lin$^1$ \and James Sharpnack$^2$ \and Alessandro
  Rinaldo$^1$ \and Ryan J. Tibshirani$^1$}
\date{$^1$Carnegie Mellon University, 
$^2$University of California at Davis} 

\begin{document}
\maketitle

\begin{abstract}
In the 1-dimensional multiple changepoint detection problem, we prove  
that any procedure with a fast enough $\ell_2$ error rate, in terms of
its estimation of the underlying piecewise constant mean vector,
automatically has an (approximate) changepoint screening  
property---specifically, each true jump in the underlying mean
vector has an estimated jump nearby. We also show, again assuming only
knowledge of the $\ell_2$ error rate, that a simple post-processing
step can be used to eliminate spurious estimated changepoints, and
thus delivers an (approximate) changepoint recovery
property---specifically, in addition to the screening property
described above, we are assured that each estimated jump has a
true jump nearby. As a special case, we focus on the application of
these results to the 1-dimensional fused lasso, i.e., 1-dimensional  
total variation denoising, and compare the implications with 
existing results from the literature.  We also study extensions
to related problems, such as changepoint detection over graphs.  

\end{abstract}

\section{Introduction} 
\label{sec:intro}

Consider the 1-dimensional multiple changepoint model
\begin{equation}
\label{eq:model}
y_i = \theta_{0,i} + \epsilon_i, \quad i=1,\ldots,n,
\end{equation}
where $\epsilon_i$, $i=1,\ldots,n$ are i.i.d.\ errors, and
$\theta_{0,i}$, $i=1,\ldots,n$ is a piecewise constant mean sequence, 
having a set of changepoints
\begin{equation*}
S_0 = \big\{ i \in \{1,\ldots,n-1\} : \theta_{0,i} \not=
\theta_{0,i+1} \big\}. 
\end{equation*}
This is a well-studied problem, and there is a large body of 
literature on estimation of the piecewise constant mean
vector   
$\theta_0 \in \R^n$ in this model, as well as detection of its
changepoints $\theta_0$.  Though estimation (of $\theta_0$) and 
detection (of its changepoints) are clearly related pursuits, they are
different enough that most works on the changepoint problem are
focused on one or the other. For example, 1-dimensional total
variation denoising 
\citep{rudin1992nonlinear}, i.e., the 1-dimensional fused
lasso \citep{tibshirani2005sparsity}, has been
primarily studied from the perspective of its estimation
properties. Meanwhile, segmentation methods like binary
segmentation
\citep{vostrikova1981detecting,venkatraman1992consistency} 
and wild binary segmentation 
\citep{fryzlewicz2014wild} have been mostly studied for their
detection properties.  In this paper, we assert that the estimation
and detection problems are very closely linked, in the following
sense: 
any estimator with $\ell_2$ estimation error guarantees automatically
has certain approximate changepoint detection guarantees, and
not surprisingly, a faster $\ell_2$ estimation error rate 
here translates into a stronger statement about approximate detection.
We use this general link to establish new approximate changepoint
recovery results for the 1d fused lasso, an estimator that is 
given central focus in our work.

\subsection{Background and related work}

Given a data vector $y \in \R^n$ from a model as in \eqref{eq:model}, 
the {\it 1-dimensional fused lasso} (1d fused lasso, or simply fused
lasso) estimate is defined by
\begin{equation}
\label{eq:fl}
\htheta = \argmin_{\theta \in \R^n} \; 
\half \sum_{i=1}^n (y_i-\theta_i)^2 + 
\lambda \sum_{i=1}^{n-1} |\theta_i-\theta_{i+1}|, 
\end{equation}
where $\lambda \geq 0$ serves as a tuning parameter.  This was
proposed by \citet{tibshirani2005sparsity}\footnote{In 
  \citet{tibshirani2005sparsity}, the authors actually used an
  additional $\ell_1$ penalty on $\theta$ itself, to induce
  componentwise sparsity; here we do not consider this extension, and
  simply refer to the estimator in \eqref{eq:fl} as the fused lasso.}, 
though the same idea had been proposed in signal processing much
earlier, under the name {\it total variation (TV) denoising}, by  
\citet{rudin1992nonlinear}.  There has been plenty of statistical
theory developed for the fused lasso, e.g., 
\citet{mammen1997locally,davies2001local,rinaldo2009properties,
harchaoui2010multiple,qian2012pattern,
rojas2014change,dalalyan2014prediction}.  In particular,
\citet{mammen1997locally} and  
\citet{dalalyan2014prediction} derived $\ell_2$ error  
rates for the fused lasso, under different settings (different
assumptions on $\theta_0$).
We review these in \Fref{sec:review_estimation}, and 
compare the latter to a new related result that we establish in    
\Fref{sec:strong_james}. 
\citet{harchaoui2010multiple} and
\citet{rojas2014change} derived approximate
changepoint recovery properties for the fused lasso. We review these
in \Fref{sec:review_recovery}, and 
compare them to our own results on approximate recovery in 
\Fref{sec:screening_recovery}.         

The literature on general approaches in multiple changepoint
detection is enormous, and we do not give an extensive overview,
but we do summarize some relevant work of \citet{donoho1994ideal,  
donoho1998minimax,fryzlewicz2007unbalanced,boysen2009consistencies,
fryzlewicz2014wild, frick2014multiscale,fryzlewicz2016tail} in
Sections \ref{sec:review_estimation} and \ref{sec:review_recovery},
and revisit these results in more detail in 
\Fref{sec:screening_recovery} when we compare them to our new
approximate changepoint recovery results.
Extensions of the fused lasso such   
as trend filtering and the graph-based fused lasso have been analyzed 
by, e.g, \citet{sharpnack2012sparsistency,tibshirani2014adaptive, 
wang2016trend,hutter2016optimal},
which are discussed later in \Fref{sec:extensions}. 

\subsection{Notation}
\label{sec:notation}

For a vector $\theta \in \R^n$, we write $S(\theta)$ for the
set of its changepoint indices, i.e.,
\begin{equation*}
S(\theta) =  \big\{ i \in \{1,\ldots,n-1\} : \theta_i \not= 
\theta_{i+1} \big\}.
\end{equation*}
We abbreviate $S_0 = S(\theta_0)$ and \smash{$\hS = S(\htheta)$} for
the changepoints of the mean $\theta_0$ in \eqref{eq:model},
and the fused lasso estimate \smash{$\htheta$} in \eqref{eq:fl},
respectively.  Throughout, we will use the words ``changepoint''
and ``jump'' interchangeably. We will also make use of the following 
quantities defined in terms of $\theta_0$.  The size
of $S_0$ is denoted $s_0 = |S_0|$.  For convenience, we write
\smash{$S_0 = \{t_1,\ldots,t_{s_0} \}$}, where 
\smash{$1 \leq t_1 < \ldots< t_{s_0} < n$}, and 
by convention, \smash{$t_0=0$}, \smash{$t_{s_0+1}=n$}. 
The smallest distance  
between jumps in $\theta_0$ is denoted by  
\begin{equation}
\label{eq:Wn}
W_n = \min_{i=0,1\ldots,s_0} \, (t_{i+1} - t_i), 
\end{equation}
and the smallest distance between consecutive levels of $\theta_0$ by   
\begin{equation}
\label{eq:Hn}
H_n= \min_{i \in S_0} \; |\theta_{0,i+1} - \theta_{0,i}|.
\end{equation}
Our notation here makes the dependence of $W_n,H_n$ on $n$
explicit (of course stemming from the fact that the
mean vector $\theta_0$ itself changes with $n$, though for simplicity
we suppress this notationally.)

For a matrix $D \in \R^{m\times n}$, we write $D_S$ to extract 
rows of $D$ indexed by a subset $S \subseteq \{1,\ldots,m\}$,
and $D_{-S}$ as shorthand for \smash{$D_{-S}$}, where $-S =
\{1,\ldots,m\} \setminus S$.  Unless otherwise specified, the notation
$D \in \R^{(n-1)\times n}$ will be used to denote the difference 
operator
\begin{equation}
\label{eq:diff}
D = \left[\begin{array}{rrrrr}
-1 & 1 & 0 & \ldots &  0\\
0 & -1 & 1 & \ldots &  0\\
\vdots &  & \ddots & \ddots & \\ 
0 & 0 & \ldots & -1 & 1 
\end{array}\right].
\end{equation}
For a vector $x \in \R^n$, we
define its scaled $\ell_2$ norm \smash{$\|x\|_n=
  \|x\|_2/\sqrt{n}$}, and its discrete total variation   
\begin{equation} \label{eq:tv}
\TV(x) = \sum_{i=1}^{n-1} |x_i-x_{i+1}| = \|Dx\|_1.
\end{equation}
For two discrete sets $A,B$, we define the metrics 
\begin{equation}
\label{eq:metrics}
d(A|B) = \max_{b\in B} \, \min_{a \in A} \, |a-b|
\quad \text{and} \quad
d_H(A,B) = \max\big\{ d(A|B), d(B|A) \}. 
\end{equation}
The former metric can seen as a one-sided screening
distance from $B$ to $A$, measuring the furthest distance of an
element in $B$ to its closest element in $A$. The latter metric is 
traditionally known as the Hausdorff distance between $A$ and $B$.
Note that if $A$ is empty, then we have \smash{$d(A|B)=0$}, and if 
$B$ is empty, then \smash{$d(A|B)=\infty$}; this makes
$d_H(A,B)=\infty$ if either $A$ or $B$ are empty.

For deterministic sequences $a_n,b_n$ we write $a_n = O(b_n)$ 
to denote that $a_n/b_n$ is bounded for large enough $n$,
$a_n=\Omega(b_n)$ to denote that $b_n/a_n$ is bounded for large enough
$n$, and
$a_n=\Theta(b_n)$ to denote that both $a_n=O(b_n)$ and
$a_n=\Omega(b_n)$.  We also write $a_n=o(b_n)$ to denote that  
$a_n/b_n \to 0$, and $a_n=\omega(b_n)$ to denote that 
$b_n/a_n \to 0$.
Finally, we write $A_n = O_\P(B_n)$ for random sequences
$A_n,B_n$ to denote that $A_n/B_n$ is bounded in probability, and
$A_n=o_\P(B_n)$ to denote that $A_n/B_n \to 0$ in probability.

\subsection{Summary of results}

A summary of our contributions is as follows.  

\begin{itemize}
\item \textbf{New $\ell_2$ error analysis for the fused lasso,  
    under strong sparsity.} In \Fref{sec:strong_james}, we give a  
  new $\ell_2$ estimation error analysis for the fused lasso, in the 
  case $s_0 = O(1)$, which we refer to as the ``strong sparsity''
  case. \Fref{thm:strong_james} provides the bound 
  \begin{equation*}
    \|\htheta-\theta_0\|_n^2 = 
    O_\P \bigg(\frac{\log{n}\log\log{n}}{n}\bigg),
  \end{equation*}
  for the fused lasso estimate \smash{$\htheta$} in
  \eqref{eq:fl}. This is sharper than the previously established error
  rate of $\log^2{n}/n$, from \citet{dalalyan2014prediction}, for the
  fused lasso under strong sparsity, and quite close to the
  ``oracle'' rate of $\log{n}/n$ under strong sparsity, as we discuss
  in \Fref{rem:others_error_strong}. Our theorem also applies beyond
  the case of a constant sparsity level $s_0$, and gives an explicit
  error bound in 
  terms of $s_0$. We believe that the proof of \Fref{thm:strong_james}
  is interesting in its own right, as it leverages a new quantity that
  we call a {\it lower interpolant} to approximate the fused lasso 
  estimate in a certain sense 
  using $2s_0+2$ piecewise monotonic segments,
  which allows for finer control of the sub-Gaussian complexity.


\item \textbf{Bound on the screening distance, based on $\ell_2$ 
    error.}  In \Fref{sec:changepoint_screening}, we derive a bound
  on the sreening distance from $S_0$ to the detected changepoints 
\smash{$S(\ttheta)$} of any estimator \smash{$\ttheta$}, given a 
bound on its $\ell_2$ error rate 
\smash{$\|\ttheta-\theta_0\|_n^2=O_\P(R_n)$}.  Specifically, in 
\Fref{thm:changepoint_screening}, we show that 
\begin{equation*}
d\big(S(\ttheta)\,|\,S_0 \big) = O_\P 
\bigg( \frac{n R_n}{H_n^2} \bigg).
\end{equation*}
To emphasize, this bound on the screening distance is agnostic about
the details of the estimator \smash{$\ttheta$}, provided that its
$\ell_2$ error rate $R_n$ is known. As two principal applications, we
plug in   
the known error rate $R_n$ for the fused lasso under two different 
settings---weak and strong 
sparsity---to derive new screening results on the fused lasso in
Corollaries \ref{cor:changepoint_screening_weak} and 
\ref{cor:changepoint_screening_strong}.  Perhaps surprisingly (since
these screening bounds are not based on fine-grained analysis of the
fused lasso, but on achieved $\ell_2$ rates alone), these results
provide interesting conclusions in each of their own settings, as
we discuss in Remarks \ref{rem:screening_weak} and
\ref{rem:screening_strong}. 


\item \textbf{Bound on the Hausdorff distance, based on $\ell_2$ error  
    and a post-processing step.} In \Fref{sec:changepoint_recovery},
  we give a bound on the Hausdorff distance between $S_0$ and the 
  detected changepoints \smash{$S(\ttheta)$} of any estimator 
  \smash{$\ttheta$}, given a  
  bound \smash{$\|\ttheta-\theta_0\|_n^2=O_\P(R_n)$} and a simple 
filtering-based technique to remove spurious changepoints in
\smash{$S(\ttheta)$} that occur far away from elements of $S_0$.  
In particular, \Fref{thm:changepoint_recovery} states that the
filtered set \smash{$S_F(\ttheta)$} of changepoints satisfies
\begin{equation*}
\P\bigg( d_H\big(S_F(\ttheta), S_0 \big) \leq \frac{nR_n \nu_n}{H_n^2}
\bigg) \to 1 \quad \text{as $n \to \infty$}, 
\end{equation*}
where $\nu_n$ is any diverging sequence (i.e., diverging as
slowly as desirable).  As two applications, we consider
post-processing the changepoints from the fused lasso estimator, in
the weak and strong sparsity settings, in Corollaries
\ref{cor:changepoint_recovery_weak} and
\ref{cor:changepoint_recovery_strong}.  We compare these to
existing approximate changepoint recovery results in the literature in
Remarks \ref{rem:recovery_weak} and \ref{rem:recovery_strong}; the 
summary is that under strong sparsity, our result on the
post-processed fused lasso is comparable with the best known recovery
results, but under weak sparsity, our result is worse than the
guarantees given in
\citet{frick2014multiscale,fryzlewicz2014wild,fryzlewicz2016tail} for
other changepoint estimators. It should be reiterated that, unlike
other results in the literature which are based on detailed analyses
of specific changepoint estimators, our results are generic and based
only on $\ell_2$ error properties, making them widely applicable.
Therefore, a lack in sharpness in some cases, such as the weak
sparsity case, is perhaps not unexpected. 

\item \textbf{Practical guidelines for post-processing.} In 
  \Fref{sec:changepoint_recovery_reduced}, we present a modification
  of the aforementioned post-processing rule, which guarantees that
  the filtered set has at most \smash{$3\tilde{s}+2$}
  elements, where \smash{$\tilde{s}=|S(\ttheta)|$}. 
  In \Fref{sec:implementation}, we describe a data-driven procedure   
  to determine an appropriate threshold level
  for the filter, and we also conduct detailed empirical
  investigations of our proposals.

\item \textbf{Extension to piecewise linear segmentation, and  
    graph changepoint detection.} In 
  \Fref{sec:extensions}, we give extensions of our screening 
  results to two related settings: piecewise linear segmentation and 
  changepoint detection over graphs. For piecewise
  linear segmentation, the main screening result is in 
  \Fref{thm:knot_screening}, and its specialization to the trend
  filtering estimator is in \Fref{cor:knot_screening_weak}; for graph 
  changepoint detection, the main result is in
  \Fref{thm:graph_screening}, and its specialization to the 2d fused  
  lasso estimator is in \Fref{cor:graph_screening_2d}. 
\end{itemize}

\section{Preliminary review of existing theory} 

We review existing statistical theory for the fused lasso, first on
$\ell_2$ estimation error, and then on (approximate) changepoint 
recovery. 

\subsection{Review: $\ell_2$ estimation error} 
\label{sec:review_estimation}

We begin by describing two major results on the quantity
\begin{equation*}
\|\htheta-\theta_0\|_n^2,
\end{equation*}
the squared $\ell_2$ estimation error between the fused lasso estimate
\smash{$\htheta$} in \eqref{eq:fl} and the mean $\theta_0$ in
\eqref{eq:model}.  In somewhat of an abuse of notation, we will simply  
refer to the above quantity as the $\ell_2$ estimation error, or
estimation error for short.

The first result, from \citet{mammen1997locally}, studies what may be
called the ``weak sparsity'' case, in
which the total variation of $\theta_0$ is controlled.  Before stating
this, we recall that a random variable $Z$ is said to have a mean zero 
sub-Gaussian distribution provided that 
\begin{equation}
\label{eq:subgauss}
\E(Z)=0 \quad \text{and} \quad \P(|Z|>t) \leq M
\exp\big(-t^2/(2\sigma^2)\big) \quad \text{for $t \geq 0$},
\end{equation}
for some constants $M,\sigma>0$.  


\begin{theorem}[\textbf{Fused lasso error rate, weak sparsity setting,      
  Theorem 10 of \citealt{mammen1997locally}}]     
\label{thm:weak_sparsity}
Assume the data model in \eqref{eq:model}, with errors
$\epsilon_i$, $i=1,\ldots,n$ i.i.d.\ from a sub-Gaussian
distribution as in \eqref{eq:subgauss}.  Also assume that 
\smash{$\TV(\theta_0) \leq C_n$}, for a nondecreasing
sequence $C_n$. Then for a choice of tuning parameter
\smash{$\lambda=\Theta(n^{1/3}C_n^{-1/3})$}, the fused lasso estimate    
\smash{$\htheta$} in \eqref{eq:fl} satisfies
\begin{equation*}
\|\htheta - \theta_0\|_n^2 = O_\P (n^{-2/3} C_n^{2/3}). 
\end{equation*}
\end{theorem}

\begin{remark}[\textbf{Consistency, optimality}]
This shows that the fused lasso estimator is consistent when
$C_n=o(n)$. When $C_n=O(1)$, its estimation error rate is
$n^{-2/3}$, which is in fact the minimax optimal rate as $\theta_0$
varies over the class of signals with bounded total variation, i.e., 
$\theta_0 \in \{\theta \in \R^n: \TV(\theta) \leq C\}$ for a constant
$C>0$ \citep{donoho1998minimax}.  For explanations of  
the above theorem and this minimax result, in
notation that is more consistent with that of the current paper, see 
\citet{tibshirani2014adaptive}.   
\end{remark}

The second result, from \citet{dalalyan2014prediction}, studies what
may be called the ``strong sparsity'' case, in which the number of
changepoints $s_0$ in $\theta_0$ is controlled.

\begin{theorem}[\textbf{Fused lasso error rate, strong sparsity
    setting, Proposition 4 of \citealt{dalalyan2014prediction}}]  
\label{thm:strong_sparsity}
Assume the data model in \eqref{eq:model}, with errors
$\epsilon_i$, $i=1,\ldots,n$ drawn i.i.d.\ from $N(0,\sigma^2)$. 
Then for a choice of 
tuning parameter \smash{$\lambda =\sqrt{ 2n \log(n/\delta)}$}, 
the fused lasso estimate \smash{$\htheta$} in \eqref{eq:fl} satisfies
\begin{equation*}
\|\htheta - \theta_0\|_n^2 \leq c 
\frac{ s_0 \log(n/\delta) }{n} 
\bigg( \log{n} + \frac{n}{W_n} \bigg),
\end{equation*}
with probability at least $1-2\delta$, for all $\delta>0$ and all $n
\geq N$, where $c,N>0$ are constants, and recall $W_n$ is the minimum
distance between jumps in $\theta_0$, as in \eqref{eq:Wn}. 
\end{theorem}

\begin{remark}[\textbf{The roles of $s_0,W_n$}]
When the number of jumps $s_0$ in $\theta_0$ grows quickly enough with 
$n$, the error rate in \Fref{thm:strong_sparsity} will become
worse than that in \Fref{thm:weak_sparsity}.  Given $s_0$ jumps,
in the best case, the minimum gap $W_n$ between jumps scales as 
$W_n=\Theta(n/s_0)$, which delivers a rate of 
\smash{$s_0 \log^2{n}/n + s_0^2 \log{n}/n$} in 
\Fref{thm:strong_sparsity}.  When $s_0$ scales faster than 
\smash{$n^{1/6}(\log{n})^{-1/2}$}, we can see that this is 
slower than the \smash{$n^{-2/3}$} rate delivered by
\Fref{thm:weak_sparsity} (assuming $C_n=O(1)$).  

Of course, \Fref{thm:strong_sparsity} is most useful when $s_0=O(1)$.
When this is true, and additionally $W_n=\Theta(n)$, we see that the
theorem implies that the fused lasso has $\ell_2$ error
\smash{$\|\htheta - \theta_0\|_n^2  =  O_\P(\log^2{n}/n)$}. This is a
very fast rate, nearly equal to the ``parametric rate'' of $1/n$ 
associated with estimating a finite-dimensional parameter.  
\end{remark}

\begin{remark}[\textbf{Alternative fused lasso error rate, strong
    sparsity setting}]
\label{rem:harchaoui_error}
When $s_0=O(1)$,
Proposition 2 in \citet{harchaoui2010multiple} proves that
the fused lasso has estimation error 
\smash{$\|\htheta - \theta_0\|_n^2  = O(\log{n}/n)$} with probability
approaching 1, under a choice
\smash{$\lambda=\Theta(\sqrt{\log{n}/n^3})$}. But the authors
must also assume that the number of changepoints in the fused lasso 
estimate \smash{$\htheta$}, which we might denote as
\smash{$\hat{s}=|\hS|$}, is bounded with probability tending to 1.  
This seems to be an unrealistic assumption, given the 
required scaling for $\lambda$.  Theoretically, we remark that such a 
small choice of $\lambda$, on the order of
\smash{$\sqrt{\log{n}/n^3}$}, does not match the much larger choices
dictated by Theorems \ref{thm:strong_sparsity} and
\ref{thm:strong_james}, both on the order of approximately 
\smash{$\sqrt{n}$}. Empirically, when
$\lambda$ scales as \smash{$\sqrt{\log{n}/n^3}$}, we find that the
number of estimated changepoints in \smash{$\htheta$} often grows very
large, even when $\theta_0$ has few jumps and the signal-to-noise
ratio is quite high.   
\end{remark}

\begin{remark}[\textbf{Comparable error rates of other estimators}] 
\label{rem:others_error}
Various other estimators obtain comparable estimation error rates to 
those descibed above for the fused lasso. The Potts estimator, defined
by replacing the $\ell_1$ penalty \smash{$\sum_{i=1}^{n-1}
  |\theta_i-\theta_{i+1}|$} in \eqref{eq:fl} with the
$\ell_0$ penalty \smash{$\sum_{i=1}^{n-1} 1\{\theta_i \not=
  \theta_{i+1}\}$}, and denoted say by 
\smash{$\htheta^{\mathrm{Potts}}$}, has been shown to satisfy 
\smash{$\|\htheta^{\mathrm{Potts}} - 
\theta_0\|_n^2 = O((\log{n}/n)^{2/3})$} a.s.\ when $\TV(\theta_0) =
O(1)$, and \smash{$\|\htheta^{\mathrm{Potts}} - 
\theta_0\|_n^2 = O(\log{n}/n)$} a.s.\ when $s_0=O(1)$,
by \citet{boysen2009consistencies}.  Wavelet 
denoising (under weak conditions on the wavelet 
basis), denoted by \smash{$\htheta^{\mathrm{wav}}$}, has been shown    
to satisfy \smash{$\E\|\htheta^{\mathrm{wav}} - 
\theta_0\|_n^2 = O(n^{-2/3})$} 
when $\TV(\theta_0) = O(1)$, 
by \citet{donoho1998minimax}, and 
\smash{$\E\|\htheta^{\mathrm{wav}} -  
\theta_0\|_n^2 = O(\log^2{n}/n)$} when $s_0=O(1)$,
by \citet{donoho1994ideal}. 
Combining unbalanced Haar (UH) wavelets  
with a basis selection method, \citet{fryzlewicz2007unbalanced}  
gave an estimator \smash{$\htheta^{\mathrm{UH}}$} with
\smash{$\E\|\htheta^{\mathrm{UH}} - 
\theta_0\|_n^2 = O(\log^2{n}/n^{2/3})$} when $\TV(\theta_0) = O(1)$,  
and \smash{$\E\|\htheta^{\mathrm{UH}} -  
\theta_0\|_n^2 = O(\log^2{n}/n)$} when $s_0=O(1)$.  Though they are
not written in this form, the results in
\citet{fryzlewicz2016tail} imply that his ``tail-greedy'' unbalanced
Haar (TGUH) estimator, \smash{$\htheta^{\mathrm{TGUH}}$}, satisfies 
\smash{$\|\htheta^{\mathrm{TGUH}} -  
\theta_0\|_n^2 = O(\log^2{n}/n)$} with probability tending to 1, 
when $s_0=O(1)$.  
\end{remark}

\subsection{Review: changepoint recovery}
\label{sec:review_recovery}

Next, we review the relevant results on the quantities
\[
d(\hS \,|\, S_0\big) \quad \text{or} \quad d_H(\hS, S_0).
\]
The former is the screening distance  from $S_0$ to the set of
changepoints \smash{$\hS=S(\htheta)$} in the fused lasso estimate
\smash{$\htheta$} in \eqref{eq:fl}; the latter is the Hausdorff
distance between $S_0$ and \smash{$\hS$}; recall, both metrics were 
defined in \eqref{eq:metrics}. We use the term ``approximate 
screening'' to mean that   
\smash{$d(\hS \,|\,  S_0)$} is controlled, and ``approximate
recovery'' to mean that \smash{$d_H(\hS, S_0)$} is controlled, though
often times we will drop the word ``approximate'' from either term,
for brevity. Below we summarize two results from
\citet{harchaoui2010multiple}. 

\begin{theorem}[\textbf{Fused lasso approximate screening and 
  recovery results, strong sparsity setting, 
  Propositions 3 and 4 of \citealt{harchaoui2010multiple}}] 
\label{thm:harchaoui_recovery}
Assume the data model in \eqref{eq:model}, where the errors
$\epsilon_i$, $i=1,\ldots,n$ are i.i.d. from a sub-Gaussian 
distribution as in \eqref{eq:subgauss}. Assume also that
(i) $s_0 = O(1)$, (ii) \smash{$\P( \hat{s} \geq s_0) \to 1$}, and
that $r_n$ is a sequence satisfying
(iii) $r_n \leq W_n$, (iv)
\smash{$r_n=\omega(\max\{\log{n}/H_n^2, \lambda/H_n\})$}, where,
recall, $W_n$ is the minimum distance between changepoints in
$\theta_0$, as defined in \eqref{eq:Wn}, and $H_n$ is the minimum gap
betwen levels of $\theta_0$, as defined in \eqref{eq:Hn}.
Then the fused lasso estimate \smash{$\hat{\theta}$} in \eqref{eq:fl}
with tuning parameter $\lambda$ satisfies
\[
\P\Big( d(\hS \,|\, S_0 ) \leq r_n \Big) \to 1 
\quad \text{as} \quad n \to \infty.
\]
Under assumptions (i), (ii') \smash{$\P( \hat{s} = s_0) \to 1$},
(iii), and (iv') \smash{$r_n=\omega(\max\{\log{n}/H_n^2,
  \log{(n^5/\lambda^2)}/H_n^2\})$}, we instead have
\[
\P\Big(d_H(\hS, S_0) \leq r_n\Big) \to 1 
\quad \text{as} \quad n \to \infty.
\]
\end{theorem}

\begin{remark}[\textbf{Stringency of conditions}]
\label{rem:stringency}
Assumption (ii') in the above result, needed for the bound on the 
Hausdorff distance, states that the number of estimated changepoints 
\smash{$\hat{s}$} in \smash{$\htheta$} equals the number of
changepoints $s_0$ in $\theta_0$ with probability tending to 1, which
is of course a very strong assumption.  Assumption (ii), needed for
the bound on  
the one-sided screening distance, states that \smash{$\hat{s} \geq
  s_0$}, which is itself fairly strong, though believable if $\lambda$ 
is chosen to be small enough.  

The tuning parameter $\lambda$ in fact plays an important role in the
achieved rates $r_n$ in \Fref{thm:harchaoui_recovery}.  In order to
satisfy condition (iv) on $r_n$, a choice of $\lambda=\log{n}/H_n$
gives the tightest possible scaling for $r_n$. Then we require $r_n$ 
to grow faster than $\log{n}$, e.g., at the rate
\smash{$\log^2{n}$}.  This choice is basically the same as that 
discussed in \citet{harchaoui2010multiple}; with this choice, their
results show that the screening distance achieved by the fused lasso
estimator is at most \smash{$\log^2{n}$}, with probability tending to
1. However, the choice of $\lambda$ here is worrisome---it is
considerably smaller than the choices known to 
achieve reasonable error rates, specifically in the strong
sparsity setting, with $s_0=O(1)$, where we expect $\lambda$ to scale
at something like a \smash{$\sqrt{n}$} rate (see Theorems 
\ref{thm:strong_sparsity} and \ref{thm:strong_james}).  With such a
small choice of \smash{$\lambda=\log n/H_n$}, there would likely be a
very large number of
estimated changepoints in \smash{$\htheta$}, rendering a 
quantity like the screening distance uninteresting.\footnote{Of
course, with $\lambda=0$, the screening distance achieved by the
fused lasso is trivially zero. Hence, when studying screening
distance, it is implicitly understood that some other aspect of
\smash{$\htheta$} must be kept in balance. In our work, we study
screening distances while maintaining that \smash{$\htheta$}
must exhibit good estimation performance, as measured by its
$\ell_2$ error rate.}  
The same critique could made be about 
condition (iv'), needed for the bound on the Hausdorff distance.
In particular, condition (ii') seems unrealistic
unless $\lambda$ is chosen to be much larger.  

For the reasons just described, we will consider a larger
scaling of $\lambda=\Theta(\sqrt{n})$, when comparing 
\Fref{thm:harchaoui_recovery} to our new results on changepoint
screening and recovery in Sections \ref{sec:changepoint_screening} 
and \ref{sec:changepoint_recovery}.
\end{remark}

\begin{remark}[\textbf{Other fused lasso recovery results}]
Several other results have appeared in the literature regarding
changepoint recovery for the fused lasso.
\citet{rinaldo2009properties} studied exact recovery
of changepoints (in which the achieved Hausdorff
distance would be zero).  There is an error in the proof of his
Theorem 2.3, which invalidates the 
result.\footnote{See the correction note posted at   
\url{http://www.stat.cmu.edu/\~arinaldo/Fused\_Correction.pdf}.} 
\citet{qian2012pattern} studied a modification of the fused lasso
defined by transforming the fused lasso problem \eqref{eq:fl} into a
lasso problem with 
particular design matrix $X$, and then applying a step that
``preconditions'' $y$ and $X$.  The authors concluded that exact  
recovery is possible with probability at tending to 1, as long as the
minimum signal gap and tuning parameter satisfy 
\smash{$H_n \geq \lambda = \omega(\sqrt{\log{n}})$}.  This is a
very strong requirement on the scaling of the signal gap
$H_n$; \citet{sharpnack2012sparsistency} showed that, when
\smash{$H_n \geq \omega(\sqrt{\log{n}})$}, 
even simple pairwise  
thresholding (i.e., thresholding based on the observed absolute
differences $|y_i-y_{i+1}|$, $i=1,\ldots,n-1$) achieves exact 
recovery.  Most recently, \citet{rojas2014change} established an 
impossibility result for the fused lasso estimator when 
$\theta_0$ exhibits a ``staircase'' pattern, which means that
\smash{$D_{S_0}\theta_0$} has two consecutive positive or negative
values; specifically, these authors proved that for such a staircase
pattern, the quantity \smash{$d_H(\hS, S_0)/n$} remains bounded away
from zero with nonzero asymptotic probability.  For non-staircase
patterns in $\theta_0$, the authors also showed, under certain
assumptions, that \smash{$d_H(\hS, S_0)/n$} converges to zero in
probability. 
\end{remark}

\begin{remark}[\textbf{Comparable recovery properties of other
    estimators}]  
\label{rem:others_recovery}
It is worth describing relevant changepoint recovery properties of
various methods in the literature. 
\citet{boysen2009consistencies} showed that the Potts estimator,
denoted by \smash{$\htheta^{\mathrm{Potts}}$}, satisfies 
\smash{$d_H(S(\htheta^{\text{Potts}}), S_0) = O(\log n)$} a.s.,
when $s_0=O(1)$. \citet{frick2014multiscale} proposed 
a simultaneous multiscale changepoint estimator (SMUCE), using  
an $\ell_0$-penalized optimization problem (like the Potts estimator), 
and under weak assumptions on $W_n,H_n$, proved that
their estimator \smash{$\htheta^{\mathrm{SMUCE}}$} satisfies 
\smash{$d_H(S(\htheta^\mathrm{SMUCE}), S_0) = O(\log{n}/H_n^2)$} 
with probability tending to 1.
There is quite a large body of literature on binary 
segmentation (BS). To the best of our knowledge,   
the sharpest analysis for BS is in
\citet{fryzlewicz2014wild}, who also proposed and analyzed a ``wild''  
(i.e., randomized) variant of the method (WBS). 
Denoting these two estimators by \smash{$\htheta^{\mathrm{BS}}$} and   
\smash{$\htheta^{\mathrm{WBS}}$}, \citet{fryzlewicz2014wild}
established that
\smash{$d_H(S(\htheta^\text{BS}), S_0) = O(n^2\log{n}/
 (W_n^2H_n^2))$} and 
\smash{$d_H(S(\htheta^\text{BS}), S_0) = O(\log{n}/H_n^2)$}, both   
with probability tending to 1, and both under certain restrictions
on $W_n,H_n$, these restrictions being stronger for BS than for WBS. 
Very recently, \citet{fryzlewicz2016tail} proved that his tail-greedy 
unbiased Haar estimator, denoted by \smash{$\htheta^{\mathrm{TGUH}}$}, 
satisfies \smash{$d_H(S(\htheta^\text{TGUH}), S_0) = O(\log^2{n})$}
with probability tending to 1, under weak conditions on $W_n,H_n$.
All of these results will be revisited in greater detail in Remarks
\ref{rem:recovery_weak} and \ref{rem:recovery_strong}. 


Lastly, it should be noted that many of the methods
described here also come  
with a guarantee (under possibly additional conditions) that they
correctly identify the number of changepoints 
$s_0$ in $\theta_0$, with probability tending to 1.  We refer the
reader to the references above, for details.
\end{remark}

\section{Error analysis under strong sparsity} 
\label{sec:strong_james}

In this section, we derive a new $\ell_2$ estimation error bound for
the fused lasso in the strong sparsity case, improving on
the result of \citet{dalalyan2014prediction} stated in
\Fref{thm:strong_sparsity}. Our proof is based on the 
concept of a {\em lower interpolant}, which as far as we can tell
is a new idea that may be of interest in its own right.   
We first state our error bound.

\begin{theorem}[\textbf{Fused lasso error rate, strong sparsity
    setting}]   
\label{thm:strong_james}
Assume the data model in \eqref{eq:model}, with errors
$\epsilon_i$, $i=1,\ldots,n$ i.i.d.\ from a sub-Gaussian
distribution as in \eqref{eq:subgauss}.  Then for a choice of tuning
parameter \smash{$\lambda=(nW_n)^{1/4}$}, the the fused lasso estimate  
\smash{$\htheta$} in \eqref{eq:fl} satisfies
\begin{equation*}
\|\htheta - \theta_0\|_n^2 \leq \gamma^2  c \frac{s_0}{n}
\Bigg( (\log{s_0}+\log\log{n}) \log{n} + \sqrt{\frac{n}{W_n}} \Bigg), 
\end{equation*}
with probability at least $1-\exp(-C\gamma)$, for all
$\gamma > 1$ and $n \geq N$, where $c,C,N>0$ are constants.
\end{theorem}

\begin{remark}[\textbf{The roles of $s_0,W_n$}] When $s_0$
grows quickly enough with $n$, the error rate provided by the above 
theorem will become worse than the weak sparsity rate in    
\Fref{thm:weak_sparsity}. Given $s_0$
evenly spaced jumps, so that  
\smash{$W_n=\Theta(n/s_0)$}, the rate in \Fref{thm:strong_james} is
\smash{$(\log{s_0}+\log\log{n}) s_0 \log{n}/n + s_0^{3/2}/n$};
when $s_0$ grows faster than \smash{$n^{2/9}$}, this is slower than
the \smash{$n^{-2/3}$} rate in \Fref{thm:weak_sparsity} (assuming 
$C_n=O(1)$).  
\Fref{thm:strong_james} gives the fastest rate when
$s_0=O(1)$, $W_n=\Theta(n)$, this being
\smash{$(\log{n}\log\log{n})/n)$}, an improvement over the 
rate in \Fref{thm:strong_sparsity}.  This comparison, and the
comparison to other results in the literature, will be drawn out in
more detail in the last remark of this section. 
\end{remark}

\begin{remark}[\textbf{Expectation bound}]
An expectation bound follows more of less directly from the high
probability bound in \Fref{thm:strong_james}.  Define
the random vairable 
\begin{equation*}
M = \frac{\| \htheta - \theta_0 \|_2^2}{c^2 s_0
  ((\log{s_0}+\log\log{n})\log{n}+\sqrt{n/W_n})}, 
\end{equation*}
which we know has the tail bound \smash{$\P(U > u) \leq 
\exp(-C\sqrt{u})$} for $u > 1$, and observe that 
\begin{equation*}
\E(U) = \int_0^\infty \P(U > u) \, du \leq 
1 + \int_1^\infty \exp(-C\sqrt{u}) \, du.
\end{equation*}
The right-hand side is a finite constant, and this gives the
result
\begin{equation*}
\E\|\htheta-\theta_0\|_n^2 \leq c \frac{s_0}{n}
\Bigg( (\log{s_0}+\log\log{n})\log{n} + \sqrt{\frac{n}{W_n}} \Bigg), 
\end{equation*}
where the constant $c>0$ is adjusted to be larger, as needed.  
\end{remark}

Here is an overview of the proof of \Fref{thm:strong_james}. 
The details are deferred until Appendix \ref{app:strong_james},
and the proofs of the lemmas stated below are given in Appendix 
\ref{app:lemma_james}. We consider a decomposition
\[
\htheta - \theta_0 = P_0 (\htheta - \theta_0) + P_1 \htheta,
\]
where $P_0$ is the projection matrix onto the piecewise constant
structure inherent to the mean $\theta_0$, and $P_1=I-P_1$.  To give
more detail, recall that we write $S_0=\{t_1,\ldots,t_{s_0}\}$ for
the changepoints in $\theta_0$, ordered as in $t_1<\ldots<t_{s_0}$, 
and we write $t_0 = 0$ and $t_{s_0+1} =
n$ for convenience.  Furthermore, define
$B_j=\{t_j+1,\ldots,t_{j+1}\}$, and write $\one_{B_j} \in \R^n$ for 
the indicator of block $B_j$, for 
$j=0,\ldots,s_0$.  With this notational setup, we may now define  
$P_0$ as the projection onto the $(s_0+1)$-dimensional linear subspace 
$\cR = \spa \{\one_{B_0}, \ldots, \one_{B_{s_0}}\}$.  It is common 
practice (e.g., see \citet{vandegeer2000empirical})
to bound the estimation error by bounding the empirical process term  
\smash{$\epsilon^\top (\htheta-\theta_0)$}, where $\epsilon \in \R^n$  
is the vector of errors in the data model \eqref{eq:model}.  Using the
decomposition above, this becomes
\[
\epsilon^\top (\htheta - \theta_0) = \epsilon^\top 
\hat\delta + \epsilon^\top \hx,
\]
where we define \smash{$\hat\delta = P_0 (\htheta - \theta_0)$} and 
\smash{$\hx = P_1 \htheta$}.  
The parameter \smash{$\hat\delta$} lies in an
$(s_0+1)$-dimensional space, which makes bounding
\smash{$\epsilon^\top \hat\delta$} 
relatively easy.  Bounding the term \smash{$\epsilon^\top \hx$}
requires a much more intricate argument, which is spelled out in the 
following lemmas.
\Fref{lem:interp} is a deterministic result ensuring the
existence of what we call the {\em lower interpolant}
\smash{$\hz$} to the vector \smash{$\hx$}.  This interpolant 
approximates \smash{$\hx$} using roughly
$2s_0+2$ monotonic pieces, and its empirical process term
\smash{$\epsilon^\top \hz$} can be finely controlled, as shown in
\Fref{lem:GCinterp}.  The residual from the interpolant
approximation, denoted \smash{$\hw = \hx - \hz$}, has an
empirical process term \smash{$\epsilon^\top \hw$} that is more
crudely controlled, in \Fref{lem:GCresid}.  Put together, as in
\smash{$\epsilon^\top \hx = \epsilon^\top \hz +
\epsilon^\top \hw$}, gives the final control on
\smash{$\epsilon^\top \hx$}. 

Before stating \Fref{lem:interp}, we define the class of vectors
containing the lower interpolant.
Given any collection of changepoints $t_1<\ldots<t_{s_0}$ (and
$t_0 = 0$, $t_{s_0+1} = n$), let $\cM$ be the set of ``piecewise
monotonic'' vectors $z \in \R^n$, with the following properties, 
for each $i = 0, \ldots, s_0$:
\begin{enumerate}
\item[(i)] there exists a point $t'_i$ such that $t_i+1 \leq t'_i \leq
  t_{i+1}$, and the absolute value $|z_j|$ is nonincreasing over
  the segment $j \in \{t_i+1,\ldots,t'_i \}$, and nondecreasing over
  the segment $j \in \{t'_i,\ldots,t_{i+1}\}$; 
\item[(ii)] the signs remain constant on the monotone pieces,
\begin{align*}
&\sign(z_{t_i}) \cdot \sign(z_j) \geq 0, \quad 
j = t_i+1,\ldots,t'_i, \\
&\sign(z_{t_{i+1}}) \cdot \sign(z_j) \geq 0, \quad
j = t'_i+1,\ldots,t_{i+1}.
\end{align*}
\end{enumerate}

Now we state our lemma that characterizes the lower
interpolant.

\begin{lemma}
\label{lem:interp}
Given changepoints $t_0<\ldots<t_{s_0+1}$, and any $x \in \R^n$, there
exists a vector $z \in \cM$ (not necessarily unique), such that the
following statements hold:
\begin{align}
\label{eq:Zcond1}
&\| D_{-S_0}  x \|_1 = \| D_{-S_0}  z \|_1 + \|D_{-S_0} ( x -  z) \|_1, \\
 \label{eq:Zcond2}
&\| D_{S_0}  x \|_1 = \| D_{S_0}  z \|_1 \le \| D_{-S_0}  z \|_1 +
\frac{4\sqrt{s_0}}{\sqrt{W_n}} \|  z\|_2, \\
\label{eq:Zcond3}
&\| z \|_2 \le \|  x \|_2 
\quad \text{and} \quad 
\| x -  z \|_2 \le \|  x \|_2,
\end{align}
where $D \in \R^{(n-1)\times n}$ is the difference matrix in
\eqref{eq:diff}. We call a vector $z$ with these properties a 
{\em lower interpolant} to $x$.
\end{lemma}

Loosely speaking, the lower interpolant \smash{$\hz$} can be 
visualized by taking a string that lies initially on top of
\smash{$\hx$}, is nailed down at the changepoints $t_0, \ldots
t_{s_0+1}$, and then pulled taut while maintaining that it is not
greater (elementwise) than \smash{$\hx$},
in magnitude.  Here ``pulling taut'' means that \smash{$\| D \hz\|_1$}
is made small. Figure \ref{fig:interp} provides
illustrations of the interpolant \smash{$\hz$} to \smash{$\hx$}
for a few examples.

\begin{figure}[!ht]
  \centering
  \includegraphics[width=\textwidth]{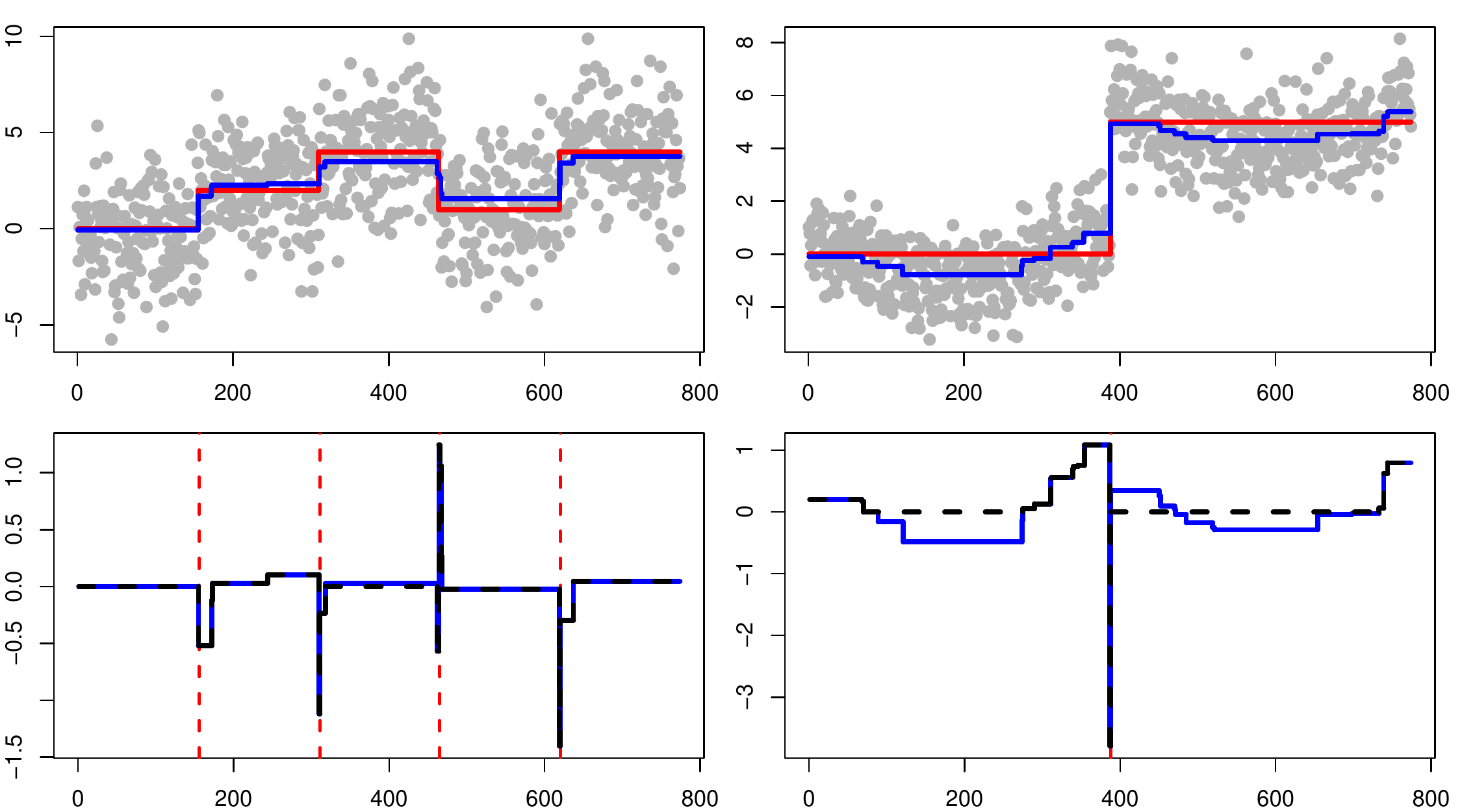}  
  \caption{\it\small
    The lower interpolants for two examples (in the left and
    right columns), each with $n=800$ points.
    In the top row, the data $y$ (in gray) and underlying signal
    $\theta_0$ (red) are plotted across the locations
    $1,\ldots,n$. Also shown is the fused lasso estimate
    \smash{$\htheta$} (blue).  In the bottom row, the error vector 
    \smash{$\hx = P_1\htheta$} is plotted (blue) as well as the
    interpolant (black), and the dotted
    vertical lines (red) denote the changepoints $t_1,\ldots
    t_{s_0}$ of $\theta_0$.}     
\label{fig:interp}
\end{figure}

Note that \smash{$\hz$} consists of $2s_0+2$ monotonic pieces. This  
special structure leads to a sharp concentration inequality. 
The next lemma is the primary contributor to the fast rate given
in \Fref{thm:strong_james}. 

\begin{lemma}
\label{lem:GCinterp}
Given changepoints $t_1<\ldots<t_{s_0}$, there exists constants $c_I,
C_I,N_I>0$ such that when $\epsilon \in \R^n$ has i.i.d.\ sub-Gaussian
components satisfying \eqref{eq:subgauss}, 
\begin{equation*}
\P \Bigg( \sup_{z \in \cM} \, \frac{|\epsilon^\top z|}{\| z \|_2} > 
\gamma c_I \sqrt{(\log{s_0}+\log\log{n}) s_0 \log{n}} \Bigg) \le 
2 \exp\big(-C_I \gamma^2 c_I^2 (\log{s_0}+\log\log{n})\big),   
\end{equation*}
for any $\gamma > 1$, and $n \geq N_I$.
\end{lemma}

Finally, the following lemma controls the residuals,
\smash{$\hw= \hx-\hz$}.

\begin{lemma}
\label{lem:GCresid}
Given changepoints $t_1<\ldots<t_{s_0}$, there exists constants $c_R,
C_R > 0$ such that when $\epsilon \in \R^n$ has i.i.d.\ sub-Gaussian
components satisfying \eqref{eq:subgauss}, 
\begin{equation*}
\P \bigg( \sup_{w \in \cR^\perp} \, 
\frac{|\epsilon^\top w|}{\sqrt{\|D_{-S_0} w \|_1 
\|w\|_2}} > \gamma c_R (n s_0)^{1/4} \bigg) \le 2
\exp(-C_R \gamma^2 c_R^2 \sqrt{s_0}), 
\end{equation*}
for any $\gamma > 1$, where $\cR^\perp$ is the orthogonal
complement of $\cR = \spa \{\one_{B_0}, \ldots, \one_{B_{s_0}}\}$. 
\end{lemma}


We conclude this section with a remark comparing
\Fref{thm:strong_james} to related results in the literature.

\begin{remark}[\textbf{Comparison to \Fref{thm:strong_sparsity}}] 
\label{rem:others_error_strong}
We compare \Fref{thm:strong_james} to the strong sparsity result in   
\citet{dalalyan2014prediction}, as stated in
\Fref{thm:strong_sparsity}.  For any $s_0,W_n$, the
former rate is sharper than the latter, since $\log{s_0} \leq \log{n}$
and \smash{$\sqrt{n/W_n} \leq n/W_n$}.  Moreover, when $s_0=O(1)$
and $W_n=\Theta(n)$, the rates are $(\log{n}\log\log{n})/n$ versus 
$\log^2{n}/n$, in Theorems \ref{thm:strong_james} and 
\ref{thm:strong_sparsity} respectively. Furthermore, in this setting,
we note that the scaling of the tuning parameter investigated by both
theorems is \smash{$\lambda=\Theta(\sqrt{n})$}. 

 As discussed in Remarks
\ref{rem:harchaoui_error} and \ref{rem:others_error}, essentially all  
rates from comparable estimators in the strong sparsity  
case scale as $\log^2{n}/n$, with the exception being  
the Potts estimator, which has a rate
of $\log{n}/n$.  Therefore the improvement from $\log^2{n}/n$ to  
$(\log{n}\log\log{n})/n$ offered by \Fref{thm:strong_james} could 
certainly be viewed as nontrivial. An error rate faster than
$\log{n}/n$ in the strong sparsity case seems likely 
unattainable by any method, as \citet{donoho1994ideal}   
showed that an oracle wavelet estimator (that is allowed the
optimal choice of wavelet threshold for each problem instance) still
has an expected estimation error on the order of $\log{n}/n$. 
\end{remark}

\section{Approximate changepoint screening and recovery}  
\label{sec:screening_recovery}

We develop results on approximate screening of changepoints by the 
fused lasso, and approximate recovery of changepoints after a
post-processing step has been applied to the fused lasso estimate. 
A distinctive feature of our results is that their proofs based on 
only the $\ell_2$ estimation error rates achieved
by the fused lasso.  In fact, in their most general form, our results
imply certain changepoint screening and recovery properties for any
estimation method that has a known $\ell_2$ error rate, which clearly
has implications well beyond the fused lasso.

\subsection{Results on approximate changepoint screening}
\label{sec:changepoint_screening}

We present a theorem that takes a general estimator
\smash{$\ttheta$} of $\theta_0$, with a known $\ell_2$
estimation error rate, and infers a bound on the 
screening distance between changepoints of $\theta_0$ and those of 
\smash{$\ttheta$}. 

\begin{theorem}[\textbf{Generic screening result}]
\label{thm:changepoint_screening}
Let $\theta_0 \in \R^n$ be a piecewise constant vector, and 
\smash{$\ttheta \in \R^n$} be an estimator that satisfies the  
error bound \smash{$\|\ttheta-\theta_0\|_n^2 = O_\P(R_n)$}.
Assume that \smash{$nR_n/H_n^2 = o(W_n)$},
where, recall, $H_n$ is the
minimum gap between adjacent levels of $\theta_0$, defined in
\eqref{eq:Hn}, and $W_n$ is the minimum distance between adjacent
changepoints of $\theta_0$, defined in \eqref{eq:Wn}. Then  
\begin{equation*}
d\big(S(\ttheta)\,|\,S_0 \big) = O_\P 
\bigg( \frac{n R_n}{H_n^2} \bigg), 
\end{equation*}
where \smash{$S(\ttheta)$} is the set of changepoints in
$\ttheta$, $S_0=S(\theta_0)$ is the set of changepoints in $\theta_0$, 
and $d(\,\cdot\,|\,\cdot\,)$ is the one-sided screening distance, as
defined in \eqref{eq:metrics}. 
\end{theorem}

\begin{proof}
The proof is derived from the $\ell_2$ rate. 
Fix any $\epsilon>0$, $C>0$.  By 
assumption, we know that there is an integer $N_1>0$ such  
that   
\begin{equation*}
\P\bigg( \|\ttheta-\theta_0\|_n^2 > \frac{C}{4} R_n\bigg) \leq  
\epsilon,  
\end{equation*}
for all $n \geq N_1$. We also know that there is an integer
$N_2>0$ such that \smash{$CnR_n/H_n^2 \leq W_n$} 
for all $n \geq N_2$. Let $N=\max\{N_1,N_2\}$, take $n \geq N$,    
and let \smash{$r_n = \lfloor CnR_n/H_n^2\rfloor$}. Suppose that    
\smash{$d(S(\ttheta) \,|\, S_0) > r_n$}.  Then, by definition, there
is a changepoint $t_i \in S_0$ such that no changepoints of  
\smash{$\ttheta$} are within $r_n$ of $t_i$, which means that 
\smash{$\ttheta_j$} is constant over 
$j \in \{t_i-r_n+1,\ldots,t_i+r_n\}$.   Denote 
\begin{equation*}
z=\ttheta_{t_i-r_n+1}=\ldots=\ttheta_{t_i}=\ttheta_{t_i+1}=
\ldots=\ttheta_{t_i+r_n}.
\end{equation*}
We then form the lower bound
\begin{equation} 
\label{eq:lower_bd}
\frac{1}{n}\sum_{j=t_i-r_n+1}^{t_i+r_n} 
\big(\ttheta_j-\theta_{0,j}\big)^2 
= \frac{r_n}{n}\big(z-\theta_{0,t_i}\big)^2 +
\frac{r_n}{n}\big(z-\theta_{0,t_i+1}\big)^2 \geq 
\frac{r_n H_n^2}{2n} > \frac{C}{4} R_n,
\end{equation}
where the first inequality holds because $(x-a)^2 +
(x-b)^2 \geq (a-b)^2/2$ for all $x$ (the quadratic in $x$ here 
is minimized at $x=(a+b)/2$), and the second because
\smash{$r_n = \lfloor CnR_n/H_n^2 \rfloor $}.
Therefore, we see that
\smash{$d(S(\ttheta) \,|\, S_0) > r_n$} implies 
\begin{equation*}
\|\ttheta-\theta_0\|_n^2 \geq 
\frac{1}{n}\sum_{j=t_i-r_n+1}^{t_i+r_n }
\big(\ttheta_j-\theta_{0,j}\big)^2 >
\frac{C}{4} R_n,
\end{equation*}
which implies 
\begin{equation*}
\P\Big( d\big(S(\ttheta)\,|\,S_0 \big) > r_n \Big) \leq 
\P\bigg( \|\ttheta-\theta_0\|_n^2 > \frac{C}{4} R_n\bigg) \leq  
\epsilon,
\end{equation*}
for all $n \geq N$, completing the proof.
\end{proof}

\begin{remark}[\textbf{Conditions on $W_n,H_n$}]
The condition that \smash{$nR_n/H_n^2 = o(W_n)$} in 
\Fref{thm:changepoint_screening} is not
strong. Consider the simple case in which
$H_n=\Omega(1)$ and $W_n=\Theta(n)$. This
condition reduces to $R_n=o(1)$, requiring only that the estimator   
\smash{$\ttheta$} in question be consistent. The theorem 
then gives the bound
\smash{$d(S(\ttheta) \,|\, S_0) = O_\P(nR_n)$} on the screening
distance obtained by \smash{$\ttheta$}.
\end{remark}

\begin{remark}[\textbf{Generic setting: no particular assumptions on  
  data model, or estimator}]
\label{rem:no_model}
Importantly, \Fref{thm:changepoint_screening} assumes no data model
whatsoever, and 
treats \smash{$\ttheta$} as a generic estimator of $\theta_0$.  Of
course, through the statement 
\smash{$\|\ttheta-\theta_0\|_n^2 = O_\P(R_n)$},
one sees that \smash{$\ttheta$} is random,
constructed from data that depends on $\theta_0$, but no specific data
model is required, nor are any specific properties of
\smash{$\ttheta$} (other than its $\ell_2$ error rate). 
This flexibility allows for the result to be applied in any problem 
setting in which one has control of the estimation error of
a piecewise constant 
parameter $\theta_0$.  Apart from the applications of
\Fref{thm:changepoint_screening} to the fused lasso estimator, where
we consider data from a standard model as in \eqref{eq:model},
with $\theta_0$ being the mean, and i.i.d.\ sub-Gaussian errors
(see Corollaries \ref{cor:changepoint_screening_weak} and 
\ref{cor:changepoint_screening_strong} below), we could instead
suppose that data is distributed according to, e.g., a Poisson model
with natural parameter $\theta_0$, 
\begin{equation}
\label{eq:model_pois}
y_i \sim \mathrm{Pois}(e^{\theta_{0,i}}), \quad i=1,\ldots,n.
\end{equation}
If we knew of an estimate \smash{$\ttheta$} for $\theta_0$ such that 
the estimation error \smash{$\|\ttheta-\theta_0\|_n^2$} was analyzable, then we
could
use \Fref{thm:changepoint_screening} to infer a bound on the
screening distance between changepoints of $\theta_0$ and
\smash{$\ttheta$}.   
In this paper, we do not describe particular applications of
\Fref{thm:changepoint_screening} beyond the sub-Gaussian error model
in \eqref{eq:model}, since we are not aware of estimation error
guarantees outside of this model.  However, establishing $\ell_2$
estimation error rates for models like \eqref{eq:model_pois} (which
may be used to describe say copy number data in genetics), and
interpreting the resulting changepoint approximation guarantees 
would be an interesting topic  
for future work.  (For model \eqref{eq:model_pois}, and
other likelihood-based models with a piecewise constant parameter 
$\theta_0$, we suspect that the fused lasso provides a basis
for a good estimator: simply replace the squared error loss
in \eqref{eq:fl} by the negative log likelihood.)
\end{remark} 

We present two different corollaries of
\Fref{thm:changepoint_screening} for the fused lasso. The
first is given by using \Fref{thm:weak_sparsity} and the associated
rate $R_n=n^{-2/3}C_n$ in the weak sparsity case, and the second
is given by using \Fref{thm:strong_james} and the associated rate
$R_n=(\log{n}\log\log{n})/n$ in the strong sparsity case. The proofs
are immediate and are hence omitted.

\begin{corollary}[\textbf{Fused lasso screening result, weak  
  sparsity setting}]  
\label{cor:changepoint_screening_weak}
Assume the conditions of \Fref{thm:weak_sparsity}, 
so that $\TV(\theta_0) \leq C_n$ for a sequence $C_n$. 
Also assume that
\smash{$H_n = \omega(n^{1/6} C_n^{1/3} / \sqrt{W_n})$}.
Let \smash{$\htheta$} be the fused lasso estimate in
\eqref{eq:fl}, with \smash{$\lambda = \Theta(n^{1/3}C_n^{-1/3})$}.  
Then 
\begin{equation*}
d\big(\hS\,|\,S_0 \big) = 
O_\P\bigg(\frac{n^{1/3} C_n^{2/3}}{H_n^2}\bigg).
\end{equation*}
\end{corollary}

\begin{corollary}[\textbf{Fused lasso screening result, strong 
  sparsity setting}]   
\label{cor:changepoint_screening_strong}
Assume the conditions of \Fref{thm:strong_sparsity}, 
so that
$s_0 = O(1)$ and $W_n=\Theta(n)$. 
Also assume that
\smash{$H_n = \omega(\sqrt{(\log{n}\log\log{n})/n})$}.
Let \smash{$\htheta$} be the fused lasso estimate in
\eqref{eq:fl}, with \smash{$\lambda = \Theta(\sqrt{n})$}. 
Then 
\begin{equation*}
d\big(\hS\,|\,S_0 \big) = O_\P\bigg(
\frac{\log{n} \log\log{n}}{H_n^2} \bigg).
\end{equation*}
\end{corollary}

\begin{remark}[\textbf{Conditions on $W_n,H_n$}]
We have rewritten the condition that
\smash{$nR_n/H_n^2 = o(W_n)$} in 
\Fref{thm:changepoint_screening} as 
\smash{$H_n = \omega(n^{1/6} C_n^{1/3} / \sqrt{W_n})$} in
\Fref{cor:changepoint_screening_weak}, and 
\smash{$H_n = \omega(\sqrt{(\log{n}\log\log{n})/n})$} in
\Fref{cor:changepoint_screening_strong} (note that in the latter, we
are assuming that $W_n=\Theta(n)$). 
\end{remark}

\begin{remark}[\textbf{Screening under weak sparsity}]
\label{rem:screening_weak}
\Fref{cor:changepoint_screening_weak}
handles a difficult setting in which the number
of changepoints $s_0$ in $\theta_0$ can grow quickly
with $n$, and yet it still provides a reasonable bound on the
screening distance \smash{$d(S(\htheta) \,|\, S_0)$} provided that
$C_n$ is not too large (i.e., $\TV(\theta_0)$ is not growing too
quickly), or $H_n$ is large enough (i.e., the minimum signal gap
in $\theta_0$ is large enough). 
As an example, suppose that $s_0=\Theta(n^{1/6})$, and the  
changepoints in $\theta_0$ are evenly spread out, so that
$W_n=\Theta(n^{5/6})$. Then \Fref{cor:changepoint_screening_weak}  
implies, provided that \smash{$H_n = \omega(n^{-1/4} C_n^{1/3})$},   
\begin{equation*}
d\big(\hS\,|\,S_0 \big) = 
O_\P\bigg(\frac{n^{1/3} C_n^{2/3}}{H_n^2}\bigg) =
o_\P(n^{5/6}),
\end{equation*}
so for each true changepoint, there is at least one estimated
changepoint that is much closer to it than all of the other true
changepoints 
(each of which is at least a distance \smash{$W_n=\Theta(n^{5/6})$} 
away). From the condition \smash{$H_n = \omega(n^{-1/4} 
  C_n^{1/3})$}, and the fact that we must always have 
$C_n \geq s_0 H_n$ (recall \smash{$s_0=\Theta(n^{1/6})$}),
we can be more explicit here about the allowable 
ranges for $H_n,C_n$: combining the last two relationships gives 
\smash{$C_n = \omega(n^{-1/8})$}, and then 
\smash{$H_n = \omega(n^{-7/24})$}.  Hence, the minimum signal
gap requirement here is very reasonable, allowing $H_n$ to
shrink to 0, just not too quickly (this is far from a trivial regime,
e.g., with \smash{$H_n=\omega(\sqrt{\log{n}})$}, when simple
thresholding of pairwise differences achieves perfect recovery, as
shown in  \citet{sharpnack2012sparsistency}). 
\end{remark}

\begin{remark}[\textbf{Comparison to \Fref{thm:harchaoui_recovery}}] 
\label{rem:screening_strong}
\Fref{cor:changepoint_screening_strong} provides a similar conclusion
to that in \citet{harchaoui2010multiple}, restated in
\Fref{thm:harchaoui_recovery}: in a strong sparsity setting, the fused
lasso has a well-controlled screening distance, only slightly larger
than \smash{$\log{n} / H_n^2$}.  However, we note that
\Fref{cor:changepoint_screening_strong} guarantees this screening
bound under a natural choice for the tuning parameter $\lambda$, known 
to provide good $\ell_2$ estimation performance (see 
\Fref{thm:strong_james}), whereas \Fref{thm:harchaoui_recovery}
implicitly requires $\lambda$ to be very small, which seems unnatural
(see \Fref{rem:stringency}).  
\end{remark}

\begin{remark}[\textbf{Changepoint detection limit}]
\label{rem:Hmin}
The restriction that 
\smash{$H_n = \omega(\sqrt{(\log{n}\log\log{n})/n})$} in
\Fref{cor:changepoint_screening_strong} is very close to the optimal
limit of \smash{$H_n=\omega(1/\sqrt{n})$} for changepoint
detection: 
\citet{duembgen2008multiscale} showed that in 
Gaussian changepoint model with a single elevated region, and
$W_n=\Theta(n)$, there is no test for detecting a changepoint that has  
asymptotic power 1 unless   
\smash{$H_n = \omega(1/\sqrt{n})$}. See also 
\citet{chan2013detection}.
\end{remark}

\subsection{Post-processing for approximate changepoint recovery}  
\label{sec:changepoint_recovery}

We study a procedure for post-processing the estimated changepoints in 
\smash{$\ttheta$}, in such a way that aims to eliminate
changepoints of \smash{$\ttheta$} that lie far away from changepoints
of $\theta_0$. Our procedure is based on convolving \smash{$\ttheta$}
with a filter that resembles the mother Haar wavelet.  Consider
\begin{equation}
\label{eq:filter}
F_i(\ttheta) = 
\frac{1}{b_n}\sum_{j=i+1}^{i+b_n}\ttheta_j -
\frac{1}{b_n}\sum_{j=i-b_n+1}^{i}\ttheta_j, \quad
\text{for $i=b_n,\ldots,n-b_n$},
\end{equation}
for an integral bandwidth $b_n>0$.
Our result in this subsection asserts that, by evaluating the
filter \smash{$F_i(\ttheta)$} at all locations $i=b_n,\ldots,n-b_n$,
and retaining only locations at which the filter value is large
(in magnitude), we can approximately recovery the changepoints of
$\theta_0$, in the Hausdorff metric.

\begin{theorem}[\textbf{Generic recovery result}]
\label{thm:changepoint_recovery}
Let $\theta_0 \in \R^n$ be a piecewise constant vector, and 
\smash{$\ttheta \in \R^n$} be an estimator that satisfies the  
error bound \smash{$\|\ttheta-\theta_0\|_n^2 = O_\P(R_n)$}. 
Consider the following procedure: we evaluate the filter in
\eqref{eq:filter} with bandwidth $b_n$ at all locations
$i=b_n,\ldots,n-b_n$, and we keep only the
locations whose filter value is greater than or equal to a threshold
level $\tau_n$, in magnitude.  Denote the resulting ``filtered'' set 
by   
\begin{equation}
\label{eq:filtered}
S_F(\ttheta) = \Big\{ i \in \{b_n,\ldots,n-b_n\} : |F_i(\ttheta)| 
\geq \tau_n \Big\}. 
\end{equation}
If the bandwidth and threshold values 
satisfy \smash{$b_n=\omega(nR_n/H_n^2)$}, $2b_n \leq W_n$, 
and $\tau_n/H_n \to \rho \in (0,1)$ as $n \to \infty$, then we
have 
\begin{equation*}
\P\Big( d_H\big(S_F(\ttheta), S_0 \big) \leq b_n \Big) \to 1 
\quad \text{as $n \to \infty$},
\end{equation*}
where $d_H(\,\cdot, \cdot\,)$ is the Hausdorff distance, as defined
in \eqref{eq:metrics}.  
\end{theorem}

\begin{proof}
The proof is not complicated conceptually, but requires some
careful bookkeeping.  Also, we make use of a few key lemmas whose   
details will be given later.  Fix $\epsilon>0$ and $C>0$. Let 
$N_1>0$ be an integer such that for all $n \geq N_1$, 
\begin{equation*}
\P\Big( \|\ttheta-\theta_0\|_n^2 > C R_n\Big) \leq   
\frac{\epsilon}{2}.
\end{equation*}
Set $\epsilon=\min\{\rho,1-\rho\}/2$. 
As \smash{$b_n=\omega(nR_n/H_n^2)$}, there is an integer 
$N_2>0$ such that for all $n \geq N_2$, 
\begin{equation*}
\frac{2CnR_n}{b_n} \leq (0.99\epsilon H_n)^2.  
\end{equation*}
As $\tau_n/H_n \to \rho \in (0,1)$, there is an integer
$N_3>0$ such that for all $n \geq N_3$,
\begin{equation*}
(\rho-\epsilon)H_n \leq \tau_n \leq (\rho+\epsilon) H_n. 
\end{equation*}
Set $N=\max\{N_1,N_2,N_3\}$, and take $n \geq N$. Note that
$\epsilon \leq \rho-\epsilon$ and $\rho + \epsilon
\leq 1-\epsilon$ by construction, and thus by the last two displays, 
\begin{equation}
\label{eq:tau_bds}
\sqrt{\frac{2CnR_n}{b_n}} < \tau_n <
H_n - \sqrt{\frac{2CnR_n}{b_n}}.
\end{equation}

Now observe   
\begin{equation}
\label{eq:screen_dists}
\P\Big( d_H\big(S_F(\ttheta), S_0\big) > b_n \Big) \leq 
\P\Big( d\big(S_F(\ttheta) \,|\, S_0\big) > b_n \Big) +
\P\Big( d\big(S_0 \,|\, S_F(\ttheta) \big) > b_n \Big). 
\end{equation}
We focus on bounding each term on the right-hand side above
separately. For the first term on the right-hand side in
\eqref{eq:screen_dists},
observe that if \smash{$F_{t_i}(\ttheta) \geq \tau_n$} for all $t_i \in S_0$,
then \smash{$d(S_F(\ttheta) \,|\, S_0) \leq b_n$}. By the contrapositive,
\begin{align}
\nonumber
\P\Big( d\big( S_F(\ttheta) \,|\, S_0 \big) > b_n \Big) &\leq
\P\Big( |F_{t_i}(\ttheta)| < \tau_n \;\,\text{for some $t_i \in S_0$}
\Big) \\
\label{eq:good_jumps}
&\leq \P\bigg( |F_{t_i}(\ttheta)| < H_n - \sqrt{\frac{2CnR_n}{b_n}}
\;\,\text{for some $t_i \in S_0$} \bigg), 
\end{align}
where in the second line we used the upper bound on $\tau_n$ in
\eqref{eq:tau_bds}. Suppose that 
\smash{$\|\ttheta-\theta_0\|_n^2 \leq CR_n$};
then, for $t_i \in S_0$,
\Fref{lem:how_small} tells us how small
\smash{$|F_{t_i}(\ttheta)|$} can be made with this error bound in
place. Specifically, define
\begin{equation*}
a = (\underbrace{-1/b_n,\ldots,-1/b_n}_{\text{$b_n$ times}}, 
\underbrace{1/b_n,\ldots,1/b_n}_{\text{$b_n$ times}})
\quad \text{and} \quad
c = (\theta_{0,t_i-b_n+1},\ldots,\theta_{0,t_i+b_n}), 
\end{equation*}
and also \smash{$r=\sqrt{CnR_n}$}.
Then
\Fref{lem:how_small} implies the following: 
if \smash{$\|\ttheta-\theta_0\|_n^2 \leq CR_n$}, then
\begin{equation*}
|F_{t_i}(\ttheta)| \geq |a^\top c| - r\|a\|_2 \geq 
|\theta_{0,t_i+1}-\theta_{0,t_i}| -
\sqrt{\frac{2CnR_n}{b_n}} \geq H_n - \sqrt{\frac{2CnR_n}{b_n}}. 
\end{equation*}
Therefore, continuing on from \eqref{eq:good_jumps},
\begin{align*}
\P\Big( d\big( S_F(\ttheta) \,|\, S_0 \big) > b_n \Big) &\leq 
\P\bigg( |F_{t_i}(\ttheta)| < H_n - \sqrt{\frac{2CnR_n}{b_n}} 
\;\,\text{for some $t_i \in S_0$} \bigg) \\
&\leq \P\Big( \|\ttheta-\theta_0\|_n^2 > C R_n\Big) \\
& \leq \frac{\epsilon}{2}.
\end{align*}

It suffices to consider the second term in \eqref{eq:screen_dists}, 
and show that this is also bounded by $\epsilon/2$.  Note that 
\begin{align}
\nonumber
\P\Big( d\big(S_0 \,|\, S_F(\ttheta) \big) > b_n \Big) &\leq 
\P\bigg( |F_i(\ttheta)| \geq \tau_n \;\,\text{at some $i$ such   
that $\theta_{0,i-b_n+1}=\ldots=\theta_{0,i+b_n}$} \bigg) \\ 
\label{eq:bad_jumps}
&\leq \P\bigg( |F_i(\ttheta)| > \sqrt{\frac{2CnR_n}{b_n}}
\;\, \text{at some $i$ such that 
$\theta_{0,i-b_n+1}=\ldots=\theta_{0,i+b_n}$} \bigg).
\end{align}
In the second inequality we used the lower bound on $\tau_n$ in 
\eqref{eq:tau_bds}. Similar to the previous argument,
suppose that \smash{$\|\ttheta-\theta_0\|_n^2 \leq CR_n$}; for any
location $i$ in
consideration in \eqref{eq:bad_jumps}, \Fref{lem:how_big} tells
us how large \smash{$|F_i(\ttheta)|$} can be made with this error
bound in place.  Defining
\begin{equation*}
a = (\underbrace{-1/b_n,\ldots,-1/b_n}_{\text{$b_n$ times}}, 
\underbrace{1/b_n,\ldots,1/b_n}_{\text{$b_n$ times}})
\quad \text{and} \quad
c = (\theta_{0,i-b_n+1},\ldots,\theta_{0,i+b_n}), 
\end{equation*}
and \smash{$r=\sqrt{CnR_n}$}, as before, the lemma says the following: 
if \smash{$\|\ttheta-\theta_0\|_n^2 \leq CR_n$}, then 
\begin{equation*}
|F_i(\ttheta)| \leq |a^\top c| + r\|a\|_2 = 
\sqrt{\frac{2CnR_n}{b_n}}.
\end{equation*}
Hence, continuing on from \eqref{eq:bad_jumps},
\begin{align*}
\P\Big(d\big(S_0 \,|\, S_F(\ttheta) \big) > b_n \Big) 
&\leq \P\bigg( |F_i(\ttheta)| > \sqrt{\frac{2CnR_n}{b_n}}
\;\, \text{at some $i$ such that 
 $\theta_{0,i-b_n+1}=\ldots=\theta_{0,i+b_n}$} \bigg) \\
&\leq \P\Big( \|\ttheta-\theta_0\|_n^2 > C R_n\Big) \\
& \leq \frac{\epsilon}{2},
\end{align*}
completing the proof.
\end{proof}

\begin{remark}[\textbf{Comparison to
    \Fref{thm:changepoint_screening}}] 
Though they are stated differently, the rates in Theorems
\ref{thm:changepoint_screening} and \ref{thm:changepoint_recovery} for
approximate changepoint screening and recovery, respectively, are
comparable. To see this, note that the conclusion in the latter
implies 
\begin{equation*}
\P\Big( d\big(S(\ttheta) \,|\, S_0 \big) \leq d_n \Big) \to 1 
\quad \text{as $n \to \infty$},
\end{equation*}
for any sequence \smash{$d_n=\omega(nR_n/H_n^2)$}, which is in
line with \Fref{thm:changepoint_recovery}.  (The original conclusion
that \smash{$d(S(\ttheta) \,|\, S_0) = O_\P(nR_n/H_n^2)$} is a 
somewhat stronger statement, though the difference is not major.) 
\end{remark}

\begin{remark}[\textbf{Generic setting: no particular assumptions on  
  data model, or estimator}]
To emphasize a similar point to that in \Fref{rem:no_model},
\Fref{thm:changepoint_recovery} does not use a specific data model,
and considers any estimator \smash{$\ttheta$} for which we have  
$\ell_2$ error control,
\smash{$\|\ttheta-\theta_0\|_n^2=O_\P(R_n)$}.  
This makes it a very flexible and broadly applicable result.  
When the data comes from a model as in \eqref{eq:model}, where
$\theta_0$ is the mean and we have i.i.d.\ sub-Gaussian errors, we 
can apply \Fref{thm:changepoint_recovery} to the fused lasso, given
our knowledge of its $\ell_2$ error rate 
(see Corollaries \ref{cor:changepoint_recovery_weak} and
\ref{cor:changepoint_recovery_strong} below).   
It could also be applied, under the same data model, to many
other estimators whose $\ell_2$ error rates are known (such as the
Potts estimator, and unbalanced Haar wavelets).  Moreover, it could be 
useful under different data models, like the Poisson model 
in \eqref{eq:model_pois}, as it would provide
approximate recovery guarantees for any method with a fast 
enough $\ell_2$ estimation error rate. (Note
that the post-processing step using the filter \eqref{eq:filter}
does not itself require assumptions about the data.)  We do not
consider such extensions in the current paper, but they suggest
interesting directions for future work.     
\end{remark}

The proof of \Fref{thm:changepoint_recovery} relied on two lemmas, 
that we state below.  Their proofs are based on simple arguments in
convex analysis and deferred until Appendix \ref{app:how_big_small}. 

\begin{lemma}
\label{lem:how_big}
Given $a,c \in \R^m$, $r \geq 0$, the optimal value of
the (nonconvex) optimization problem
\begin{equation}
\label{eq:how_big}
\max_{x \in \R^m} \; |a^\top x| \;\;\st\;\; \|x-c\|_2 \leq r 
\end{equation}
is $|a^\top c| + r\|a\|_2$. 
\end{lemma}

\begin{lemma}
\label{lem:how_small}
Given $a,c \in \R^m$, $r \geq 0$ such that
$|a^\top c| - r\|a\|_2\geq 0$, the optimal value of 
the (convex) optimization problem
\begin{equation}
\label{eq:how_small}
\min_{x \in \R^n} \; |a^\top x| \;\;\st\;\; \|x-c\|_2 \leq r
\end{equation}
is $|a^\top c|-r\|a\|_2$. 
\end{lemma}

We finish this subsection with two corollaries of
\Fref{thm:changepoint_recovery} for the fused lasso estimator, in the
weak and strong sparsity cases.  The proofs are immediate and are thus
omitted. 

\begin{corollary}[\textbf{Fused lasso recovery result, weak 
  sparsity setting}]  
\label{cor:changepoint_recovery_weak} 
Assume the conditions of \Fref{thm:weak_sparsity}, 
so that $\TV(\theta_0) \leq C_n$ for a sequence $C_n$. 
Let \smash{$\htheta$} be the fused lasso estimate in
\eqref{eq:fl}, with \smash{$\lambda = \Theta(n^{1/3}C_n^{-1/3})$}, and 
consider applying the filter in \eqref{eq:filter} to \smash{$\htheta$},
as described in \Fref{thm:changepoint_recovery}, to produce a filtered
set \smash{$\hS_F=S_F(\htheta)$}. If the bandwidth and threshold
satisfy \smash{$b_n = \lfloor n^{1/3} C_n^{2/3} \nu_n^2 / H_n^2
  \rfloor \leq W_n/2$} for a sequence $\nu_n \to \infty$, and
$\tau_n/H_n \to \rho \in (0,1)$, then  
\begin{equation*}
\P\bigg( d_H(\hS_F, S_0) \leq 
\frac{n^{1/3} C_n^{2/3} \nu_n^2}{H_n^2}\bigg) \to 1
\quad \text{as $n \to \infty$}.
\end{equation*}
\end{corollary}

\begin{corollary}
[\textbf{Fused lasso recovery result, strong 
sparsity setting}]  
\label{cor:changepoint_recovery_strong}
Assume the conditions of \Fref{thm:strong_sparsity}, 
so that $s_0 = O(1)$ and $W_n=\Theta(n)$.
Let \smash{$\htheta$} denote the fused lasso estimate in
\eqref{eq:fl}, with \smash{$\lambda = \Theta(\sqrt{n})$}, and 
consider applying the filter in \eqref{eq:filter} to
\smash{$\htheta$}, as in \Fref{thm:changepoint_recovery}, to
produce a filtered set \smash{$\hS_F=S_F(\htheta)$}.   
If the bandwidth and threshold values satisfy
{$b_n = \lfloor (\log{n}\log\log{n}) \nu_n^2 / H_n^2 \rfloor \leq
  W_n/2$} for a  sequence $\nu_n \to \infty$, and $\tau_n/H_n \to \rho
\in (0,1)$, then  
\begin{equation*}
\P\bigg( d_H(\hS_F, S_0) \leq 
\frac{(\log{n} \log\log{n}) \nu_n^2}{H_n^2} \bigg) \to 1
\quad \text{as $n \to \infty$}.
\end{equation*}
\end{corollary}

\begin{remark}[\textbf{Recovery under weak sparsity, comparison to
    BS}]  
\label{rem:recovery_weak}
\Fref{cor:changepoint_recovery_weak} considers a challenging
setting in which the number of 
changepoints $s_0$ in $\theta_0$ could be growing quickly with $n$,
and the only control that we have is $\TV(\theta_0) \leq
C_n$. We draw a comparison here to known results on
binary segmentation (BS).  
\Fref{cor:changepoint_recovery_weak} on the (filtered)
fused lasso and Theorem 3.1 in \citet{fryzlewicz2014wild} on the BS  
estimator \smash{$\htheta^{\mathrm{BS}}$}, 
each under appropriate conditions on $W_n,H_n$, state that 
\begin{equation}
\label{eq:fl_bs_weak}
d_H(\hS_F, S_0) \leq 
\frac{n^{1/3} C_n^{2/3} \log{n}}{H_n^2}
\quad \text{versus} \quad
d_H\big(S(\htheta^{\mathrm{BS}}), S_0\big) \leq  
c \frac{n \log{n}}{H_n^2}
\quad \text{respectively},
\end{equation}
where $c>0$ is a 
constant, and both bounds hold with probability approaching 1.
The result on \smash{$\hS_F$} is obtained by 
choosing \smash{$\nu_n=\sqrt{\log{n}}$}
and then \smash{$b_n = \lfloor n^{1/3} C_n^{2/3} \log{n} / H_n^2
  \rfloor$} in \Fref{cor:changepoint_recovery_weak}.  Examining 
\eqref{eq:fl_bs_weak}, we see that, when $C_n$
scales more slowly than $n$, \Fref{thm:changepoint_recovery} provides
the stronger result: the term \smash{$n^{1/3} C_n^{2/3}$} will 
be smaller than $n$, and thus the bound on \smash{$d_H(\hS_F, S_0)$}  
will be sharper than that on 
\smash{$d_H(S(\htheta^{\mathrm{BS}}), S_0)$}.  

But we must also examine the specific restrictions that each result in 
\eqref{eq:fl_bs_weak} places on $s_0,W_n,H_n$.  Consider the
simplification \smash{$W_n=\Theta(n/s_0)$}, corresponding to a
case in which the changepoints in $\theta_0$ are spaced evenly apart.   
For \Fref{cor:changepoint_recovery_weak}, starting with the
condition \smash{$n^{1/3} C_n^{2/3} \log{n} / H_n^2 \leq W_n/2$}, 
plugging in the relationship $C_n \geq s_0H_n$, and rearranging
to derive a lower bound on the minimum signal gap, gives 
\smash{$H_n=\Omega( s_0^{5/4} n^{-1/2} \log^{3/4}{n})$}.  If 
$s_0=\Theta(n^{2/5})$, then we see that the minimum signal gap 
requirement becomes \smash{$H_n=\Omega(\log^{3/4}{n})$}, 
which is growing with $n$ and is thus too stringent to be 
interesting (recall, as discussed previously, that 
\citet{sharpnack2012sparsistency} showed simple thresholding of 
pairwise differences achieves perfect recovery when
\smash{$H_n=\omega(\sqrt{\log{n}})$}).  Hence, to
accommodate signals for which $H_n$ 
remains constant or even
shrinks with $n$, we must restrict the number of jumps in $\theta_0$
according to \smash{$s_0 = O(n^{2/5-\delta})$}, for any fixed
$\delta>0$.  Meanwhile, inspection of Assumption 3.2 in
\citet{fryzlewicz2014wild} reveals that his Theorem 3.1 
requires \smash{$s_0 = O(n^{1/4-\delta})$}, for any $\delta>0$,  
in order to handle signals such that $H_n$ remains constant or 
shrinks with $n$.  In short, the (effectively) allowable range for
$s_0$ is larger for \Fref{thm:changepoint_recovery} than for Theorem
3.1 in \citet{fryzlewicz2014wild}. 
Even when we look within their common 
range, \Fref{thm:changepoint_recovery} places weaker
conditions on $H_n$.  As an example, consider
\smash{$s_0=\Theta(n^{1/6})$} and \smash{$W_n=\Theta(n^{5/6})$}.  The
fused lasso result in \eqref{eq:fl_bs_weak} requires
\smash{$H_n=\Omega(n^{-7/24} \log^{4/3}{n})$},
and the BS result in \eqref{eq:fl_bs_weak} requires 
\smash{$H_n=\Omega(n^{-1/6+\delta})$}, for any $\delta>0$.  Finally,
to reiterate, the fused lasso result in \eqref{eq:fl_bs_weak}
gives a better Hausdorff recovery bound  
when $C_n$ is small compared to $n$; at the extreme end, this is
better by a full factor of \smash{$n^{2/3}$}, when $C_n=O(1)$.  

While the post-processed fused lasso looks favorable 
compared to BS, based on its approximate changepoint recovery 
properties in the weak sparsity setting, we must be clear that the
analyses for other methods---wild binary segmentation (WBS),   
the simultaneous multiscale changepoint estimator (SMUCE), 
and tail-greedy unbiased Haar (TGUH) wavelets---are still much
stronger in this setting.  Such methods have Hausdorff recovery bounds
that are only possible for the post-processed fused lasso (at least,
using our current analysis technique) when we assume strong
sparsity. This is discussed next.   
\end{remark}

\begin{remark}[\textbf{Recovery under strong sparsity, comparison to
    other methods}] 
\label{rem:recovery_strong}
When $s=O(1)$ and $W_n=\Theta(n)$,
\Fref{cor:changepoint_screening_strong} shows that the post-processed
fused lasso estimator delivers a Hausdorff bound of  
\begin{equation}
\label{eq:fl_strong}
d_H(\hS_F, S_0) \leq  
\frac{\log^2{n}}{H_n^2},
\end{equation}
on the set \smash{$\hS_F$} of filtered changepoints, with probability
approaching 1.  This is obtained by choosing (say)
\smash{$\nu_n=\sqrt{\log{n}/\log\log{n}}$} 
and \smash{$b_n = \lfloor \log^2{n}/H_n^2 \rfloor \leq W_n/2  $} in 
the corollary. The effective restriction on the minimum signal gap is
thus \smash{$H_n=\Omega(\log{n}/\sqrt{n})$}, which is quite
reasonable, as \smash{$H_n=\omega(1/\sqrt{n})$} is needed for any
method to detect a changepoint with probability tending to 1 (recall 
\Fref{rem:Hmin}).  Several other methods---the Potts estimator
\citep{boysen2009consistencies}, binary segmentation (BS) and wild
binary segmentation (WBS) \citep{fryzlewicz2014wild}, the simultaneous
multiscale changepoint estimator (SMUCE)
\citep{frick2014multiscale}, and tail-greedy unbiased Haar wavelets
(TGUH) \citep{fryzlewicz2016tail}---all admit Hausdorff recovery
bounds that essentially match \eqref{eq:fl_strong}, under similarly
weak restrictions on $H_n$.  But, it should be noted that the
latter three methods---WBS, SMUCE, and TGUH---continue to enjoy these 
same sharp Hausdorff bounds {\it outside of} the strong sparsity
setting, namely, their analyses do not require that $s_0=O(1)$ and 
$W_n=\Theta(n)$, and instead just place weak restrictions on the
allowed combinations of $W_n,H_n$ (e.g., the analysis of WBS in
\citet{fryzlewicz2014wild} only requires \smash{$W_n H_n^2 \geq
  \log{n}$}).  These analyses (and those for all previously
described estimators) are more refined than that given 
in \Fref{cor:changepoint_recovery_strong}: they are based on specific 
properties of the estimator in question.  The corollary, on the other
hand, follows from \Fref{thm:changepoint_screening}, which uses a 
completely generic analysis that only assumes knowledge of the
$\ell_2$ error rate.     
\end{remark}

\subsection{Post-processing on a reduced set}
\label{sec:changepoint_recovery_reduced}

Recall that 
the strategy studied in \Fref{thm:changepoint_recovery} was to apply
the Haar filter in \eqref{eq:filter} at each location
$i=b_n,\ldots,n-b_n$ and then check for large absolute values.
Computationally, this 
not expensive---it only requires $O(n)$ operations---but there is an
undesirable feature of this strategy with respect to practical 
usage. Writing the original number of estimated changepoints as 
\smash{$\tilde{s}=|S(\ttheta)|$}, it is possible in practice for the
size of the filtered set \smash{$S_F(\ttheta)$} in \eqref{eq:filtered}
to be much larger than \smash{$\tilde{s}$}, if the bandwidth and
threshold parameters are not set appropriately.
Indeed, as the filter
is being applied at $n-2b_n$ locations, it is possible for the
filtered set to have precisely this many elements.  

Here we propose a modified strategy that runs the
filter on (at most) \smash{$3\tilde{s}+2$} changepoints, and then as 
usual, keeps only changepoints whose absolute filter values are
large. This modified strategy has essentially same the theoretical
guarantee 
of approximate changepoint recovery as the original ``exhaustive'' 
strategy from \Fref{sec:changepoint_recovery}, but enjoys the practical 
advantage that, no matter how the bandwidth and threshold parameters
are chosen, the final set of detected changepoints is bounded in size
by 3 times the number of changepoints in \smash{$\ttheta$} (plus 2, to
be precise).  Before stating the main result of this subsection, we
introduce a ``candidate'' set for locations for changepoints, 
\begin{equation}
\label{eq:candidates}
I_C(\ttheta) = \Big\{ i \in \{b_n, \ldots, n-b_n\} : i \in
S(\ttheta), \;\text{or}\;\, i+b_n \in S(\ttheta),  
\;\text{or}\;\, i-b_n \in S(\ttheta) \Big\} \cup \{b_n, n-b_n\}.
\end{equation}
These are estimated changepoints, locations that are at a
distance $b_n$ from estimated changepoints, or boundary points.

\begin{theorem}[\textbf{Generic recovery result, reduced
  post-processing}] 
\label{thm:changepoint_recovery_reduced}
Assume the conditions of \Fref{thm:changepoint_recovery}, but consider
a modified strategy in which we only evaluate the filter in
\eqref{eq:filter} at locations in the candidate set
\smash{$I_C(\ttheta)$} in \eqref{eq:candidates}, and define
a ``reduced'' set of filtered points based on the locations whose
filter value is at least $\tau_n$,
\begin{equation}
\label{eq:filtered_reduced}
S_R(\ttheta) = 
\Big\{ i \in I_R(\ttheta) : |F_i(\ttheta)| \geq \tau_n \Big\}.
\end{equation}
Then, subject to the same conditions on $b_n,\tau_n$ as in 
\Fref{thm:changepoint_recovery}, we have 
\begin{equation*}
\P\Big( d_H\big(S_R(\ttheta), S_0 \big) \leq 2b_n \Big) 
\to 1 \quad \text{as $n \to \infty$}.
\end{equation*}
\end{theorem}

\begin{proof}
We will show that 
\begin{equation}
\label{eq:haus_dist_full_to_reduced}
\Big\{d_H\big(S_F(\tilde{\theta}), S_0\big) \leq b_n \Big\} \subseteq 
\Big\{d_H\big(S_R(\tilde{\theta}), S_0\big) \leq 2b_n \Big\},
\end{equation}
Since the left-hand side occurs with probability tending to 1, by 
\Fref{thm:changepoint_recovery}, so will the right-hand side.
To show the desired containment, recall that, by the definition of
Hausdorff distance, 
\begin{equation}
\label{eq:haus_dist_reduced}
\Big\{d_H\big(S_F(\ttheta), S_0\big) \leq b_n \Big\} =
\Big\{ d\big(S_0 \,|\, S_F(\ttheta) \big) \leq b_n \Big\} \cap
\Big\{d\big(S_F(\ttheta) \,|\, S_0\big) \leq b_n\Big\}. 
\end{equation} 
Inspecting the first term on the right-hand side of
\eqref{eq:haus_dist_reduced}, we observe 
\begin{equation}
\label{eq:haus_dist_reduced_lhs}
\Big\{ d\big(S_0 \,|\, S_F(\ttheta) \big) \leq b_n \Big\} \subseteq
\Big\{ d\big(S_0 \,|\, S_F(\ttheta) \big) \leq 2b_n \Big\} \subseteq 
\Big\{ d\big(S_0 \,|\, S_R(\ttheta) \big) \leq 2b_n \Big\},
\end{equation}
where the last containment holds as \smash{$S_R(\ttheta) \subseteq
S_F(\ttheta)$}.  Inspecting the second term on the right-hand side of 
\eqref{eq:haus_dist_reduced}, we use \Fref{lem:local_max} which states
that for each $j\in \{b_n,\ldots, n-b_n\}$, 
there exists \smash{$i\in I_C(\tilde{\theta})$} such that
$|i-j|\leq b_n$ and \smash{$|F_i(\ttheta)| \geq
  |F_j(\ttheta)|$}. Using this, we see
\begin{align}
\Big\{d\big(S_F(\ttheta) \,|\, S_0\big) \leq b_n\Big\} &= \Big\{
\text{for all $\ell \in S_0$, there exists $j \in S_F(\ttheta)$ 
such that $|\ell-j|\leq b_n$} \Big\} \nonumber \\  
&\subseteq 
\Big\{ \text{for all $\ell \in S_0$, there exists $i \in I_C(\ttheta)$   
such that $|\ell-i|\leq 2b_n$} \Big\} \nonumber \\
& = \Big\{ d\big(S_R(\ttheta) \,|\, S_0\big) \leq 2b_n \Big\}.  
\label{eq:haus_dist_reduced_rhs}
\end{align}
Above, we have used \Fref{lem:local_max} for the containment in
the second line. Combining \eqref{eq:haus_dist_reduced}, 
\eqref{eq:haus_dist_reduced_lhs}, and
\eqref{eq:haus_dist_reduced_rhs}, we have established 
\eqref{eq:haus_dist_full_to_reduced}, as desired.
\end{proof}

The proof of \Fref{thm:changepoint_recovery_reduced} relied on the
following lemma.  Its proof can be found in Appendix
\ref{app:local_max}. 

\begin{lemma}
\label{lem:local_max}
Let \smash{$I_C(\tilde{\theta})$} be the candidate set defined in
\eqref{eq:candidates}.  For every location $j \in \{b_n, \ldots,
n-b_n\}$ where \smash{$|F_j(\ttheta)| > 0$}, there exists a location
\smash{$i \in I_C(\ttheta)$} such that $|i-j| \leq b_n$ and
\smash{$|F_i(\ttheta)| \geq |F_j(\ttheta)|$}. 
\end{lemma}

\section{Implementation considerations and experiments} 
\label{sec:implementation}

We develop a data-driven procedure to determine the 
threshold level $\tau_n$ of the filter in \eqref{eq:filter}, used to
derive a post-processed set of changepoints \smash{$S_F(\ttheta)$}
from an estimate \smash{$\ttheta$}, as described in
\eqref{eq:filtered} in \Fref{thm:changepoint_recovery}.  We also
present a number of simulation results to support and complement the
theoretical developments in this paper. 

\paragraph{A data-driven procedure for choosing $\tau_n$.}

Let $\cA(\cdot)$ denote a fitting algorithm that, applied to data
$y$, outputs an estimate \smash{$\ttheta$} of $\theta_0$ (e.g., 
$\cA(y)$ could be the minimizer in \eqref{eq:fl}, so that its output is
the fused lasso estimate).  In \Fref{alg:choosing_tau}, 
we present a heuristic but intuitive method for choosing the threshold 
level $\tau_n$, based on (entrywise) permutations of the residual 
vector  \smash{$y-\ttheta$}.  
Aside from the choice of fitting
algorithm $\cA(\cdot)$, we must specify a number of permutations $B$
to be explored, a bandwidth $b_n$ for the filter in
\eqref{eq:filter}, and a quantile level $q \in (0,1)$. 
The intuition behind \Fref{alg:choosing_tau} is to set $\tau_n$ large
enough to suppress ``false positive'' changepoints $100 \cdot q \%$ of
the time (according to the permutations).  This is revisited later, in
the discussion of the simulation results.

Some example settings: we may choose
$\cA(\cdot)$ to be the fused lasso estimator,
where the tuning parameter $\lambda$ is selected to minimize 5-fold 
cross-validation (CV) error, $B=100$, and $q=0.95$.  The choice of
bandwidth $b_n$ is more subtle, and unfortunately, there is no
specific answer that works for all problems.\footnote{We note that in
some situations, problem-specific intuition can yield a 
reasonable choice of bandwidth $b_n$. Also, it should be possible
to extend \Fref{alg:choosing_tau} to choose both $\tau_n$ and $b_n$, 
but we do not pursue this, for simplicity.} 
But, the theory in the last section provides some
general guidance: e.g., for problems in which we believe 
there are a small number of changepoints (i.e., $s_0=O(1)$) of
reasonably large magnitude (i.e., $H_n=\Omega(1)$), 
\Fref{thm:changepoint_recovery} instructs us to choose a bandwidth
that grows faster than $\log{n}\log\log{n}$, so, choosing $b_n$ to  
scale as \smash{$\log^2{n}$} would suffice.  We will use this scaling,
as well as the above suggestions for $\cA(\cdot)$, $B$, and $q$ in all 
coming experiments, unless otherwise specified.

\begin{algorithm}[htb]
\caption{Permutation-based approach for choosing $\tau_n$} 
\label{alg:choosing_tau}

\begin{enumerate}
\item[0.] Input a fitting algorithm $\cA(\cdot)$, number of
  permutations $B$, bandwidth $b_n$, and quantile level $q \in
  (0,1)$. 

\item Compute \smash{$\ttheta=\cA(y)$}. Let   
  \smash{$\tilde{S}=S(\ttheta)$} denote the changepoints, and 
  \smash{$r = y-\ttheta$} the residuals.
\item For each $b=1,\ldots,B$, repeat the following steps:
\begin{enumerate}
\item Let \smash{$r^{(b)}$} be a randomly-chosen permutation of
  $r$, and define auxiliary data \smash{$y^{(b)} = \ttheta +
    r^{(b)}$}.   

\item Rerun the fitting algorithm on the auxiliary data to yield
  \smash{$\ttheta^{(b)} = \cA(y^{(b)})$}.  

\item Apply the filter in \eqref{eq:filter} to $\ttheta^{(b)}$ (with
  the specified bandwidth $b_n$), and record the largest magnitude
  \smash{$\hat{\tau}^{(b)}$} of the filter values at locations greater
  than $b_n$ away from \smash{$\tilde{S}$}.  Formally, 
  \[
  \hat{\tau}^{(b)} = \max_{\substack{i \in \{b_n,\ldots,n-b_n\} : \\ 
      d(\tilde{S}|\{i\}) > b_n}} \; \big|F_i(\ttheta^{(b)})\big|. 
  \]
\end{enumerate}

\item Output \smash{$\hat\tau_n$}, the level $q$ quantile of the
  collection \smash{$\hat\tau^{(b)}$, $b=1,\ldots,B$}.
\end{enumerate}
\end{algorithm}

After running \Fref{alg:choosing_tau} to compute \smash{$\hat\tau_n$},  
the idea is to proceed with the full filter \smash{$S_F(\ttheta)$} or the
reduced filter \smash{$S_R(\ttheta)$}, applied at the level
\smash{$\tau_n=\hat\tau_n$}, to the estimate  
\smash{$\ttheta$} computed on the original data $y$ at hand. In the
experiments that follow, we use the reduced filter, though
similar conclusions would hold with the full filter.

\paragraph{Simulation setup.}

In our experiments, we use the following simulation setup.
For a given $n$, the mean parameter $\theta_0 \in \R^n$ is  
defined to have $s_0 = 5$ equally-sized segments, with levels
0, 2, 4, 1, 4, from left to right. Data $y \in \R^n$ is
generated around $\theta_0$ using i.i.d.\ $N(0,4)$ noise.  Lastly, the
sample size $n$ is varied between $100$ and $10,000$, equally-spaced
on a log scale.   \Fref{fig:truth} shows example data sets with
$n=774$ and $n=10,000$.      

\begin{figure}[!ht]
\centering
\includegraphics[width=\textwidth]{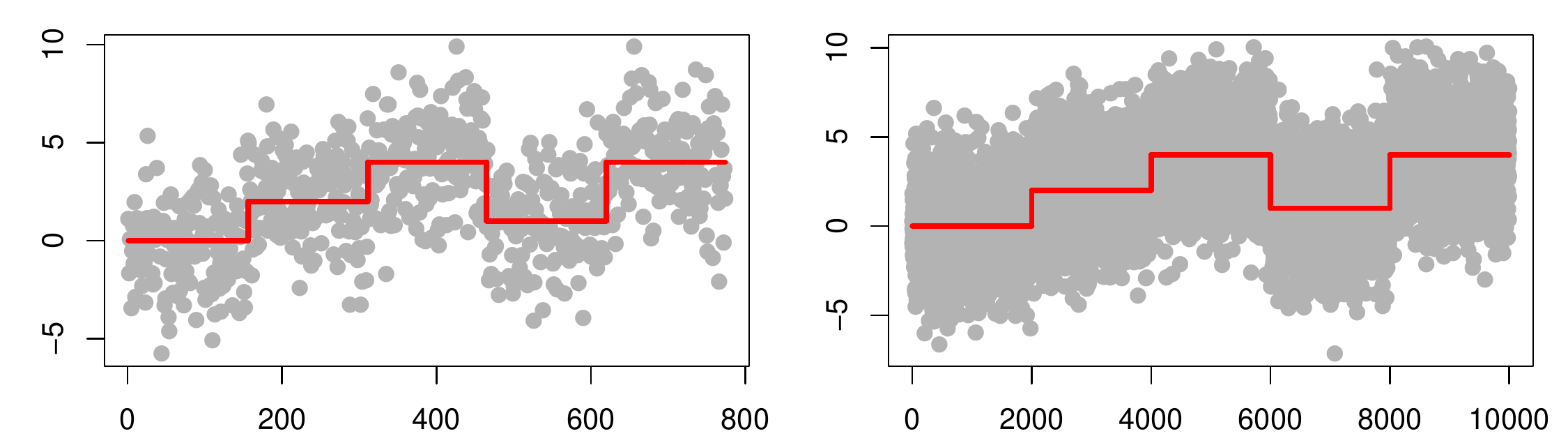} 
\caption{\it\small An example from our simulation setup for $n=774$ (left)
  and $n =10,000$ (right), where in each panel, the mean $\theta_0$ is
  plotted in red, and the data points in gray.}
\label{fig:truth}
\end{figure}

\paragraph{Evaluation of the filter.}

We demonstrate that the filter in \eqref{eq:filter},  
with \smash{$b_n = \lfloor 0.25 \log^2 n \rfloor$}, can be 
effective at reducing the Hausdorff distance between 
estimated and true changepoint sets. We first 
illustrate the use of the filter in a single
data example with $n=774$, in \Fref{fig:filter}.  
As we can see, the fused lasso originally places a spurious jump  
around location 250, but this jump is eliminated when we apply the  
filter, provided that we set the threshold to be (say) $\tau_n=0.5$. 

\begin{figure}[!ht]
\centering
\includegraphics[width=\textwidth]{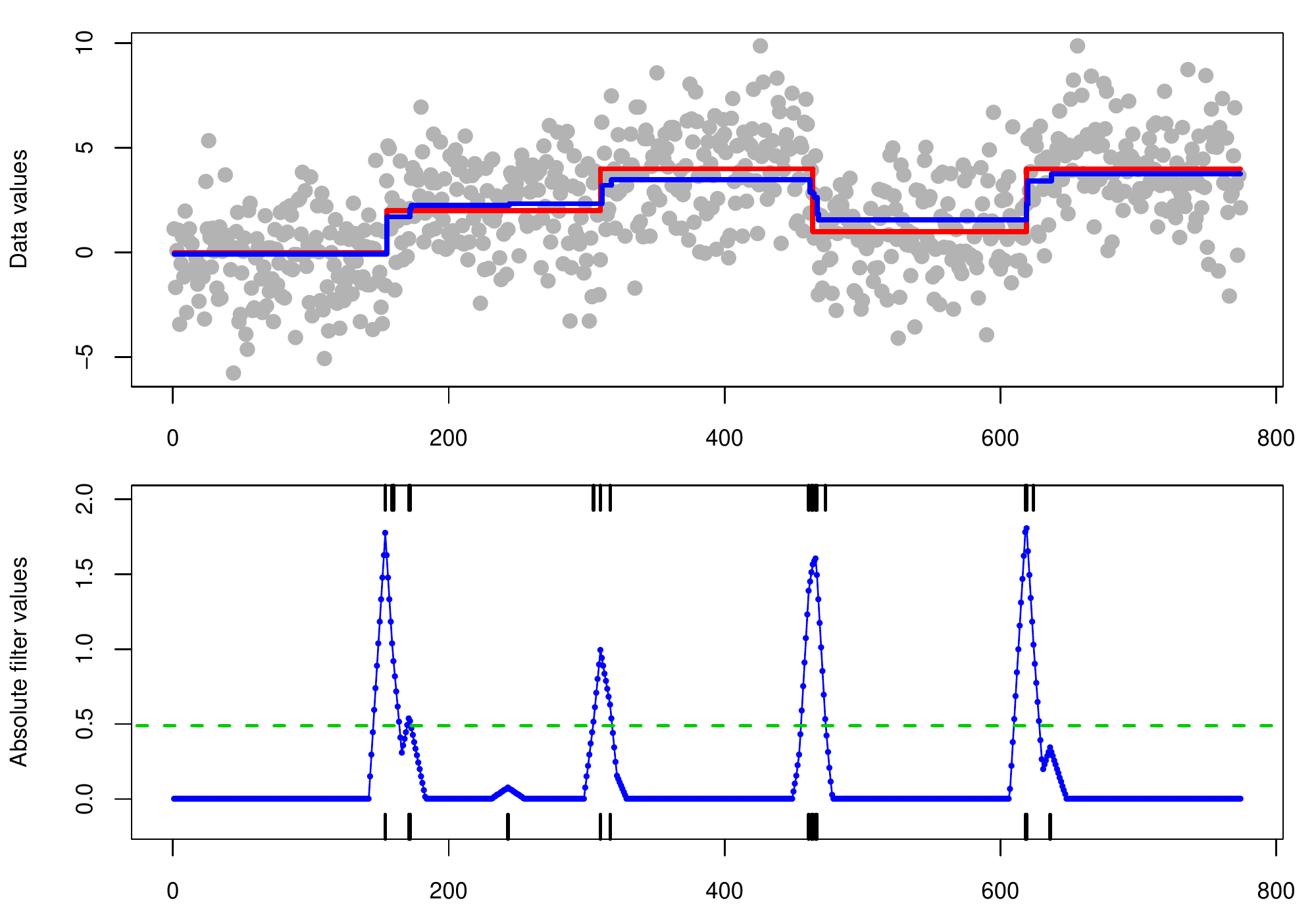}
\caption{\it\small In the top plot, an example with $n=774$ is shown
  from our simulation setup, where the data $y$ is drawn in gray,  
  the mean $\theta_0$ in red, and the fused lasso estimate
  \smash{$\htheta$} in 
  blue.  In the bottom plot, the filter values $F_i(\htheta)$,
  $i=1,\ldots,n$ are drawn in blue, and the threshold $\tau_n$ is
  drawn as a horizontal green line.
  Changepoints before and after filtering are marked by short
  black lines along the bottom and top x-axes, respectively.}
\label{fig:filter}
\end{figure}

\Fref{fig:haus} now reports the results from applying the filter in
problems of sizes between $n=100$ and $n=10,000$, using 50 trials 
for each $n$.  We consider three different sets of changepoint
estimates: \smash{$\hS=S(\htheta)$}, the original changepoints from
fused lasso estimate \smash{$\htheta$} tuned with 5-fold CV tuning;
\smash{$S_R(\htheta)$}, the
changepoints after applying the reduced filter as described in
\Fref{thm:changepoint_recovery_reduced} to \smash{$\htheta$}, with 
$\tau_n$ chosen by \Fref{alg:choosing_tau}; and
\smash{$S_O(\htheta)$}, an oracle set of changepoints given by trying 
a wide range of $\tau_n$ values and choosing the value that 
minimizes the Hausdorff distance after filtering (this assumes
knowledge of $S_0$, and is infeasible in practice).  These
are labeled as ``original'', ``data-driven'', and ``oracle'' in the 
figure, respectively. 
As we can see from the left and middle panels, the Hausdorff distance
achieved by the original changepoint set grows nearly linearly with
$n$, but after applying the reduced filter, the Hausdorff distance
becomes very small, provided that $n$ is larger than 1000 or so.
Empirically, the Hausdorff distance associated with the filtered set
appears to grow very slowly with $n$, nearly constant (slower than the
the $\log n \log \log n$ rate guaranteed by
\Fref{cor:changepoint_recovery_strong}). The right panel shows that
our data-driven choices of $\tau_n$ are not
substantially different from those made by the oracle.

\begin{figure}[!ht]
\centering
\includegraphics[width=\textwidth]{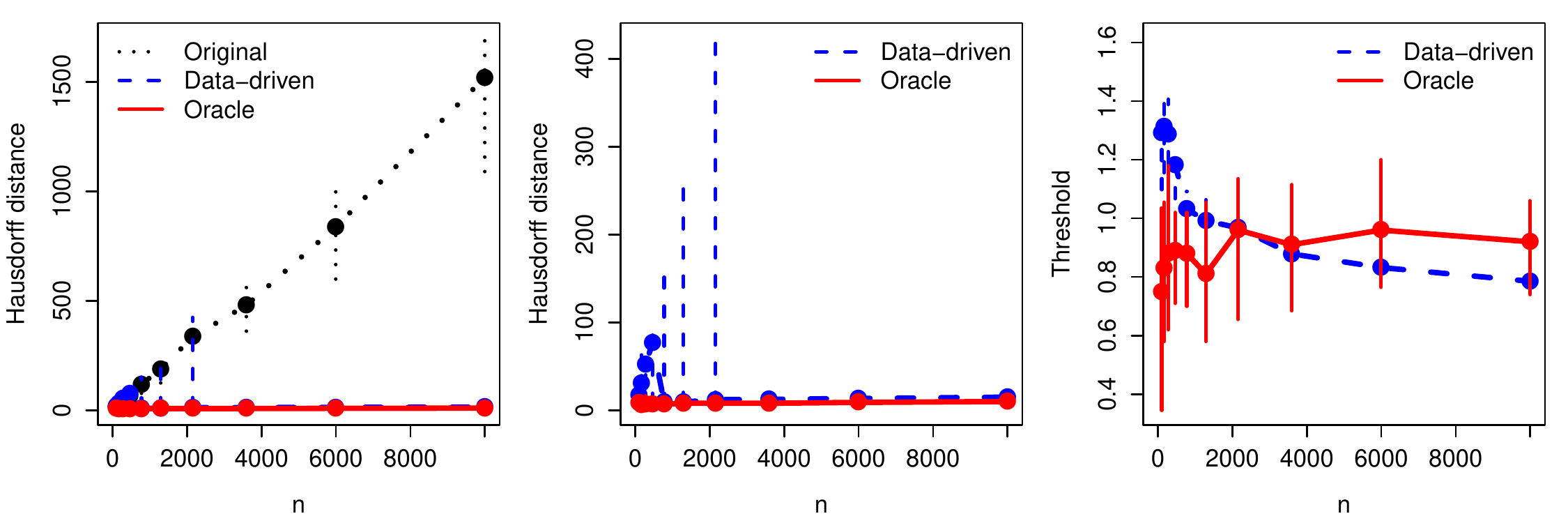}
\caption{\it\small In the left panel, the Hausdorff distances between 
  original changepoints, filtered changepoints with a data-driven
  threshold, and filtered changepoints with an oracle threshold, are
  plotted (in black, blue, and red, respectively).  The results are 
  aggregated across 50 trial runs for each 
  sample size $n$; the solid dots display the median values, and the
  vertical segments display the interquartile ranges (25th
  to 75th percentiles).  The 
  middle panel zooms in on the Hausdorff distances for the
  data-driven and oracle filtering procedures, and the right panel
  displays the choices of $\tau_n$ for these procedures.}
\label{fig:haus}
\end{figure}

\paragraph{Screening distances, false positives.} 

\Fref{fig:thres} examines the outcomes from
varying the filter threshold $\tau_n$ in between 0
and 2, and then applying the reduced filter to produce
\smash{$S_R(\htheta)$}. The results are aggregated over 500 trials
when $n=774$ (i.e., 500 data instances drawn from the simulation
setup), and 
the screening distance \smash{$d(S_R(\htheta) \, | \,  
  S_0)$} and ``precision distance'' \smash{$d(S_0 \, |
  \, S_R(\htheta))$} are plotted with $\tau_n$.  The former
increases with $\tau_n$, and the latter decreases; recall, the
Hausdorff distance is the maximum of the two.  We see that
threshold levels from 0.5 to 1 yield a small Hausdorff distance.    

\begin{figure}[!ht]
\centering
\includegraphics[width=0.6\textwidth]{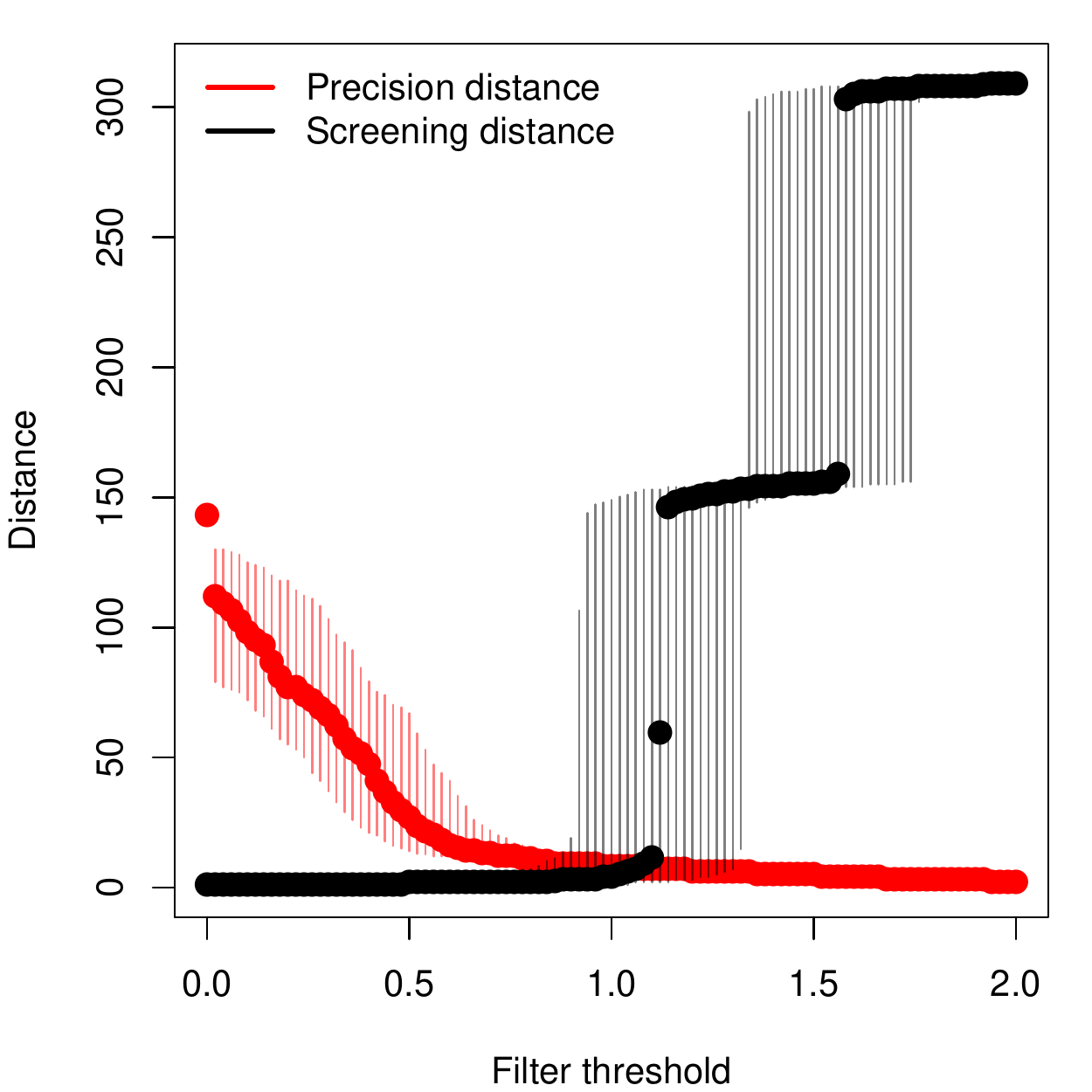}
\caption{\it\small The screening distance \smash{$d(S_R(\htheta) \,|
    \, S_0)$} (in black) and the precision distance \smash{$d(S_0 \,|
    \, S_R(\htheta))$} (in red) are shown as functions of 
  the threshold $\tau_n$ used for the filtered set.  These
  were aggregated over 500 trials, in which $n=774$; the dots display 
  the median values, and the vertical segments are drawn from the 25th    
  to 75th percentiles.  We can see a substantial jump in the
  screening distance once a bit after $\tau_n = 1$ and again after
  $\tau_n = 1.5$, where the median value is close to one of the
  quartiles.  This is due to \smash{$S_R(\htheta)$} suppressing all of
  the estimated changepoints near a particular true changepoint,
  at these critical values of $\tau_n$.}
\label{fig:thres}
\end{figure}

The left panel of \Fref{fig:roc} shows the same results, but with
the screening distance on the x-axis, and the precision distance on the 
y-axis. The red dot marks the screening distance and
precision distance achieved by the
data-driven rule from \Fref{alg:choosing_tau}, using $B=150$ 
permutations. This lies basically at the ``elbow'' of the curve, 
just as we would  
desire. The middle panel of the figure plots the proportion of false
positive detections (out of the 500 repetitions total) on the x-axis,
versus 
the proportion of true positive detections on the y-axis.  Here, note,
we define a false positive detection to be the event that
{\it any} estimated changepoint is more than $b_n$ away from all
true changepoints, or simply, the event that \smash{$d(S_0 \,|\,
  S_R(\htheta)) > b_n$}, and a true positive detection to be the event
that {\it all} true changepoints have estimated changepoints 
at most $b_n$ away, or simply, \smash{$d( S_R(\htheta) \,|\, S_0)
  \leq b_n$}. Therefore, to be perfectly concrete, the x-axis and
y-axis are displaying a certain type of false positive and true  
positive rates (FPR and TPR), defined as 
\begin{equation*}
\mathrm{FPR} = \frac{\text{\# trials in which  
    $d(S_0 \,|\, S_R(\htheta))>b_n$}}{\text{\# of 
    trials}} 
\quad \text{and} \quad
\mathrm{TPR} = \frac{\text{\# trials in which  
    $d(S_R(\htheta) \,|\, S_0 ) \leq b_n$}}{\text{\# of  
    trials}}.
\end{equation*}
The red dot again marks the FPR and TPR achieved by the data-driven
rule in \Fref{alg:choosing_tau} for choosing the threshold, about
0.26 and 0.7, respectively. We might expect here, having set $q=0.95$
in \Fref{alg:choosing_tau}, to see a FPR close to 0.05 (because the
choice of threshold in \Fref{alg:choosing_tau} precisely controls the
FPR at 0.05 over the permutations encountered in the procedure).
However, this is not the case on in our simulation, and the actual FPR
is higher. This phenomenon is not specific to the quantile choice of  
$q=0.95$, as shown in the right panel of \Fref{fig:roc}.  For a
varying quantile level $q$ in between 0 and 1, we ran
\Fref{alg:choosing_tau}, used the corresponding threshold for our
filter, and measured the FPR achieved by the filtered changepoint
set.  As we can see, the actual FPR is generally higher than $1-q$.  

\begin{figure}[!ht]
\centering
\includegraphics[width=\textwidth]{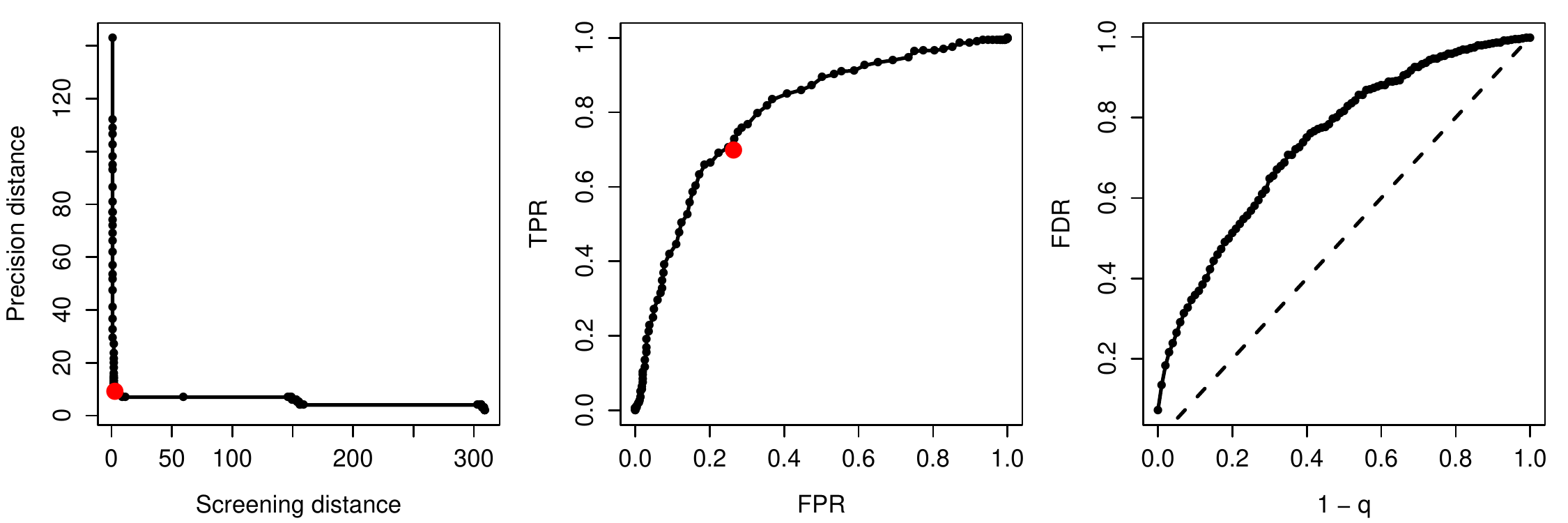}
\caption{\it\small The left panel plots the precision distance
  \smash{$d(S_0 \,| \, S_R(\htheta))$} and the
  screening distance \smash{$d(S_R(\htheta) \,| \, S_0)$}, as the
  threshold $\tau_n$ is varied from 0 to 2.  Shown are
  the median distances over 500 trials for a problem with $n=774$.
  The red dot marks the precision and screening distances achieved by
  the data-driven threshold level chosen by \Fref{alg:choosing_tau}
  with $q=0.95$.  The middle panel shows the same, but with true
  positive rate (TPR) against false positive rate (FPR).  The right
  panel shows the achieved FPR against $1-q$, as the input quantile
  level $q$ is varied in \Fref{alg:choosing_tau}.}
\label{fig:roc}
\end{figure}

\paragraph{Fused lasso fast $\ell_2$ error rate,
  under strong sparsity.}  

We finish by examining the (squared) $\ell_2$ error 
\smash{$\|\htheta - \theta_0\|_n^2$} as it scales with $n$, when the 
fused lasso estimate \smash{$\htheta$} in \eqref{eq:fl} is tuned
appropriately.  
For different sample sizes ranging from $n=100$ to $n=10,000$, we 
generated 50 example data sets from the same setup described
previously, and on 
each data set, computed the fused lasso estimate \smash{$\htheta$}
with 5-fold CV to select the tuning parameter $\lambda$.  
\Fref{fig:lambdamse} reports the median value of $\lambda$, and the
median achieved $\ell_2$ error rate
\smash{$\|\htheta-\theta_0\|_n^2$}, over the 50 trials, as functions
of $n$.  The results support the theoretical conclusion in
\Fref{thm:strong_james}, as the achieved $\ell_2$ error rate scales
at about the rate $(\log{n}\log\log{n})/n$. Also, since $s_0 = O(1)$,
the results support the theoretical result that  $\lambda$ scales
 with $\sqrt{n}$.

\begin{figure}[!h]
\centering
\includegraphics[width=\textwidth]{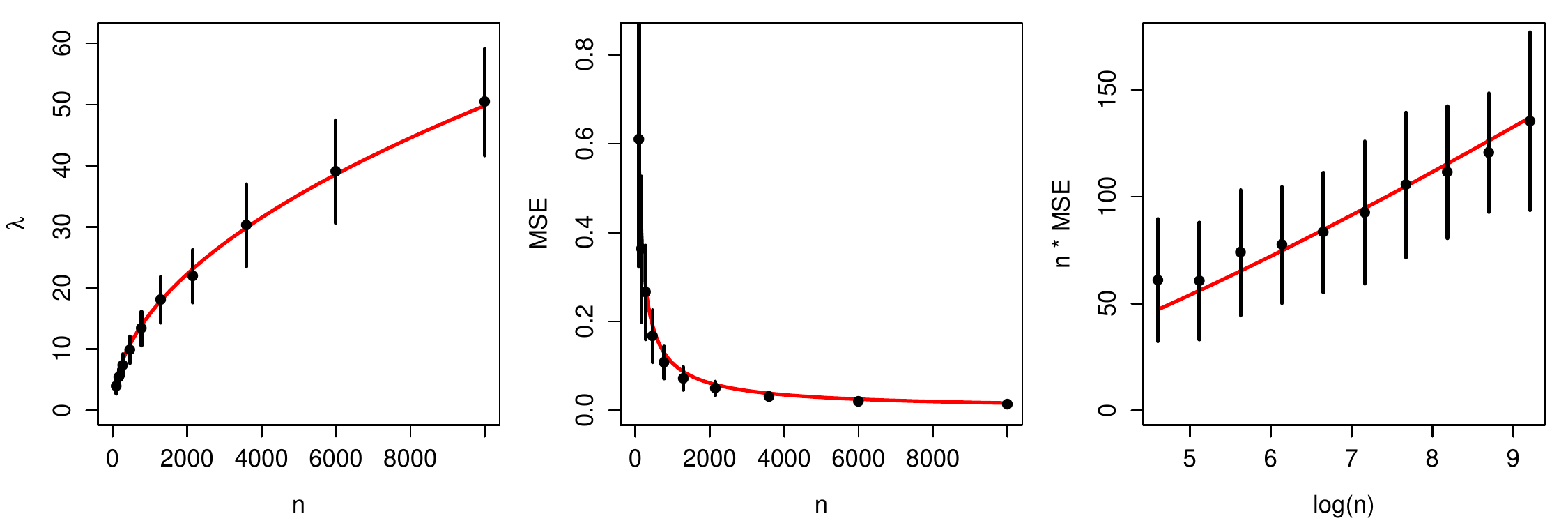}
\caption{\small\it The left panel shows the value of $\lambda$
  chosen to minimize 5-fold CV error over the fused lasso path,
  aggregated over repetitions in our simulation setup, as the sample 
  size $n$ varies.  This scales approximately as 
  \smash{$\sqrt{n}$}, which is drawn as  
  a red curve (with a best-fitting constant).  The middle panel shows
  the corresponding (squared) $\ell_2$ estimation error  
  \smash{$\|\htheta-\theta_0\|_n^2$}, again aggregated over repetitions,  
  as $n$ varies.  The scaling is 
  about \smash{$(\log{n}\log\log{n})/n$} (red curve).
  The right panel plots the median achieved values of \smash{$n   
    \|\htheta-\theta_0\|_n^2$} against $\log{n}$; this looks close to
  linear (red line), which provides empirical  
  support to the claim that the fused lasso error rate is indeed
  \smash{$(\log{n}\log\log{n})/n$} instead of \smash{$\log^2{n}/n$} 
  (as this would have appeared as a quadratic trend in the right
  panel). In each panel, the vertical bars denote $\pm 1$ standard 
  deviations.}   
\label{fig:lambdamse}
\end{figure}

\section{Extensions} 
\label{sec:extensions}

We study screening properties that are implied by
$\ell_2$ error properties in two related problems: first, piecewise 
linear segmentation, and then, segmentation on graphs.

\subsection{Piecewise linear segmentation}
\label{sec:linear}

We now consider data from a model as in \eqref{eq:model} but where
\smash{$\theta_{0,i}$}, $i=1,\ldots,n$ is a piecewise linear (rather 
than a piecewise constant) sequence.  The main estimator of interest 
is {\it linear trend filtering}
\citep{steidl2006splines,kim2009trend,tibshirani2014adaptive},
which can be seen as an extension of the fused lasso that penalizes
second-order (rather than first-order) differences:
\begin{equation}
\label{eq:tf}
\htheta = \argmin_{\theta \in \R^n} \; 
\half \sum_{i=1}^n (y_i-\theta_i)^2 + 
\lambda \sum_{i=1}^{n-2} |\theta_i-2\theta_{i+1}+\theta_{i+2}|,  
\end{equation}
for a tuning parameter $\lambda \geq 0$.  Several other estimators are
available in the piecewise linear segmentation problem, but given its
ties to the fused lasso (and our focus on the fused lasso thus far),
we focus on linear trend filtering in particular.  

In terms of detection, we are 
now interested in the locations of nonzero second-order
differences, i.e., the ``knots'', which mark the changes in slope
across the entries of a parameter $\theta \in \R^n$:
\begin{equation*}
S_2(\theta) =  \bigg\{ i \in \{2,\ldots,n-1\} : \theta_i \not= 
\frac{\theta_{i-1}+\theta_{i+1}}{2} \bigg\}. 
\end{equation*}
We use the abbreviations $S_{0,2}=S_2(\theta_0)$ and
\smash{$\hS_2=S_2(\htheta)$}. We again write $S_{0,2} = 
\{t_1,\ldots,t_{s_0}\}$, where \smash{$2 \leq t_1 < \ldots < t_{s_0} <
  n$} and $s_0=|S_{0,2}|$, and for convenience $t_0=0$,
$t_{s_0+1}=n$.  We also carry forward analogous definitions for
$W_n,H_n$: 
\begin{equation}
\label{eq:WnHn_linear}
W_n = \min_{i=0,1\ldots,s_0} \, (t_{i+1} - t_i)
\quad \text{and} \quad
H_n= \min_{i \in S_{0,2}} \; |\theta_{0,i-1} - 2\theta_{0,i} +
\theta_{0,i+1}|. 
\end{equation}
Lastly, we define the discrete second-order total variation operator,
acting on a vector $x \in \R^n$, by
\begin{equation*}
\TV_2(x) = \sum_{i=2}^{n-1} |x_{i-1} - 2x_i + x_{i+1}|.
\end{equation*}

The following describes the $\ell_2$ estimation error of 
linear trend filtering, under weak sparsity.

\begin{theorem}[\textbf{Trend filtering error rate, weak sparsity
  setting, Theorem 10 of \citealt{mammen1997locally}}]     
\label{thm:weak_sparsity_linear}
Assume the data model in \eqref{eq:model}, with errors
$\epsilon_i$, $i=1,\ldots,n$ i.i.d.\ from a sub-Gaussian
distribution as in \eqref{eq:subgauss}.  Also assume that 
\smash{$\TV_2(\theta_0) \leq C_n$}, for a sequence
$C_n$. Then for 
\smash{$\lambda=\Theta(n^{1/5}C_n^{-3/5})$}, the linear trend
filtering estimate \smash{$\htheta$} in \eqref{eq:tf} satisfies 
\begin{equation*}
\|\htheta - \theta_0\|_n^2 = O_\P (n^{-4/5} C_n^{2/5}). 
\end{equation*}
\end{theorem}

\begin{remark}[\textbf{Consistency, optimality}]
The lemma shows that linear trend filtering is consistent when
$C_n=o(n^2)$. When $C_n=O(1)$, its error rate is 
$n^{-4/5}$, which is in fact minimax optimal as $\theta_0$  
varies over the class of signals having bounded second-order total 
variation, i.e., $\theta_0 \in \{\theta \in \R^n: \TV_2(\theta) \leq 
C\}$ for a constant $C>0$ \citep{donoho1998minimax}.  As in the fused
lasso case, we refer the reader to \citet{tibshirani2014adaptive} for
explanations of the above theorem and this minimax result, in
notation that is more consistent with that of the current paper.
\end{remark}

\begin{remark}[\textbf{Strong sparsity, higher polynomial degrees}]
Results for linear trend filtering in the strong sparsity setting,
i.e., one in which $s_0$ is assumed to be bounded (so that we are
estimating a piecewise linear function with few knots) are not
currently available, to the best of our knowledge.  However, we
suspect that the achieved error rate here will be close to the
``parametric'' $1/n$ rate, as in Theorems \ref{thm:strong_sparsity}  
and \ref{thm:strong_james}, on the fused lasso.  It is also worth
noting that the extension of 
trend filtering to fit piecewise polynomials of higher degrees (i.e.,
higher than 1, as in the current piecewise linear case) is covered
in \citet{tibshirani2014adaptive}, where $\ell_2$ estimation error
rates (under weak sparsity) are also derived.  For simplicity, we do
not consider the general piecewise polynomial setting in our study of  
approximate screening, below, though such an extension should be
possible. 
\end{remark}

Now we give our generic approximate screening result, analogous to
that in \Fref{thm:changepoint_screening}.

\begin{theorem}[\textbf{Generic screening result, piecewise linear
  segmentation}] 
\label{thm:knot_screening}
Let $\theta_0 \in \R^n$ be a piecewise linear vector,
and \smash{$\ttheta \in \R^n$} be an estimator satisfying the     
error bound \smash{$\|\ttheta-\theta_0\|_n^2 = 
  O_\P(R_n)$}. Assume that \smash{$n^{1/3}R_n^{1/3}H_n^{-2/3} =
  o(W_n)$}, where, recall, $W_n,H_n$ are as defined in
\eqref{eq:WnHn_linear}. Then    
\begin{equation*}
d\big(S_2(\ttheta)\,|\,S_{0,2} \big) = O_\P 
\bigg( \frac{n^{1/3}R_n^{1/3}}{H_n^{2/3}} \bigg).  
\end{equation*}
\end{theorem}

\begin{proof}
The proof follows that of \Fref{thm:changepoint_screening} closely,
but differs in the lower bound asserted in \eqref{eq:lower_bd}.
As before, given any $\epsilon>0$, $C>0$, as know that for some
integer $N_1>0$ and all $n \geq N_1$,
\begin{equation*}
\P\bigg( \|\ttheta-\theta_0\|_n^2 > \frac{C^3}{4} R_n\bigg) \leq  
\epsilon.
\end{equation*}
We also know that for some integer $N_2>0$ and $n \geq n_2$, it holds
that \smash{$Cn^{1/3}R_n^{1/3}H_n^{-2/3} \leq W_n$}. Let 
$N=\max\{N_1,N_2\}$, take $n \geq N$,  
and let \smash{$r_n = \lfloor
  Cn^{1/3}R_n^{1/3}H_n^{-2/3}\rfloor$}. Suppose that    
\smash{$d(S_2(\ttheta) \,|\, S_{0,2}) > r_n$}.  Then there is a  
knot $t_i \in S_{0,2}$ such that there are no knots in 
\smash{$\ttheta$} within $r_n$ of $t_i$, which means that  
\smash{$\ttheta_j$} displays a linear trend over the entire
segment $j \in \{t_i- r_n,\ldots,t_i+r_n\}$.  Hence
\begin{equation} 
\label{eq:lower_bd_linear}
\frac{1}{n}\sum_{j=t_i-r_n}^{t_i+r_n}
\big(\ttheta_j-\theta_{0,j}\big)^2 
\geq \frac{13r_n^3 ((\theta_{0,t_i+1}-\theta_{0,t_i}) -  
(\theta_{0,t_i}-\theta_{0,t_i-1}))^2}{24n} 
\geq \frac{13r_n^3 H_n^2}{24n} > \frac{C^3}{4} R_n.
\end{equation}
Here, the first inequality holds due to \Fref{lem:lower_bd_linear},
the second holds by definition of $H_n$, and the third by definition
of $r_n$.  We see that
\smash{$d(S_2(\ttheta) \,|\, S_{0,2}) > r_n$} implies the  
estimation error exceeds $(C^3/4)R_n$, an event that we know occurs
with probability at most $\epsilon$, completing the proof.  
\end{proof}

The proof of \Fref{thm:knot_screening} relied on the next lemma, to 
construct the key lower bound \eqref{eq:lower_bd_linear}.  The
lemma characterizes how well a piecewise linear function can be
approximated by a linear one, and is proved in Appendix
\ref{app:lower_bd_linear}. 

\begin{lemma}
\label{lem:lower_bd_linear}
Let $f(x)$ be a piecewise linear function, defined over
$x=-r,\ldots,r$, by  
\begin{equation*}
f(x) = \begin{cases}
a_1x & \text{for $x \geq 0$} \\
a_2x & \text{for $x < 0$}
\end{cases}.
\end{equation*}
Let \smash{$\tilde{a}x+\tilde{b}$} be the optimal linear function for
estimating $f(x)$, according to squared error loss, i.e.,
\begin{equation*}
(\tilde{a},\tilde{b}) = \argmin_{a,b \in \R} \; \sum_{x=-r}^r 
\big( f(x) - ax-b\big)^2.
\end{equation*}
Then 
\begin{equation*}
\sum_{x=-r}^r \big( f(x) - \tilde{a}x-\tilde{b}\big)^2 \geq  
(a_2 - a_1)^2 \frac{13r^3}{24}.
\end{equation*}
\end{lemma}

By combining Theorems \ref{thm:weak_sparsity_linear} and 
\ref{thm:knot_screening}, we have the following approximate screening
result for linear trend filtering. The proof is omitted.

\begin{corollary}[\textbf{Trend filtering screening result, weak
    sparsity setting}] 
\label{cor:knot_screening_weak} 
Assume the conditions in \Fref{thm:weak_sparsity_linear}, thus 
\smash{$TV_2(\theta_0) \leq C_n$} for a sequence $C_n$.  Also assume
\smash{$H_n = \omega(n^{1/10}C_n^{1/5}W_n^{-3/2})$}.  Let
\smash{$\htheta$} denote the linear trend filtering estimate in
\eqref{eq:tf} with \smash{$\lambda=\Theta(n^{1/5}C_n^{-2/3})$}. Then
\begin{equation*}
d\big(\hS_2 \,|\, S_{0,2}\big) = O_\P \bigg(  
\frac{n^{1/15}C_n^{2/15}}{H_n^{2/3}} \bigg). 
\end{equation*}
\end{corollary}

\begin{remark}[\textbf{Knot screening under weak sparsity}]
To give an example of a challenging case that can be accommodated
by \Fref{cor:knot_screening_weak}, consider a setting in which
$\theta_0$ has \smash{$s_0=\Theta(\sqrt{n})$} knots, evenly spread
apart, so that \smash{$W_n=\Theta(\sqrt{n})$}. Then,
provided \smash{$H_n=\omega(n^{-13/20} C_n^{1/5})$},
\Fref{cor:knot_screening_weak} says 
\begin{equation*}
d\big(\hS_2 \,|\, S_{0,2}\big) = O_\P \bigg( 
\frac{n^{1/15}C_n^{2/15}}{H_n^{2/3}} \bigg) = 
o_\P(\sqrt{n}),
\end{equation*}
so that each true knot has a detected knot that is much closer to it
than all other true knots.  Note that $C_n \geq s_0 H_n$, and
combining this with the requirement on $H_n$ reveals 
the implicit requirement 
\smash{$C_n=\omega(n^{-3/16})$}, 
which in turn implies that \smash{$H_n=\omega(n^{-11/16})$}.  This
seems to be a weak requirement on the minimum nonzero change in slopes 
that is present in $\theta_0$.
\end{remark}

\subsection{Changepoint detection on a graph} 
\label{sec:graph}

We depart from the 1-dimensional setting considered throughout the 
paper thus far, and study the model \eqref{eq:model} in a case where 
the mean parameter has components $\theta_{0,i}$, $i=1,\ldots,n$
that correspond to nodes $V=\{1,\ldots,n\}$ of a graph $G$, with edges  
$E=\{e_1,\ldots,e_m\}$.  Note that, for each $\ell=1,\ldots,m$, we may
write $e_\ell=(i,j)$ for some nodes $i,j$ (and all edges are to be
considered undirected, so that $(i,j)$ and $(j,i)$ are equivalent).
Moreover, the mean $\theta_0$ is assumed to behave in a piecewise 
constant fashion over the graph, which means that there are clusters
of nodes over which $\theta_0$ admits constant values, or,
equivalently, $\theta_{0,i}=\theta_{0,j}$ for many edges $(i,j) \in 
E$. For estimation of $\theta_0$, we focus on the
{\it graph fused lasso} or {\it graph-based total variation denoising}  
\citep{tibshirani2005sparsity,hoefling2010path, 
tibshirani2011solution,sharpnack2012sparsistency}, defined by   
\begin{equation}
\label{eq:gfl}
\htheta = \argmin_{\theta \in \R^n} \; 
\half \sum_{i=1}^n (y_i-\theta_i)^2 + 
\lambda \sum_{(i,j) \in E} |\theta_i-\theta_j|, 
\end{equation}
for a tuning parameter $\lambda \geq 0$.  When $G$ is a 1d chain
graph (i.e., $E = \{(1,2), (2,3),\ldots, (n-1,n)\}$), the
estimator in \eqref{eq:gfl} reduces to the ``usual'' 1d fused lasso
estimator in \eqref{eq:fl}.

In the current graph-based setting, the ``changepoints'' of interest
are actually edges for which the corresponding nodes display differing  
values, under a vector $\theta \in \R^n$: 
\begin{equation*}
S_G(\theta) =  \big\{ (i,j) \in E : \theta_i \not= \theta_j \big\}. 
\end{equation*}
We use the abbreviations $S_{0,G}=S_G(\theta_0)$ and
\smash{$\hS_G=S_G(\htheta)$}.  For an edge $(i,j) \in E$,  
let $P_{ij}$ denote the set of paths in $G$ centered around $(i,j)$,
and embedded within two constant clusters of nodes, i.e., 
\begin{multline*}
P_{ij} = \Big\{\big\{(i_\ell,i_{\ell-1})\ldots,(i_1,i),(i,j),
(j,j_1),\ldots,(j_{\ell-1}, j_\ell)\big\} \;\,\text{for any 
  $\ell = 1,\ldots,n$} : \\
\theta_{0,i}=\theta_{0,i_1} = \ldots =\theta_{0,i_\ell}
\;\, \text{and} \;\,
\theta_{0,j}=\theta_{0,j_1} = \ldots =\theta_{0,j_\ell}\Big\}.  
\end{multline*}
We now define $W_n,H_n$ over the graph $G$, in an analogous
fashion to our notions in the 1d setting,  
\begin{equation}
\label{eq:WnHn_graph}
W_n = \min_{(i,j)\in S_{0,G}} \; \max_{p \in P_{ij}} \;
\frac{|p|-1}{2}
\quad \text{and} \quad
H_n= \min_{(i,j) \in S_{0,G}} \; |\theta_{0,i} - \theta_{0,j}|,
\end{equation}
where we write $|p|$ for the number of edges that form a path $p$. 
Note that, as defined, $2W_n+1$ is the minimax length of any
path centered around a changepoint in $S_0$; in other words, by 
construction, for each changepoint $(i,j) \in S_0$, there exists a
path of at least $W_n$ edges embedded entirely within a cluster on
either side of $(i,j)$. When $W_n$ is small, this is indicative of one 
of the constant clusters of nodes in $\theta_0$ being small in size.
For $x \in \R^n$, we define its
graph-based discrete total variation to be
\begin{equation*}
\TV_G(x) = \sum_{(i,j) \in E} |x_i - x_j|.
\end{equation*}
Finally, we must precisely define our screening distance metric in the 
graph-based setting.  For any two edges $e_1=(i_1,j_1),e_2=(i_2,j_2)
\in E$, let $d_G(e_1,e_2)$ denote the length of the shortest path that
starts either $i_1$ or $j_1$, and ends at either $i_2$ or $j_2$.
For sets $A,B \in E$, we define the screening distance 
\begin{equation*}
d_G(A|B) = \max_{e_1=(i_1,j_1) \in B} \; \min_{e_2=(i_2,j_2) \in A} \;   
d_G(e_1,e_2).
\end{equation*}
Hence, if \smash{$d_G(A|B)=k$}, then for any edge in $B$, there is a
path of at most $k$ edges starting from this edge, and ending at an  
edge in $A$. 

Between \citet{wang2016trend} and \citet{hutter2016optimal},
various estimation error rates are available for the graph fused
lasso. These results take on different forms, depending on the
assumptions placed on $\theta_0$ and on the graph $G$.  Below we
recite a result from \citet{hutter2016optimal} in the case that $G$
is a 2d grid graph.  Here \smash{$\htheta$} in \eqref{eq:gfl} is
called the {\it 2d fused lasso} estimate.

\begin{theorem}[\textbf{2d fused lasso error rate, weak and
  strong sparsity settings, Corollary 5 of
  \citealt{hutter2016optimal}}]      
\label{thm:weak_sparsity_graph}
Assume the data model in \eqref{eq:model}, with errors 
$\epsilon_i$, $i=1,\ldots,n$ i.i.d.\ from $N(0,\sigma^2)$.  
Assume that $G$ is a 2d grid graph, with $n$ nodes (hence the 2d
grid is of dimension \smash{$\sqrt{n} \times \sqrt{n}$}).
Write $s_0=|S_{0,G}|$, and \smash{$\TV_G(\theta_0) \leq C_n$}, for
some nondecreasing sequence $C_n$. Then for 
\smash{$\lambda=\Theta(\log{n})$}, the 2d fused lasso estimate
\smash{$\htheta$} in \eqref{eq:gfl} satisfies  
\begin{equation*}
\|\htheta - \theta_0\|_n^2 = O_\P \bigg( \min\{s_0, C_n\}
\frac{\log^2{n}}{n} \bigg).
\end{equation*}
\end{theorem}

\begin{remark}[\textbf{Consistency, optimality}] According to the theorem,
the 2d fused lasso estimator is consistent when either
\smash{$s_0=o(n/\log^2{n})$} or \smash{$C_n=o(n/\log^2{n})$}.
\Fref{thm:weak_sparsity_graph} covers both the weak and 
strong sparsity cases (since it allows us to draw conclusions
involving either $s_0$ or $C_n$).  In the case of weak sparsity, the
\smash{$C_n \log^2{n}/n$} rate achieved by the 2d fused lasso was 
recently shown to be essentially minimax optimal (differing only by
log factors), over the class of signals having bounded total variation 
over the 2d grid $G$, i.e., \smash{$\theta_0 \in \{\theta \in \R^n: 
\TV_G(\theta) \leq C_n\}$}, by \citet{sadhanala2016total}.    
\end{remark}

\begin{remark}[\textbf{Other graphs}] Basically the same result as in
\Fref{thm:weak_sparsity_graph} holds for 3d and higher-dimensional 
grids (except with one fewer log factor) \citep{hutter2016optimal}.
Estimation error rates for various types of random graphs, the
complete graph, and star graphs are derived in
\citet{wang2016trend,hutter2016optimal}.  For simplicity, we do not
consider any of these cases when we give an application of our
generic graph screening result, below; however, given the availability
of $\ell_2$ rates, we remark that results over different graph models
(over than the 2d grid) are certainly possible, and are just a matter
of plugging in the proper rates in the proper settings.
\end{remark}

Here is our generic graph screening result, analogous to that in
\Fref{thm:changepoint_screening} for the 1d chain.

\begin{theorem}[\textbf{Generic screening result, changepoint
  detection on a graph}] 
\label{thm:graph_screening}
Let $\theta_0 \in \R^n$ be piecewise constant over a graph $G$, and 
\smash{$\ttheta \in \R^n$} be an estimator that
satisfies 
\smash{$\|\ttheta-\theta_0\|^2_n = O_\P(R_n)$}.  
Assume that $nR_n/H^2_n = o(W_n)$, where, recall $W_n,H_n$ are as
defined in \eqref{eq:WnHn_graph}. Then
\[
d_G \big(S_G(\ttheta) \,|\, S_{0,G} \big) = O_\P
\bigg( \frac{n R_n}{H_n^2} \bigg).
\]
\end{theorem}

\begin{proof}
The proof is again very similar to the proof of
\Fref{thm:changepoint_screening}. Fix $\epsilon>0$, $C>0$, and let
$N_1>0$ be an integer such that, for $n \geq N_1$, 
\begin{equation*}
\P\bigg( \|\ttheta-\theta_0\|_n^2 > \frac{C}{4} R_n\bigg) \leq  
\epsilon.
\end{equation*}
Let $N_2>0$ be an integer 
such that, for $n \geq N_2$, we have \smash{$CnR_n/H_n^2 \leq W_n$}.
Let $N=\max\{N_1,N_2\}$, take $n \geq N$,  
and define \smash{$r_n = \lfloor CnR_n/H_n^2\rfloor$}. Suppose that     
\smash{$d_G(S_G(\ttheta) \,|\, S_{0,G}) > r_n$}.  By definition, there
exists a changepoint $(i,j) \in S_0$ such that no changepoints in  
\smash{$\ttheta$} are within $r_n$ of $(i,j)$, which in our distance
metric, means that \smash{$\ttheta_k$} is constant over all nodes
$k$ that are $r_n$ away from $i$ or $j$.
Construct an arbitrary path $p$ centered around edge $(i,j)$
with $2r_n+1$ edges 
\[
p = \Big\{(i_{r_n}, i_{r_n-1}),\ldots,(i_1,i), (i,j), (j,j_1) \ldots, 
(j_{r_n-1}, j_{r_n}) \Big\},
\]
where 
\smash{$\theta_{0,i} = \theta_{0,i_1} = \ldots = \theta_{0,i_{r_n}}$}
and 
\smash{$\theta_{0,j} = \theta_{0,j_1} = \ldots = \theta_{0,j_{r_n}}$}.  
(This is possible because $r_n \leq W_n$.)  Denote 
\begin{equation*}
z=\ttheta_{i_{r_n}}=\ldots=\ttheta_i=\ttheta_j=\ldots=  
\ttheta_{j_{r_n}}.
\end{equation*}
Also let \smash{$I(p)=\{i_{r_n-1}, \ldots, i_1, i, j, j_1,
  \ldots, j_{r_n-1} \}$} denote the internal nodes of the path
$p$. Then
\begin{equation*}
\frac{1}{n}\sum_{k \in I(p)}
\big(\ttheta_k-\theta_{0,k}\big)^2 
= \frac{r_n}{n}\big(z-\theta_{0,i}\big)^2 +
\frac{r_n}{n}\big(z-\theta_{0,j}\big)^2 \geq 
\frac{r_n H_n^2}{2n} > \frac{C}{4} R_n,
\end{equation*}
where the first inequality holds because, as argued before, $(x-a)^2 +  
(x-b)^2 \geq (a-b)^2/2$ for all $x$, and the second by definition of
$r_n$. Invoking the assumed $\ell_2$ error rate for \smash{$\ttheta$} 
completes the proof. 
\end{proof}

Combining Theorems \ref{thm:weak_sparsity_graph} and
\ref{thm:graph_screening} gives the next and final result, whose proof
is omitted. 

\begin{corollary}[\textbf{2d fused lasso screening result, weak and
    strong sparsity settings}]
\label{cor:graph_screening_2d}
Assume the conditions in \Fref{thm:weak_sparsity_graph}, so that $G$ 
is a 2d grid, and \smash{$s_0=|S_{0,G}|$}, 
\smash{$\TV_G(\theta_0) \leq C_n$}.  Also assume that
\smash{$H_n=\omega(\min\{\sqrt{s_0},\sqrt{C_n}\}\log{n}/\sqrt{W_n})$}.   
Let \smash{$\htheta$} denote the 2d fused lasso estimate in
\eqref{eq:gfl}, with the choice of tuning parameter
\smash{$\lambda=\Theta(\log{n})$}.  Then 
\[
d\big(\hS_G\,|\,S_{0,G} \big) = 
O_\P\bigg(\min\{s_0,C_n\} \frac{\log^2{n}}{H_n^2} \bigg).  
\]
\end{corollary}

\begin{remark}[\textbf{Screening over a 2d grid}] 
Consider, as a concrete example, a case in which 
$\theta_0$ is piecewise constant with just 2 pieces or clusters, over
the 2d grid $G$ (of dimension, recall, \smash{$\sqrt{n} \times
  \sqrt{n}$}).  Then \smash{$s_0=|S_{0,G}|$} reflects the length of
the boundary separating the 2 pieces.  In the typical case (in which
the 2 pieces are of roughly equal size, and both have volume
proportional to $n$), this scales as \smash{$s_0=\Theta(\sqrt{n})$}.
Moreover, in the typical case, the length $W_n$ of the longest
path on either side of this boundary also scales as
\smash{$W_n=\Theta(\sqrt{n})$}.  Here, $C_n \geq s_0H_n$, which is
larger than $s_0$ unless $H_n$ is quite small ($H_n \leq 1$).
Thus when $H_n$ is large ($H_n>1$), we can use
\Fref{cor:graph_screening_2d} to conclude that  
\begin{equation*}
d\big(\hS_G\,|\,S_{0,G} \big) = O_\P\bigg(
\frac{\sqrt{n}\log^2{n}}{H_n^2} \bigg).
\end{equation*}
Roughly speaking, this says if we were to place a ``tube'' of radius
\smash{$B_n=O_\P(\sqrt{n}\log^2{n}/H_n^2)$} around the boundary
edges \smash{$S_{0,G}$}, containing edges that are at distance of at
most $B_n$ from \smash{$S_{0,G}$}, 
then each changepoint in \smash{$S_{0,G}$} has a corresponding
detected changepoint in \smash{$\hS_G$} lying inside this tube.
Seeing as the entire grid itself is of dimension \smash{$\sqrt{n}
  \times \sqrt{n}$}, this statement is not really interesting unless
$H_n$ is fairly large, say \smash{$H_n=\Theta(n^{1/8} \log{n})$}.
Then \smash{$B_n=O_\P(n^{1/4})$}, giving a reasonably tight tube
around the boundary \smash{$S_{0,G}$}.  
\end{remark}

An illustration of the true changepoints versus those detected
by the 2d fused lasso, in a simple simulated 2d image
example, is given in \Fref{fig:graph}.  See the figure caption for
details. 

\begin{figure}[htb]
\centering
\includegraphics[width=\textwidth]{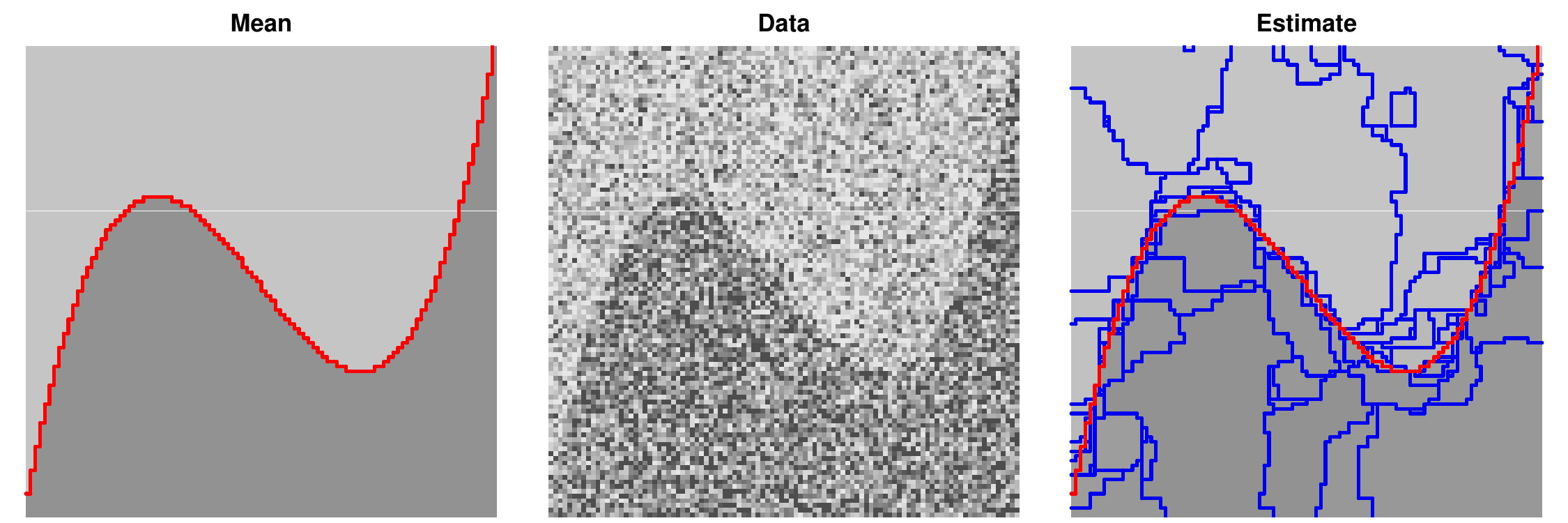}
\caption{\it\small An example on a 2d grid graph, of dimension $100
  \times 100$ (i.e., $n=10,000$). The left panel displays the mean 
  $\theta_0$; the middle panel the data $y$, whose entries were
  generated by 
  adding i.i.d.\ $N(0,1)$ noise to $\theta_0$; and the right panel
  displays the 2d fused lasso estimate, whose tuning parameter
  $\lambda$ was chosen to minimize the $\ell_2$ estimation error in
  retrospect. The mean only has two constant pieces, taking values 0
  and 1, denoted by dark gray and light gray colors, respectively, in
  the left panel.  A consistent color scale is used throughout the
  three plots.  The true changepoint set \smash{$S_{0,G}$} is drawn
  in red, and the estimated changepoint set \smash{$\hS_G$} 
  in blue. In this example, \smash{$d_G(\hS_G \,|\, S_{0,G})=1$}.}
\label{fig:graph}
\end{figure}

\section{Discussion}

We have derived a new $\ell_2$ error bound for the fused 
lasso in a strong sparsity setting, which, to the best of our
knowledge, yields the sharpest available rate in this setting.  We
have also undertaken a detailed study of the manner in which $\ell_2$ 
error bounds for generic estimators \smash{$\ttheta$} can be used to
prove changepoint screening results for \smash{$\ttheta$}, and after
simple post-processing, changepoint recovery results for
\smash{$\ttheta$}.  As a prime example, we have derived new
changepoint screening and recovery results for the fused lasso
estimator, in various settings, based solely on its $\ell_2$ error
guarantees, in these settings. To reiterate, our general technique for 
analyzing changepoint screening and recovery properties is not
specific to the fused lasso, and is potentially much more broadly
applicable, as it only assumes knowledge of the $\ell_2$ error rate of
the estimator \smash{$\ttheta$} in question.  This could be applied
even outside of the typical Gaussian data model.  

We have also presented extensions to the piecewise linear segmentation
and graph changepoint detection problems, as well as detailed
simulations. The code to run all our simulations is located at
\url{https://github.com/linnylin92/fused_lasso}, and relies on the 
R package {\tt genlasso}.

\bibliographystyle{agsm}
\bibliography{bib}

\newpage
\appendix

\section{Proof of \Fref{thm:strong_james}} 
\label{app:strong_james}


Here and henceforth, we write 
$N(r,S,\|\cdot\|)$ to denote the covering
number of a set $S$ in a norm $\|\cdot\|$, i.e., the
smallest number of $\|\cdot\|$-balls of radius $r$  
needed to cover $S$.  We call $\log N(r,S,\|\cdot\|)$ 
the log covering or entropy number.  Recall that we define
the scaled norm \smash{$\|\cdot\|_n = \|\cdot\|_2/\sqrt{n}$}. 

In the proof of \Fref{thm:strong_james}, we will rely on the following
result from \citet{van1990estimating} (which is derived closely from
Dudley's chaining for sub-Gaussian processes).

\begin{theorem}[\textbf{Theorem 3.3 of \citealt{van1990estimating}}] 
\label{thm:van}
Assume that $\epsilon=(\epsilon_1,\ldots,\epsilon_n) \in \R^n$ has
i.i.d.\ components drawn from a
sub-Gaussian distribution, as in \eqref{eq:subgauss}.  Consider a set
$\cX \subseteq \R^n$ with $\|x\|_n \leq 1$ for all $x \in \cX$,
and let $\cK(\cdot)$ be a continuous function upper bounding 
the $\|\cdot\|_n$ entropy of $\cX$, i.e., \smash{$\cK(r) \geq 
  \log N(r,\cX,\|\cdot\|_n)$}.  Then there are constants 
$C_1,C_2,C_3,C_4>0$ (depending only on $M,\sigma$, the parameters of 
the underlying sub-Gaussian distribution 
in \eqref{eq:subgauss}), such that for all $t > C_1$, with  
\begin{equation*}
t > C_2 \int_0^{t_0} \sqrt{\cK(r)} \, dr,
\end{equation*}
where \smash{$t_0 = \inf\{r : \cK(r) \leq C_3t^2 \}$}, we
have   
\begin{equation*}
\P \bigg( \sup_{x \in \cX} \; \frac{|\epsilon^\top x|}{\sqrt{n}} >
t\bigg) \leq 2 \exp(-C_4 t^2). 
\end{equation*}
\end{theorem}

Now we give the proof of \Fref{thm:strong_james}.

\begin{proof}[Proof of \Fref{thm:strong_james}]
The proof is given in two parts, one in which we bound 
\smash{$\|\hat\delta\|^2_2$} and the other in which we bound
\smash{$\|\hx \|^2_2$} . 
Recall that \smash{$\hat\delta = P_0(\htheta - \theta_0)$} and  
\smash{$\hx = P_1\htheta$}. 
Each part begins with a different ``basic inequality'', 
established by comparing the fused lasso objective at different
points. First, we define the following events,
\begin{align}   
\label{eq:omega0}
\Omega_0 = &\bigg\{ \sup_{z \in \cM} \; \frac{|\epsilon^\top z|}{\|
  z \|_2} \leq \gamma c_I \sqrt{(\log{s_0}+\log\log{n}) s_0 \log{n}}
\bigg\}, \\  
\label{eq:omega1}
\Omega_1 = &\bigg\{ \sup_{w \in \cR^\perp} \; \frac{|\epsilon^\top
  w|}{\| \op_{-S_0} w \|^{1/2}_1 \|w\|^{1/2}_2} \le \gamma
c_R (n s_0)^{1/4} \bigg\}, \\ 
\label{eq:omega2}
\Omega_2 = &\bigg\{ \sup_{\delta \in \cR} \; \frac{|\epsilon^\top
  \delta|}{\|\delta\|_2} \le \gamma c_S \sqrt{s_0} \bigg\},  
\end{align}
where $\gamma > 1$ is parameter free to vary in our analysis,  
$c_I, c_R>0$ are the constants in Lemmas \ref{lem:GCinterp}, 
\ref{lem:GCresid}, and $c_S > 0$ is a constant to be determined
below. Focusing on the third event, we will lower bound its
probability by applying \Fref{thm:van} to \smash{$\cX=\cR \cap \{
  \delta : \|\delta\|_n \leq 1\}$}. Note that 
\begin{equation*}
\log N(r, \cR \cap \{ \delta : \|\delta\|_n \leq 1\}, \|\cdot\|_n)
\leq (s_0+1) \log(3/r),
\end{equation*}
as $\cR$ is $(s_0+1)$-dimensional, and it is well-known that in
$\R^d$, the number of balls of radius $r$ that are needed to cover 
the unit ball is at most \smash{$(3/r)^d$}.  The quantity
$t_0$ in \Fref{thm:van} may be taken to be
\smash{$t_0 = \inf\{r : (s_0+1) \log(3/r) \leq C_3C_1^2\} = 
  3\exp(-C_3C_1^2/(s_0+1))$}. The restrictions on
$t$ are hence $t > C_1$, as well as  
\begin{equation*}
t > C_2 \int_0^{t_0} \sqrt{(s_0+1) \log(3/r)} \, dr.
\end{equation*}
But, writing $\erf(\cdot)$ for the error function, 
\begin{equation*}
C_2 \int_0^{t_0} \sqrt{(s_0+1) \log(3/r)} \, dr
= (\sqrt{s_0+1}) \cdot 3C_2 \bigg[r \sqrt{\log{\frac{1}{r}}} - \half   
\erf\bigg(\sqrt{\log{\frac{1}{r}}} \bigg) \bigg] \bigg|_0^{t_0/3} \leq
C_2 \sqrt{s_0},
\end{equation*}
where the constant $C_2>0$ is adjusted to be larger, as needed. 
Let us define $c_S=\max\{C_1,C_2\}$ and $C_S=C_4$. Then we have by 
\Fref{thm:van}, for \smash{$t=\gamma c_S \sqrt{s_0}$} and any $\gamma 
> 1$,  
\begin{equation}
\label{eq:prob_omega2}
1 - 2\exp(-C_S \gamma^2 c_S^2 s_0) \leq 
\P \bigg( \sup_{\delta \in \cR} \; 
\frac{|\epsilon^\top \delta|}{\sqrt{n} \|\delta\|_n} \leq \gamma c_S
\sqrt{s_0} \bigg) = \P \bigg( \sup_{\delta \in \cR} \;  
\frac{|\epsilon^\top \delta|}{\|\delta\|_2} \leq \gamma c_S \sqrt{s_0}
\bigg)  = \P(\Omega_2). 
\end{equation}

\bigskip
\noindent
\textbf{Controlling \smash{$\hat\delta$}.} 
Comparing the objective in \eqref{eq:fl} at \smash{$\htheta = \Proj_0
  \htheta + \Proj_1 \htheta = \theta_0 + \hat \delta + \hx$} and at 
\smash{$\theta_0 + \hx$}, we have
\begin{equation*}
\| \hat \delta + \hx - \epsilon \|_2^2 + \lambda \| \op \htheta \|_1
\le \| \hx - \epsilon \|_2^2 + \lambda \| \op (\theta_0 + \hx)\|_1,   
\end{equation*}
and by rearranging terms we obtain our basic inequality,
\begin{equation}
\label{eq:basic_ineq_d}
\| \hat \delta \|_2^2 \le 2 \hat \delta^\top \epsilon + \lambda \Big(
  \|\op (\theta_0 + \hx)\|_1 - \| \op \htheta \|_1 \Big),  
\end{equation}
which follows from the fact that \smash{$\hat\delta^\top \hx = 0$}
(as they lie in orthogonal subspaces).
Furthermore, since \smash{$\htheta =  \theta_0 + \hat \delta + \hx$},
\begin{equation*}
\|\op (\theta_0 + \hx)\|_1 - \| \op \htheta \|_1 
\leq \|\op \hat \delta \|_1 = \| \op_{S_0} \hat \delta \|_1, 
\end{equation*}
where we used the triangle inequality, and the
fact that \smash{$\op_{-S_0} \hat\delta = 0$}. 
So from our basic inequality in \eqref{eq:basic_ineq_d}, we have that 
\begin{equation*}
\| \hat\delta \|_2^2 \le 2 \hat \delta^\top \epsilon + \lambda
\|\op_{S_0} \hat\delta \|_1,
\end{equation*}
and dividing by \smash{$\|\hat\delta \|_2$}, we get 
\begin{equation*}
\|\hat\delta \|_2 \le 2 \frac{|\hat \delta^\top \epsilon|}{\|\hat 
\delta\|_2} + \lambda \frac{\|\op_{S_0} \hat \delta \|_1}{\|\hat
\delta\|_2}. 
\end{equation*}
Now observe that
\begin{equation*}
\| \op_{S_0} \hat \delta \|_1 = \sum_{i=1}^{s_0} |\hat
\delta_{t_{i+1}} - \hat \delta_{t_i}| 
\le 2 \sum_{i=1}^{s_0+1} |\hat \delta_{t_i}| 
\le 2 \sqrt{(s_0+1)
\sum_{i=1}^{s_0+1} {\hat\delta_{t_i}}^2} 
\le 4 \sqrt{s_0 \sum_{i=1}^{s_0+1} \frac{t_i
- t_{i-1}}{W_n} {\hat\delta_{t_i}}^2} = 4\sqrt{\frac{s_0}{W_n}}  
\| \hat\delta \|_2.
\end{equation*}
The second inequality used Cauchy-Schwartz, and the last
equality used that \smash{$\hat\delta$} is piecewise
constant on the blocks \smash{$B_0,\ldots,B_{s_0}$}, as
\smash{$\hat\delta \in \cR = \spa \{\one_{B_0}, \ldots,
  \one_{B_{s_0}}\}$}.  Hence, on the event $\Omega_2$, we have 
\begin{equation}
\label{eq:delta_bound}
\| \hat\delta \|_2 \le 2 \gamma c_S \sqrt{s_0}  + 4 \lambda  
\sqrt{\frac{s_0}{W_n}}. 
\end{equation}

\bigskip
\noindent
\textbf{Controlling $\hx$.}
We can establish our next basic inequality by comparing the objective  
in \eqref{eq:fl} at \smash{$\htheta$} and \smash{$\theta_0 +
  \hat\delta$}, 
\begin{equation*}
\| \hx + \hat\delta - \epsilon \|_2^2 + \lambda \| \op \htheta \|_1 
\le \| \hat\delta - \epsilon \|_2^2 + \lambda \|\op
(\theta_0+\hat\delta) \|_1,  
\end{equation*}
or, rearranged,
\begin{align}
\nonumber
\|\hx\|_2^2 & \le 2 \epsilon^\top \hx + \lambda \Big( \|\op_{S_0} 
(\theta_0+\hat\delta )\|_1 - \|\op_{S_0} \htheta\|_1 - \|\op_{-S_0}
\hx\|_1 \Big) \\    
\label{eq:basic_ineq_x}
&\le 2 \epsilon^\top \hx + \lambda \Big(  
\|\op_{S_0} \hx \|_1 - \|\op_{-S_0} \hx\|_1  \Big),
\end{align}
where the first line used \smash{$\hx^\top \hat\delta=0$} and
\smash{$\op_{-S_0} \theta_0 = \op_{-S_0} \hat\delta=0$}, and the
second used \smash{$\htheta = \theta_0 + \hat\delta + \hx$} and the
triangle inequality.  

Decompose \smash{$\hx=\hz+\hw$}, where
\smash{$\hz \in \cM$} is the lower interpolant to \smash{$\hx$}, as 
defined in \Fref{lem:interp}, and \smash{$\hw = \hx -
  \hz$} is the remainder. Combining the basic inequality in
\eqref{eq:basic_ineq_x} with \eqref{eq:Zcond1} and \eqref{eq:Zcond2}
from \Fref{lem:interp},
\begin{align}
\nonumber
\| \hx \|_2^2 &\leq 2 \epsilon^\top \hz + 2 \epsilon^\top \hw +
\lambda \Big( \|\op_{S_0} \hz \|_1 - \|\op_{-S_0} \hz\|_1
- \|\op_{-S_0} \hw \|_1 \Big) \\  
\label{eq:basic_ineq_x2}
&\le 2 \epsilon^\top \hz + 4 \lambda \sqrt{\frac{s_0}{W_n}} \|\hz \|_2 
+ 2 \epsilon^\top \hw - \lambda \|\op_{-S_0} \hw \|_1.
\end{align}
On the event $\Omega_0$ in \eqref{eq:omega0}
\begin{equation*}
\epsilon^\top \hz \le \gamma c_I 
\sqrt{(\log{s_0}+\log\log{n}) s_0 \log{n}} \| \hz \|_2.  
\end{equation*}
Further, on the event $\Omega_1$ in \eqref{eq:omega1}, 
since \smash{$P_1 \hw \in \cR^\perp$}, \smash{$\| \op_{-S_0} P_1\hw
  \|_1 = \| \op_{-S_0} \hw \|_1$}, and \smash{$\| P_1 \hw \|_2 \le \|
  \hw \|_2$},
\begin{equation*}
  \epsilon^\top P_1 \hw 
  \le \gamma c_R (n s_0)^{1/4} \| \op_{-S_0} \hw 
  \|_1^{1/2} \| \hw \|_2^{1/2},
\end{equation*}
Also, on the event $\Omega_2$ in \eqref{eq:omega2}, since 
$P_0 \hw \in \cR$, 
\begin{equation*}
  \epsilon^\top P_0 \hw \le \gamma c_S \sqrt{s_0} \| \hw \|_2. 
\end{equation*}
Hence, on the event $\Omega_0 \cap \Omega_1 \cap \Omega_2$, combining
the last three displays with \eqref{eq:basic_ineq_x2},
\begin{multline}
\label{eq:basic_ineq_x3}
\| \hx \|_2^2 \le 2\bigg(\gamma c_I 
\sqrt{(\log{s_0}+\log\log{n}) s_0 \log{n}} + 
2\lambda \sqrt{\frac{s_0}{W_n}} \bigg) \| \hz \|_2 + 2\gamma 
c_S \sqrt{s_0} \| \hw \|_2 + {} \\ 
2\gamma c_R (n s_0)^{1/4} \| \op_{-S_0} \hw \|_1^{1/2} \| \hw
\|_2^{1/2} -  \lambda \| \op_{-S_0} \hw \|_1.
\end{multline}
Consider the first case in which
\smash{$2\gamma c_R (n s_0)^{1/4} \|\op_{-S_0} \hw \|_1^{1/2} \|
  \hw \|_2^{1/2} \ge \lambda \|\op_{-S_0} \hw \|_1$}. Then 
\begin{equation*}
\|\op_{-S_0} \hw \|_1 \le 4\bigg( \frac{\gamma c_R}{\lambda} \bigg)^2 
\sqrt{n s_0} \| \hw \|_2,
\end{equation*}
and from \eqref{eq:basic_ineq_x3}, on the event $\Omega_0 \cap
\Omega_1 \cap \Omega_2$, 
\begin{equation}
\label{eq:x_bound}
\| \hx \|_2 \le 2\gamma c_I \sqrt{(\log{s_0}+\log\log{n}) s_0 \log{n}}
+ 4 \lambda \sqrt{\frac{s_0}{W_n}} + 2\gamma c_S\sqrt{s_0} +
\frac{4\gamma^2 c_R^2 \sqrt{n s_0}}{\lambda}.
\end{equation}
where in the above we used \eqref{eq:Zcond3}.
In the case that
\smash{$2\gamma c_R (n s_0)^{1/4} \|\op_{-S_0} \hw \|_1^{1/2} \|
  \hw \|_2^{1/2} < \lambda \|\op_{-S_0} \hw \|_1$}, 
we have from \eqref{eq:basic_ineq_x3}, on the event $\Omega_0 \cap
\Omega_1 \cap \Omega_2$,  
\begin{equation*}
\| \hx \|_2 \le 2\gamma c_I \sqrt{(\log{s_0}+\log\log{n}) s_0 \log{n}}
+ 4 \lambda \sqrt{\frac{s_0}{W_n}} + 2\gamma c_S \sqrt{s_0}. 
\end{equation*}
Therefore, the bound \eqref{eq:x_bound} always holds on the event
$\Omega_0 \cap \Omega_1 \cap \Omega_2$.

\bigskip
\noindent
\textbf{Putting it all together.}
As \smash{$\| \htheta - \theta_0 \|_2 \le \| \hx \|_2 + \| \hat
  \delta \|_2$}, combining \eqref{eq:x_bound} and
\eqref{eq:delta_bound}, we see that  
\begin{equation*}
  \| \htheta - \theta_0 \|_2 \le
4 \gamma c_S \sqrt{s_0} +
8 \lambda \sqrt{\frac{s_0}{W_n}} +
2 \gamma c_I \sqrt{(\log{s_0}+\log\log{n}) s_0 \log{n}} +  
\frac{4\gamma^2 c_R^2 \sqrt{n s_0}}{\lambda},
\end{equation*}
on the event $\Omega_0 \cap \Omega_1 \cap \Omega_2$.
We see that there exists a constant $c>0$,
such that for large enough $n$, and any $\gamma > 1$,  
\begin{equation}
\label{eq:final_bd}
\| \htheta - \theta_0 \|_2^2 \le \gamma^4 c s_0
\Bigg( (\log{s_0}+\log\log{n}) \log{n} + 
\frac{\lambda^2}{W_n} + \frac{n}{\lambda^2} \Bigg),   
\end{equation}
on the event $\Omega_0 \cap \Omega_1 \cap \Omega_2$. Furthermore,
using the union bound along with Lemmas \ref{lem:GCinterp},
\ref{lem:GCresid}, and \eqref{eq:prob_omega2}, we find that
\begin{multline*}
\P \big((\Omega_0 \cap \Omega_1 \cap \Omega_2)^c \big) 
\leq 2\exp\big(-C_I \gamma^2 c_I^2(\log{s_0}+\log\log{n})\big) 
+ {} \\ 2\exp(-C_R \gamma^2 c_R^2 \sqrt{s_0}) +
2\exp(-C_S \gamma^2 c_S^2 s_0) \leq \exp(-C \gamma^2),  
\end{multline*}
for an appropriately defined constant $C>0$.  Optimizing the 
bound in \eqref{eq:final_bd} to choose the tuning parameter
$\lambda$ yields \smash{$\lambda=(nW_n)^{1/4}$}. Plugging this in
gives the final result.
\end{proof}

\section{Proofs of Lemmas \ref{lem:interp}, \ref{lem:GCinterp},
  \ref{lem:GCresid}}  
\label{app:lemma_james}

\begin{proof}[Proof of \Fref{lem:interp}]
We give an explicit construction of a lower interpolant $z \in \cM$  
to $x$, given the changepoints
\smash{$0=t_0<\ldots < t_{s_0+1}=n$}. We will use the notation
\smash{$a_+=\max\{0,a\}$}.  
For $i = 0,\ldots,s_0$, define \smash{$z^{(i+)} \in 
\R^{t_{i+1} - t_{i}}$} by setting \smash{$g_i^+ = \sign(x_{t_i})$} and  
\begin{equation*}
  z^{(i+)}_j = g_i^+ \cdot \min 
\Big\{ (g_i^+  x_{t_i+1})_+, \ldots, (g_i^+  x_{t_i + j})_+ \Big\},
\quad j = 1,\ldots, t_{i+1} - t_i. 
\end{equation*}
Similarly, define \smash{$z^{(i-)}\in \R^{t_{i+1} - t_{i}}$} by
setting \smash{$g_i^- = \sign( x_{t_{i+1}-1})$} and
\begin{equation*}
z^{(i-)}_j = g_i^- \cdot \min 
\Big\{ (g_i^-  x_{t_{i} + j})_+, \ldots, (g_i^- x_{t_{i+1}})_+ \Big\},
\quad j = 1,\ldots, t_{i+1} - t_i.  
\end{equation*}
Note that \smash{$z^{(i+)}_1 = x_{t_i+1}$} and
\smash{$z^{(i-)}_{t_{i+1}-t_i} = x_{t_{i+1}}$}; also, 
$\{|z^{(i+)}_j|\}_{j=1}^{t_{i+1}-t_i}$ is a nonincreasing sequence,
and \smash{$\{|z^{(i-)}_j|\}_{j=1}^{t_{i+1}-t_i}$} is nondecreasing.  
Furthermore,
\[
\sign \big(z^{(i+)}_1\big) \cdot \sign \big(z^{(i+)}_j\big) \geq 0
\quad \text{and} \quad 
\sign \big(z^{(i-)}_{t_{i+1}-t_i}\big) \cdot \sign
\big(z^{(i-)}_j\big) \geq 0, \qquad j = 1,\ldots, t_{i+1} - t_i. 
\]
Lastly, notice that there exists a point \smash{$j'\in
  1,\ldots,t_{i+1}-t_i-1$} (not necessarily unique) such that 
\begin{align}
\min_{k\in \{1,\ldots, t_{i+1}-t_i\}}\big|z^{(i+)}_{k}\big| &= 
\big|z^{(i+)}_{j'+1}\big| = \big|z_{j}^{(i+)}\big|, \qquad j = j'+1,
\ldots, t_{i+1}-t_i, \label{eq:minz1}\\ 
\min_{k\in \{1,\ldots, t_{i+1}-t_i\}}\big|z^{(i-)}_{k}\big| &= 
\big|z^{(i-)}_{j'}\big| = \big|z_{j}^{(i-)}\big|, \qquad j = 1,
\ldots, j'. \label{eq:minz2} 
\end{align}
We construct by \smash{$z_{t_i + j} = z_j^{(i+)}$} for $j =
1,\ldots,j'$, and \smash{$z_{t_i + j} = z_j^{(i-)}$} for $j =
j'+1,\ldots,t_{i+1}-t_i$. Letting $t'_i = t_i + j'$ and repeating 
this process for $i = 0,\ldots, s_0$, we have constructed $z \in 
\cM$. 

We now verify the claimed properties for the constructed lower
interpolant $z$. For $i=0,\ldots, s_0$, and any 
$j = 1,\ldots, t_{i+1} - t_i$, we have 
\begin{align}
\label{eq:zcond1}
&\sign( z_{j}^{(i+)}) \cdot \sign( x_{t_i + j}) \ge 0, \\ 
\label{eq:zcond2}
&|  z_{j}^{(i+)} | \le |  x_{t_i + j} |,
\end{align}
Further, for any $j = 1,\ldots,t_{i+1} - t_i - 1$,
\begin{align}
\label{eq:diffzcond1}
&\sign\Big( (\op  z^{(i+)})_j \Big) \cdot \sign \big( (\op
  x)_{t_i + j} \big) \ge 0,\\ 
\label{eq:diffzcond2}
&\Big| (\op  z^{(i+)})_j \Big| \le \big| (\op  x)_{t_i + j}
\big|. 
\end{align}
To see why \eqref{eq:diffzcond1} holds, note that
\smash{$\sign (\op z^{(i+)})_j \in \{-1,0\}$}, \smash{$(\op
  z^{(i+)})_j < 0$} 
imply \smash{$(\op (g_i^+ x)_+)_{t_i + j} < 0$}. 
To see why \eqref{eq:diffzcond2} holds, 
if \smash{$(\op z^{(i+)})_j \ne 0$}, then we know that
\begin{equation*}
|z^{(i+)}_{j+1} - z^{(i+)}_{j}| \leq 
\Big|\min\Big\{(g_i^+x_{t_i+j+1})_+, (g_i^+x_{t_i+j})_+\Big\}
- (g_i^+x_{t_i+j})_+\Big| \leq |x_{t_i+j+1} - x_{t_i+j}|,
\end{equation*}
where we used the observation that 
\smash{$|\min\{a,b\} - b| \geq |\min\{a,b,c\} - \min\{b,c\}|$}. 

It can be shown by nearly equivalent steps that \smash{$z^{(i-)}$},
and $z$ both satisfy properties analogous to 
\eqref{eq:zcond1}--\eqref{eq:diffzcond2}. 
Using \eqref{eq:zcond1} and \eqref{eq:zcond2} on $z$ gives  
\eqref{eq:Zcond3}.
Using \eqref{eq:diffzcond1} and \eqref{eq:diffzcond2} on
$z$ gives \eqref{eq:Zcond1} (note that
if $\sign(a) = \sign(b)$ and $|a|>|b|$, then $|a| = |b| +
  |a-b|$). Because \smash{$z_{t_i+1} =  x_{t_i+1}$} and 
\smash{$z_{t_{i+1}} =  x_{t_{i+1}}$} for 
all $i=0,\ldots,s_0$, we have the equality in \eqref{eq:Zcond2} (since  
\smash{$\op_{t_i}z = z_{t_i+1} - z_{t_i} = x_{t_i+1} - x_{t_i} =
  \op_{t_i}x$}). 

Finally, for each $i = 0,\ldots,s_0$, define
\smash{$t''_i = t'_i$ if $|z_{t'_i}| \geq |z_{t'_i+1}|$} and $t''_i =
t'_i+1$ otherwise.  Observe that by \eqref{eq:minz1} and
\eqref{eq:minz2}, it holds that \smash{$|z_{t''_i}| = \min_{j =
    1,\ldots, t_{i+1}-t_i} |z_{t_i+j}|$}.
The inequality in \eqref{eq:Zcond2} is finally established by the 
following chain of inequalities:
\begin{align*}
\| \op_{S_0}  z\|_1 &= \sum_{i=1}^{s_0} | z_{t_i+1} -  z_{t_i}| \le
\sum_{i=1}^{s_0} | z_{t_i+1}| + | z_{t_{i}}| \\ 
  &= \sum_{i=1}^{s_0} \big(| z_{t_i+1}| - | z_{t''_i}|\big) + \big(|
  z_{t_{i}}| - | z_{t''_{i-1}}|\big) + | z_{t''_{i-1}}|+ |
  z_{t''_i}| \\ 
& \le \| \op_{-S_0}  z \|_1 + 2 \sum_{i=0}^{s_0} | z_{t''_i}| 
\le \| \op_{-S_0}  z \|_1 + 4\sqrt{\frac{s_0}{W_n}} \|  z\|_2, 
\end{align*}  
where in the second inequality, we used $|a| - |c| \leq |a-c| \leq
|a-b| + |b-c|$, and in the last inequality, we used the above property
of $z_{t''_i}$ and 
\[
\sum_{i=0}^{s_0} | z_{t''_i}| \leq 2\sqrt{s_0} \sqrt{\sum_{i=0}^{s_0} |
z_{t''_i}|^2}
\leq 2\sqrt{s_0 \sum_{i=0}^{s_0} \frac{t_{i+1}-t_i}{W_n} z_{t''_i}^2}
\leq 2\sqrt{\frac{s_0}{W_n}}\|z\|_2.
\]
This completes the proof.
\end{proof}

\begin{proof}[Proof of \Fref{lem:GCinterp}]
We consider $\epsilon \in \R^n$, an i.i.d.\ sub-Gaussian vector as 
referred to in the statement of the lemma, and arbitrary $z \in \cM$. 
In this proof, we will also consider $E(t)$ and $Z(t)$,
real-valued functions over $[0,n]$, constructed so that
\smash{$E(t) = \epsilon_{\lceil t \rceil}$} for all $t$ (i.e.,
$E(t)$ is a step function), $Z(t) = z_t$ for
$t=1,\ldots,n$, and $Z(t)$ is smooth and monotone over $(t_i, t'_i]$  
and $(t'_i,t_{i+1}]$ for $i=0,\ldots,s_0$.  These functions will also
satisfy the boundary conditions $E(0) = \epsilon_1$ and
$Z(0) = z_1$.  

Let \smash{$F(t) = \int_0^t E(u) \, du$}. 
As $\epsilon$ is random, $E(t)$ and $F(t)$ are also random.  
It can be shown that there exists 
constants $c_I,C_I>0$ such that for any $\gamma > 1$,
\begin{multline}
\label{eq:claim}
\P \left( \frac{|F(t) - F(t_i)|}{\sqrt{|t - t_i|}} \le\gamma  c_I
\sqrt{\log{s_0}+\log\log{n}}, \;
\text{for $t \in (t_i,t_{i+1}]$, $i=0,\ldots,s_0$} \right) \\ \geq 1 -
2 \exp \big(-C_I \gamma^2 c_I^2 (\log{s_0}+\log\log{n}) \big).    
\end{multline}
So as not to distract from the main flow of ideas, we now proceed 
to prove \Fref{lem:GCinterp}, and we later provide a proof of 
\eqref{eq:claim}. Let $\Omega_3$ denote the event in consideration on 
the left-hand side of \eqref{eq:claim}. By integration by parts, 
\begin{equation*}
\int_{t_i}^{t'_i} E(t) Z(t) \, dt = Z(t'_i) (F(t'_i) - F(t_i))  -
\int_{t_i}^{t'_i} Z'(t) (F(t) - F(t_i)) \, dt
\end{equation*}
where \smash{$Z'(t) = \frac{d}{dt}Z(t)$}.
Thus, on the event $\Omega_3$,
\begin{equation}
\label{eq:inner1}
\left|\int_{t_i}^{t_i'} E(t) Z(t) \, dt \right| \leq \gamma c_I 
\sqrt{\log{s_0}+\log\log{n}}
\left(|Z(t_i')| \sqrt{t_i' - t_i} + \left| \int_{t_i}^{t_i'} Z'(t)
    \sqrt{t - t_i} \, dt \right| \right),  
\end{equation}
since $Z'$ does not change sign within the intervals $(t_i,t_i'], 
(t_i',t_{i+1}]$ (as $z \in \cM$). 
For $n$ large enough, we can upper bound the last term in 
\eqref{eq:inner1} as follows
\begin{equation}
\label{eq:decomp}
\left|\int_{t_i}^{t_i'} Z'(t) \sqrt{t - t_i} \, dt \right| =  
\left|\int_{t_i}^{t_i+n^{-1}} Z'(t) \sqrt{t - t_i} \, dt \right|
+ \left|\int_{t_i+n^{-1}}^{t_i'} Z'(t)\sqrt{t - t_i} \, dt \right|. 
\end{equation}
Using integration by parts and the triangle inequality, on the
second term in \eqref{eq:decomp},
\begin{equation}
\label{eq:inner2}
\left|\int_{t_i + n^{-1}}^{t_i'} Z'(t) \sqrt{t - t_i} \, dt \right| = 
|Z(t_i')| \sqrt{t_i' - t_i} + \left| \frac{Z(t_i + n^{-1})}{\sqrt n}
\right| +  \half \left| \int_{t_i
  + n^{-1}}^{t_i'} \frac{Z(t)}{\sqrt {t - t_i}} \, dt \right|.
\end{equation}
By Cauchy-Schwartz on the last term in \eqref{eq:inner2}, 
\begin{align}
\nonumber
 \left| \int_{t_i + n^{-1}}^{t_i'} \frac{Z(t)}{\sqrt {t - t_i}} \, dt 
 \right| &\le \left( \int_{t_i + n^{-1}}^{t_i'} Z(t)^2 \, dt
 \right)^{1/2} \left( \int_{t_i + n^{-1}}^{t_i'} \frac{1}{t - t_i} \,
   dt \right)^{1/2} \\ 
\label{eq:inner3} &\le \left( \int_{t_i + n^{-1}}^{t_i'} Z(t)^2 \, dt   
\right)^{1/2} \sqrt{2 \log n}.    
\end{align}
Now examining the first term in \eqref{eq:decomp},
\begin{equation*}
  \left| \int_{t_i}^{t_i+n^{-1}} Z'(t) \sqrt{t - t_i} \, dt \right| \le
  n^{-1/2} \left| \int_{t_i}^{t_i+n^{-1}} Z'(t) \, dt \right| =
  \frac{|Z(t_i+n^{-1}) - Z(t_i)|}{\sqrt n}.
\end{equation*}
But because we only require $Z$ to be a piecewise monotonic and smooth
interpolant then we are at liberty to make $Z(t_i + n^{-1}) = Z(t_i)$, forcing
this term to be $0$. 
In order to bound $Z(t_i')$, notice that because $|Z(t)|$ is
non-increasing over the interval $(t_i,t_i']$ we have that  
\begin{equation}
\label{eq:inner4}
  Z(t'_i)^2 |t_i' - t_i| \le \int_{t_i}^{t_i'} Z(t)^2 \, dt.
\end{equation}
Combining \eqref{eq:inner1}--\eqref{eq:inner4}, we have that on the
event $\Omega_3$, 
\begin{equation}
\label{eq:first_half_bd}
\left|\int_{t_i}^{t_i'} E(t) Z(t) \, dt \right| \leq \alpha_n
\left(2 + \sqrt{ \frac{\log n}{2} } \right) \left(\int_{t_i}^{t_i'}
  Z(t)^2 \, dt \right)^{1/2} + \alpha_n \frac{|Z(t_i)|}{\sqrt n}.
\end{equation}
where we have abbreviated \smash{$\alpha_n=\gamma
  c_I\sqrt{\log{s_0}+\log\log{n}}$}. 
Through nearly identical steps we can show that on the event
$\Omega_3$, 
\begin{equation}
\label{eq:second_half_bd}
\left|\int_{t_i'}^{t_{i+1}} E(t) Z(t) \, dt \right| \leq \alpha_n
\left(2 + \sqrt{ \frac{\log n}{2} } \right)
\left(\int_{t'_i}^{t_{i+1}} Z(t)^2 \, dt \right)^{1/2} + \alpha_n
\frac{|Z(t_{i+1})|}{\sqrt n}. 
\end{equation}
Therefore
\begin{align}
\nonumber
\left|\int_{0}^{n} E(t) Z(t) \, dt \right| &\le \sum_{i=0}^{s_0}  
\left( \, \left|\int_{t_i}^{t_i'} E(t) Z(t) \, dt \right| +
\left|\int_{t_i'}^{t_{i+1}} E(t) Z(t) \, dt \right| \, \right) \\ 
\label{eq:all_bd}
&\leq \alpha_n \sqrt{2s_0 + 2} \left(2 + 
\sqrt{ \frac{\log n}{2}} \right) \left( \int_0^n Z(t)^2 \, dt
\right)^{1/2} + 2 \alpha_n \frac{\| z \|_1}{\sqrt n}, 
\end{align}
where in the second line we applied \eqref{eq:first_half_bd},
\eqref{eq:second_half_bd}, and the Cauchy-Schwartz inequality.  
Because we can choose $Z(t)$ to be arbitrarily close to
\smash{$z_{\lceil t \rceil}$} over all $t$, the integral 
\smash{$( \int_0^n Z(t)^2 \, dt)^{1/2}$} is approaching $\|z\|_2$ and  
\smash{$\int_0^n E(t) Z(t) \, dt$} is approaching $\epsilon^\top z$.
Furthermore, because \smash{$\| z \|_1 \le \sqrt n \| z \|_2$}, the
first term in \eqref{eq:all_bd} dominates. Hence on the event
$\Omega_3$, we have established yet 
\begin{equation*}
|\epsilon^\top z | \leq \gamma c_I \sqrt{(\log{s_0}+\log\log{n})
s_0 \log{n}} \|z\|_2,
\end{equation*}
where the constant $c_I$ is adjusted to be larger, as needed. Noting
that the event $\Omega_3$ does not depend on $z$, the result follows.  
\end{proof}

\begin{proof}[Proof of claim \eqref{eq:claim}]
We will construct a covering for \smash{$\cV = \cup_{i=0}^{s_0}
  \cV_i$}, where for each $i = 0,\ldots,s_0$, 
\begin{equation*}
\cV_i = \bigg\{ \sqrt{\frac{n}{|A|}} \one_A
: A = \{t_i, \ldots, t\}, \; t = t_i + 1, \ldots, n \bigg\} 
\;\cup\; \bigg\{ \sqrt{\frac{n}{|A|}} \one_A
: A = \{t, \ldots, t_i\}, \; t = 1, \ldots, t_i-1 \bigg\}.
\end{equation*}
Note that our scaling is such that, for any 
\smash{$a = \sqrt{n/|A|}  \one_A$}, where $A \subseteq \{1,\ldots,n\}$,
we have $\|a\|_n=1$. Further, for any other 
\smash{$b = \sqrt{n/|B|}\one_B$}, where $B \subseteq \{1,\ldots,n\}$,
we have  
\begin{equation}
\label{eq:ab_norm}
\| a - b \|_n^2 = \frac{|A \cap B|}{(\sqrt{|A|}-\sqrt{|B|})^2}
+ \frac{|A\setminus B|}{|A|} + \frac{|B\setminus A|}{|B|}
= 2\bigg(1 - \frac{|A \cap B|}{\sqrt{|A||B|}}\bigg),
\end{equation}

We first construct a covering for each set $\cV_i$, $i=0,\ldots,s_0$,
restricting our attention to a radius \smash{$0<r<\sqrt{2}$}.  
Let \smash{$\alpha = \lceil (1 - r^2/2)^{-2} \rceil$}, and consider
the set   
\begin{multline*}
\cC_i = \left\{
\sqrt{\frac{n}{|A|}} \one_A : A = \big\{t_i,\ldots,
\min\{t_i + \alpha^j,n\}\big\}, \; j = 1,\ldots, 
\lceil \log n / \log \alpha\rceil\right\} \\ 
\;\cup\; \left\{
\sqrt{\frac{n}{|A|}} \one_A : A = \big\{\max\{t_i-\alpha^j,1\},\ldots,    
t_i\big\},\;  j = 1,\ldots, \lceil \log n / \log \alpha\rceil
\right\}. 
\end{multline*}
Here, the set $\cC_i$ has at most \smash{$2 \lceil 
  \log{n}/\log{\alpha} \rceil \leq 4 \log n/\log \alpha$} elements,
and by \eqref{eq:ab_norm}, balls of radius $r$ around elements
in $\cC_i$ cover the set $\cV_i$.  This establishes that 
\begin{equation}
\label{eq:vi_cover_bd}
N(r,\cV_i,\|\cdot\|_n) \leq \frac{-2\log n}{\log (1 - r^2/2)}.
\end{equation}

For a radius \smash{$0<r<\sqrt{2}$}, the covering number 
for \smash{$\cV = \cup_{i=0}^{s_0} \cV_i$} can be obtained
by just taking a union of the covers in \eqref{eq:vi_cover_bd} over
$i=0,\ldots,s_0$, giving 
\begin{equation} 
\label{eq:v_cover_bd}
N(r,\cV,\|\cdot\|_n) \leq \sum_{i=0}^{s_0} N(r,\cV_i,\|\cdot\|_n)  
\leq 2(s_0+1)\bigg(\frac{-\log n}{\log(1-r^2/2)} \bigg).
\end{equation}
Using \eqref{eq:ab_norm} once more, the diameter of the set $\cV$ 
is \smash{$\sqrt{2}$}, hence if \smash{$r\geq  1/\sqrt{2}$}, then    
we need only 1 ball to cover $\cV$.  Combining this fact with
\eqref{eq:v_cover_bd}, we obtain
\begin{equation} 
\label{eq:v_cover_bd_all_r}
N(r,\cV,\|\cdot\|_n) \leq 
\begin{cases}
\displaystyle
2(s_0+1)\bigg(\frac{-\log n}{\log(1-r^2/2)} \bigg)
& \text{if $0 < r <  1/\sqrt{2}$} \\  
1 & \text{if $r \geq  1/\sqrt{2}$} 
\end{cases}.
\end{equation}

Now let us apply \Fref{thm:van}, with $\cX=\cV$.  First, we remark 
that the quantity $t_0$ in \Fref{thm:van} may be taken to be
\smash{$t_0=1/\sqrt{2}$}. The bounds on $t$ in the theorem are $t > 
C_1$, as well as  
\begin{equation*}
t > C_2 \int_0^{1/\sqrt{2}} \sqrt{\log\left(
2(s_0+1)\frac{-\log{n}}{\log(1 - r^2/2)}\right)} \, dr. 
\end{equation*}
Next, we know that the right-hand side above is upper bounded by
\begin{multline*}
C_2 \int_0^{1/\sqrt{2}} \left[\sqrt{\log \big(2(s_0+1)\log{n}\big) }  
+  \sqrt{\log\left(\frac{-1}{\log(1 - r^2/2)}\right)} \; \right] \,
dr \\ =  C_2\sqrt{\frac{\log \big(2(s_0+1)\log{n}\big)}{2}} \, 
+ C_2 \sqrt{2} \int_0^{1/2} \sqrt{ \log \left(
\frac{1}{\log\big(\frac{1}{1-x^2}\big)}\right)} \, dx.
\end{multline*}
One can verify that the the integral in the second term above
converges to a finite constant (upper bounded by 1 in fact).
Thus the entire expression above is upper bounded by 
\smash{$C_2 \sqrt{\log{s_0}+\log\log{n}}$}, where the constant $C_2>0$ 
is adjusted to be larger, as needed. Therefore, letting
$c_I=\max\{C_1,C_2\}$, we may restrict our attention to 
\smash{$t > c_I \sqrt{\log{s_0}+\log\log{n}}$} in \Fref{thm:van},
and letting $C_I=C_4$, the conclusion reads, for $t=\gamma c_I$ and  
$\gamma > 1$,       
\begin{equation*}
\P \bigg( \sup_{a\in \cV} \;
\frac{\epsilon^\top a}{\sqrt{n}} > \gamma c_I
\sqrt{\log{s_0}+\log\log{n}} \bigg) \leq 2\exp\big(-C_I\gamma^2c_I^2
(\log{s_0}+\log\log{n}) \big).
\end{equation*}
Recalling the form of
\smash{$a=\sqrt{n/|A|} \one_A \in \cV$}, the above may be rephrased as  
\begin{multline}
\label{eq:v_bd}
\P \bigg(\frac{\sum_{j=t_i}^t \epsilon_j}{\sqrt{|t - t_i|}} 
> \gamma c_I \sqrt{\log{s_0}+\log\log{n}}, \;
\text{for $t=1,\ldots,n$, $i=0,\ldots,s_0$} \bigg) \\
\leq 2\exp\big(-C_I\gamma^2c_I^2
(\log{s_0}+\log\log{n}) \big).
\end{multline}

Finally, consider the following event
\begin{equation*}
\Omega_4 = \left\{ \frac{|F(t) - F(t_i)|}{\sqrt{|t - t_i|}} \leq
  \gamma c_I \sqrt{\log{s_0}+\log\log{n}}, 
\text{for $t=1,\ldots,n$, $i=0,\ldots,s_0$} \right\}.   
\end{equation*}
Recalling that \smash{$E(t) = \epsilon_{\lceil t \rceil}$} for all 
$t\in[0,1]$, we have \smash{$F(t) = \int_{0}^{t}E(u)\,du =
  \sum_{j=0}^{t} \epsilon_j$} for $t=1,\ldots,n$.
In \eqref{eq:v_bd},
we have thus shown \smash{$\P(\Omega_4) \geq
1- 2\exp(-C_I\gamma^2 c_I^2 (\log{s_0}+\log\log{n}))$}.  
Note that $|F(t) - F(t_i)|$ is piecewise linear with knots at
$t=1,\ldots,n$ and \smash{$\sqrt{|t - t_i|}$} is concave in between
these knots, so if \smash{$|F(t) - F(t_i)|/ \sqrt{|t - t_i|} \leq
  \gamma c_I \sqrt{\log{s_0}+\log\log{n}}$} for $t=1,\ldots,n$, then
the same bound must hold over all $t \in [0,n]$.  This shows that
$\Omega_4 \supseteq \Omega_3$, where $\Omega_3$ is the event in
question in the left-hand side of \eqref{eq:claim}; in other words, we
have verified \eqref{eq:claim}.  
\end{proof}

For the proof of \Fref{lem:GCresid}, we will need the following result
from \citet{van1990estimating}.

\begin{lemma}[\textbf{Lemma 3.5 of \citealt{van1990estimating}}]   
\label{lem:localentropy}
Assume the conditions in \Fref{thm:van}, and additionally, 
assume that for some $\zeta \in (0,1)$ and $K>0$,
\[
\cK(r) \leq Kr^{-2\zeta},
\]
where, recall, $\cK(r)$ is a continuous function
upper bounding the entropy number
\smash{$\log N(r,\cX,\|\cdot\|_n)$}. 
Then there exists constants $C_0,C_1$ (depending only on $M,\sigma$ in  
\eqref{eq:subgauss}) such that for any $t \geq C_0$,  
\[
\P\bigg( \sup_{x \in \mathcal{X}} \; \frac{|\epsilon^\top x|}
{\sqrt{n}\|x\|_n^{1-\zeta}} > t \sqrt{K} \bigg) \leq
 \exp (-C_1 t^2 K).
\]
\end{lemma}

\begin{proof}[Proof of Lemma \ref{lem:GCresid}]
Recall that for $i=0,\ldots,s_0$, we define $B_i = \{t_i+1, \ldots,  
t_{i+1}\}$. For $i=0,\ldots,s_0$, also define $n_i = |B_i|$, the 
scaled norm \smash{$\|\cdot\|_{n_i} = \|\cdot\|_2/\sqrt{n_i}$},
and
\[
\mathcal{X}_i = \Big\{w^{(i)} \in \R^{n_i} : 
(\one^{(i)})^\top w^{(i)} = 0, \; \| \op^{(i)} w^{(i)}  \|_1 \le 1, \;      
\|w^{(i)} \|_{n_i} \leq 1 \Big\}. 
\]
Here, we write \smash{$\one^{(i)} \in \R^{n_i}$} for the vector of all
1s, and \smash{$D^{(i)} \in \R^{(n_i-1) \times n}$} for the difference
operator, as in \eqref{eq:diff} but of smaller dimension. The set
$\cX_i$ is the discrete total variation space 
in \smash{$\R^{n_i}$}, where all elements are centered and have scaled  
norm at most 1.  From well-known results on entropy bounds for
total variation spaces (e.g., from Lemma 11 and Corollary
12 of \citet{wang2016trend}), we have
\begin{equation*}
\log N(r, \cX_i, \|\cdot\|_{n_i}) \leq \frac{C}{r},
\end{equation*}
for a universal constant $C>0$. Hence we may apply 
\Fref{lem:localentropy} with $\cX=\cX_i$ and $\zeta=1/2$: 
for the random variable
\begin{equation*} 
 M_i = \sup \left\{ \frac{|\epsilon_{B_i}^\top w^{(i)}|} 
{\sqrt{n_i}\|w^{(i)}\|_{n_i}^{1/2} } : w^{(i)} \in \mathcal{X}_i
\right\},  
\end{equation*}
we may take $t=\gamma C_0$ in the lemma, for any $\gamma > 1$, 
and conclude that
\begin{equation*}
\P\Big( M_i > \gamma C_0 \sqrt{C} \Big) \le    
\exp(-C_1 \gamma^2C_0^2 C).
\end{equation*}
Notice that we may rewrite $M_i$ as 
\begin{equation*}
M_i = \sup \left\{
\frac{|\epsilon_{B_i}^\top w^{(i)}|}{n_i^{1/4} \| \op^{(i)} w^{(i)} 
  \|_1^{1/2} \| w^{(i)} \|_2^{1/2}} : w^{(i)} \in \R^{n_i}, \;
(\one^{(i)})^\top w^{(i)} = 0 \right\}, 
\end{equation*}
and therefore 
\begin{equation*}
\P\bigg( \sup_{w^{(i)} \in \R^{n_i}, \,
(\one^{(i)})^\top w^{(i)}=0} \; 
\frac{|\epsilon_{B_i}^\top w^{(i)} |}
{\| \op^{(i)}  w^{(i)} \|_1^{1/2} \| w^{(i)} \|_2^{1/2}} > 
\gamma C_0 \sqrt{C} n_i^{1/4} \bigg) \le   
\exp(-C_1 \gamma^2C_0^2 C).
\end{equation*}
Using the union bound, 
\begin{equation*}
\P\left( \sup_{\substack{w^{(i)} \in \R^{n_i}, \,
(\one^{(i)})^\top w^{(i)}=0 \\ i=0,\ldots,s_0}} \; 
\frac{|\epsilon_{B_i}^\top w^{(i)} |}
{\| \op^{(i)}  w^{(i)} \|_1^{1/2} \| w^{(i)} \|_2^{1/2}} > 
\gamma C_0 \sqrt{C} n_i^{1/4} \right) \le   
(s_0+1)\exp(-C_1 \gamma^2C_0^2 C).
\end{equation*}
Define the constants \smash{$c_R=\max\{C_0\sqrt{C},1\}$} and
$C_R=\max\{C_1/2,1\}$. Then this ensures that we have 
\smash{$2C_R \gamma^2 c_R^2 \sqrt{s_0} \geq \log(s_0+1)$} for any 
$\gamma>1$ and any $s_0$, thus  
\begin{equation*}
\P\left( \sup_{\substack{w^{(i)} \in \R^{n_i}, \,
(\one^{(i)})^\top w^{(i)}=0 \\ i=0,\ldots,s_0}} \; 
\frac{|\epsilon_{B_i}^\top w^{(i)} |}
{\| \op^{(i)}  w^{(i)} \|_1^{1/2} \| w^{(i)} \|_2^{1/2}} > 
\gamma c_R (n_i s_0)^{1/4} \right) \leq \exp(-C_R \gamma^2 c_R^2 
\sqrt{s_0}).
\end{equation*}
The proof is completed by noting the following: if $w \in \cR^\perp$,
then \smash{$(\one^{(i)})^\top w_{B_i}=0$} for all $i=0,\ldots,s_0$,
and so on the event in consideration in the last display,   
\begin{align*}
|\epsilon^\top w| &\leq \sum_{i=0}^{s_0} 
|\epsilon_{B_i}^\top w_{B_i}| \le \gamma c_R s_0^{1/4}
\sum_{i=0}^{s_0} n_i^{1/4} \| \op^{(i)} w_{B_i} \|_1^{1/2} \|
w_{B_i}\|_2^{1/2} \\ 
&\le \gamma c_R s_0^{1/4} \left( \sum_{i=0}^{s_0} \| \op^{(i)}
  w_{B_i} \|_1 \right)^{1/2} \left( \sum_{i=0}^{s_0} n_i^{1/2}\|
  w_{B_i} \|_2 \right)^{1/2} \\   
&= \gamma c_R s_0^{1/4} \| \op_{-S_0} w \|_1^{1/2}  
\left( \sum_{i=0}^{s_0} n_i^{1/2} \| w_{B_i} \|_2 \right)^{1/2} \\ 
  &\le \gamma c_R s_0^{1/4} \| \op_{-S_0} w \|_1^{1/2} \left(
    \sum_{i=0}^{s_0} \| w_{B_i} \|_2^2 \right)^{1/4} \left(
    \sum_{i=0}^{s_0} n_i \right)^{1/4} \\
&= \gamma c_R s_0^{1/4} \| \op_{-S_0} w \|_1^{1/2} 
\| w\|_2^{1/2} n^{1/4}, 
\end{align*}
by two successive uses of Cauchy-Schwartz.
\end{proof}

\section{Proofs of Lemmas \ref{lem:how_big}, \ref{lem:how_small}}  
\label{app:how_big_small}

Both proofs follow from standard techniques in convex analysis.

\begin{proof}[Proof of \Fref{lem:how_big}]
We first consider the convex optimization problem 
\begin{equation}
\label{eq:how_big_min}
\min_{x \in \R^m} \; a^\top x \;\;\st\;\; \|x-c\|_2 \leq r,
\end{equation}
whose Lagrangian may be written as, for a dual variable $\lambda \geq
0$, 
\begin{equation*}
L(x,\lambda) = a^\top x + \lambda(\|x-c\|_2^2 - r^2).
\end{equation*}
The stationarity condition is
$a + \lambda(x-c) = 0$, thus $x = c - a/\lambda$.
By primal feasibility, $\|x-c\|_2 \leq r$, we see that we can take  
$\lambda = \|a\|_2/r$, which gives a solution $x=c -r a / \|a\|_2$.
The optimal value in \eqref{eq:how_big_max} is therefore 
$a^\top x = a^\top c -r \|a\|_2$.  By the same logic, the optimal
value of the convex problem
\begin{equation}
\label{eq:how_big_max}
\max_{x \in \R^m} \; a^\top x \;\;\st\;\; \|x-c\|_2 \leq r  
\end{equation}
is $a^\top c + r\|a\|_2$.  Now we can read off the optimal value of 
\eqref{eq:how_big} from those of \eqref{eq:how_big_min},
\eqref{eq:how_big_max}: its optimal value is 
\begin{equation*}
\max\big\{ -\big(a^\top c - r\|a\|_2^2\big),
\; a^\top c + r\|a\|_2 \big\} = |a^\top c| + r\|a\|_2^2, 
\end{equation*}
completing the proof.
\end{proof}

\begin{proof}[Proof of \Fref{lem:how_small}]
The proof is nearly immediate from the proof of Lemma
\ref{lem:how_big}, above. Notice that the optimal value of
\eqref{eq:how_small} is lower bounded by that of
\eqref{eq:how_big_min},    
which we already know is $a^\top c - r\|a\|_2^2$.  But when  
the latter is nonnegative, this is also the optimal value of
\eqref{eq:how_small}. Repeating the argument with $-a$ in place of $a$  
gives the result as stated in the lemma. 
\end{proof}

\section{Proof of \Fref{lem:local_max}}
\label{app:local_max}

To facilitate the proof, we define the concept of a 
{\it local maximum} among the absolute filter values:
a location $i$ is a local maximum if its absolute filter value
$|F_i(\ttheta)|$ is be greater than or equal
to the absolute values at  
neighboring locations, and strictly greater than at least one of these 
values (where the boundary points are treated as having just one
neighboring location). Specifically, a location $i$ must satisfy one
of the following conditions  
\begin{align}
&|F_{i-1}(\ttheta)| < |F_{i}(\ttheta)|, \;
|F_{i+1}(\ttheta)| \leq |F_{i}(\ttheta)|, \quad 
&\text{if $i \in \{b_n+1,\ldots,n-b_n-1\}$},
\label{eq:localmax1} \\
&|F_{i-1}(\ttheta)| \leq |F_{i}(\ttheta)|, \;
|F_{i+1}(\ttheta)| < |F_{i}(\ttheta)|, \quad 
&\text{if $i \in \{b_n+1,\ldots,n-b_n-1\}$},
 \label{eq:localmax2} \\
&|F_{i+1}(\ttheta)| < |F_{i}(\ttheta)| \quad 
&\text{if $i=b_n$}, \label{eq:localmax3} \\
&|F_{i-1}(\ttheta)| < |F_{i}(\ttheta)| \qquad 
&\text{if $i=n-b_n$}. \label{eq:localmax4}
\end{align}
Let \smash{$L(\ttheta)$} denote the set of local maximums derived from 
the filter with bandwidth $b_n$, i.e., the set of locations $i$
satisfying one of the four conditions
\eqref{eq:localmax1}--\eqref{eq:localmax4}.  

We first show that  
\smash{$L(\ttheta) \subseteq I_C(\tilde{\theta})$}. 
Fix \smash{$i \in L(\ttheta)$}. The
boundary cases, $i=b_n$ or $i=n-b_n$, are handled directly by 
the definition of \smash{$I_C(\ttheta)$}.  Hence, we may assume that
$i \in \{b_n+1,\ldots,n-b_n-1\}$, and
without a loss of generality, 
\begin{equation*}
|F_i(\ttheta)| > |F_{i-1}(\ttheta)| 
\quad \text{and} \quad 
|F_i(\ttheta)| \geq |F_{i+1}(\ttheta)|,
\end{equation*}
as well as \smash{$F_i(\ttheta)>0$}.  This means that  
\begin{equation*}
F_i(\ttheta) > |F_{i-1}(\ttheta)| 
\quad \text{and} \quad 
F_i(\ttheta) \geq |F_{i+1}(\ttheta)|,
\end{equation*}
which of course implies  
\begin{equation*}
F_i(\ttheta) > F_{i-1}(\ttheta) 
\quad \text{and} \quad
F_i(\ttheta) \geq F_{i+1}(\ttheta).
\end{equation*}
Applying the
definition of the filter in \eqref{eq:filter} gives
\begin{align*}
\bigg(\sum_{j=i+1}^{i+b_n}\ttheta_j -
\sum_{j=i-b_n+1}^{i}\ttheta_j\bigg) -
\bigg(\sum_{j=i}^{i+b_n-1}\ttheta_j -
\sum_{j=i-b_n}^{i-1}\ttheta_j\bigg) &> 0\\
\bigg(\sum_{j=i+1}^{i+b_n}\ttheta_j -
\sum_{j=i-b_n+1}^{i}\ttheta_j\bigg) -
\bigg(\sum_{j=i+2}^{i+b_n+1}\ttheta_j -
\sum_{j=i-b_n+2}^{i+1}\ttheta_j\bigg) &\geq 0,
\end{align*}
or, after simplification,
\[
\ttheta_{i+b_n} - 2\ttheta_{i} + \ttheta_{i-b_n} > 0 \quad
\text{and} \quad 
-\ttheta_{i+b_n+1} + 2\ttheta_{i+1} - \ttheta_{i-b_n+1} \geq 0.
\]
Adding the above two equations together, we get
\[
-\big(\ttheta_{i+b_n+1}-\ttheta_{i+b_n}\big) +
2\big(\ttheta_{i+1} - \ttheta_{i}\big)
- \big(\ttheta_{i-b_n+1}- \ttheta_{i-b_n}\big) > 0,
\]
which implies at least one of the three bracketed pairs of terms must
be nonzero, i.e., a changepoint must occur at one of the locations
$i$, $i+b_n$, or $i-b_n$.  The proves that 
\smash{$L(\ttheta) \subseteq I_C(\tilde{\theta})$}. 

Now we show the intended statement. Let $j \in
\{b_n,\ldots,n-b_n\}$, and \smash{$i \in L(\ttheta)$} be in the
direction of ascent from $j$ with respect to \smash{$F(\ttheta)$},
where $j \leq i$, without a loss of generality (for the case $i<j$, 
replace $\ell+b_n$ below by $\ell-b_n$). That is, the
location $i$ is a local maximum where   
\begin{equation} 
\label{eq:direction_of_localmax} 
|F_j(\ttheta)| \leq |F_{j+1}(\ttheta)| \leq \ldots \leq
|F_{i-1}(\ttheta)| \leq |F_i(\ttheta)|. 
\end{equation}
If $|i-j| \leq b_n$, then we have the desired result, due
to \eqref{eq:direction_of_localmax}. If $|i-j| > b_n$, then there must 
be at least one location \smash{$\ell  
  \in S(\ttheta)$} such that $|\ell-j|\leq b_n$. (To see this, note
that if \smash{$\ttheta_{j-b_n+1} = \ldots = \ttheta_{j+b_n}$}, then  
\smash{$F_j(\ttheta) = 0$}.)  Thus, at least one of $\ell,\ell+b_n$
lies in between $j$ and $i$, and then again
\eqref{eq:direction_of_localmax} implies the result, completing the  
proof.

\section{Proof of \Fref{lem:lower_bd_linear}}
\label{app:lower_bd_linear}

Our optimization problem may be rewritten as
\begin{equation*}
(\tilde{a},\tilde{b}) = \argmin_{a,b \in \R} \; 
b^2 + \sum_{x=1}^r \Big( (ax+b-a_1 x)^2 + (ax -b-a_2 x)^2 \Big). 
\end{equation*}
Taking a derivative of the criterion with respect to $b$ and setting
this equal to 0 gives 
\begin{equation*}
0 = b + \sum_{x=1}^r ( ax+b-a_1x - ax+ b + a_2x),
\end{equation*}
i.e., we see that the optimal value is
\begin{equation*}
\tilde{b} = (a_1-a_2) \frac{\sum_{x=1}^r x}{2r+1} =  
(a_1-a_2) \frac{r(r+1)}{2(2r+1)}.
\end{equation*}
Taking a derivative of the criterion with respect to $a$ and setting
this equal to 0 gives
\begin{equation*}
0 = b + \sum_{x=1}^r ( ax+b-a_1x + ax-b - a_2x),
\end{equation*}
i.e., we see that the optimal value is
\begin{equation*}
\tilde{a} = \frac{a_1+a_2}{2}.
\end{equation*}
Plugging in \smash{$\tilde{a},\tilde{b}$} into the criterion, and
abbreviating $c_r=r(r+1)/(2(2r+1))$, we can compute the optimal
criterion value: 
\begin{multline*}
(a_1-a_2)^2 c_r^2 + \sum_{x=1}^r \bigg[\bigg( \frac{a_2-a_1}{2} x +
(a_2-a_1) c_r \bigg)^2 + \bigg( \frac{a_1-a_2}{2} x - (a_1 - a_2) c_r
\bigg)^2 \bigg] \\
\begin{aligned}
&= (a_1-a_2)^2 c_r^2 + \frac{1}{2}(a_1-a_2)^2 \sum_{x=1}^r x^2 + 
2r(a_1-a_2)^2 c_r^2 + 2(a_1-a_2)^2 c_r \sum_{x=1}^r x \\
&= (a_1-a_2)^2 \bigg( \frac{r^2(r+1)^2}{4(2r+1)^2}
+ \frac{r(r+1)(2r+1)}{12} 
+ \frac{r^3(r+1)^2}{2(2r+1)^2}
+ \frac{r^2(r+1)^2}{2(2r+1)} \bigg) \\
&\geq (a_1-a_2)^2 \bigg( \frac{r^2}{16} 
+ \frac{r(r+1)(2r+1)}{12}
+ \frac{r^3}{8} 
+ \frac{r^2(r+1)}{4} \bigg) \\
&\geq (a_1-a_2)^2 r^3
\bigg( \frac{1}{6}+\frac{1}{8}+\frac{1}{4} \bigg) \\
&= (a_1-a_2)^2 \frac{13 r^3}{24}.
\end{aligned}
\end{multline*}
\hfill\qedsymbol
\end{document}